\documentclass[showpacs,amsmath,amssymb,twocolumn,aps,pra,superscriptaddress,notitlepage,floatfix,10pt]{revtex4-2}

\usepackage{stylesetting}
\usepackage{orcidlink}
\usepackage[normalem]{ulem}

\makeatletter

\newcommand{\appendixtableofcontents}{%
  \section*{Contents}%
  \@starttoc{atoc}%
}

\newcommand{\startappendixtoc}{%
  \let\originaladdcontentsline\addcontentsline
  \renewcommand{\addcontentsline}[3]{%
    \def\firstarg{##1}%
    \def\tocname{toc}%
    \ifx\firstarg\tocname
      \originaladdcontentsline{atoc}{##2}{##3}%
    \else
      \originaladdcontentsline{##1}{##2}{##3}%
    \fi
  }%
}

\makeatother

\newcommand{\FudanEMW}{Key Laboratory for Information Science of Electromagnetic Waves (Ministry of Education), Fudan University, Shanghai 200433, China}
\newcommand{\PKUCs}{Center on Frontiers of Computing Studies, School of Computer Science, Peking University, Beijing 100871, China}

 \date{\today}

\begin{document}

\title{Ultra-Precise Quantum Projective Designs in Constant Depth}

\author{Qingyue Zhang~\orcidlink{0009-0000-5638-9746}}
\affiliation{\FudanEMW}

\author{Junjie Chen~\orcidlink{0009-0004-0401-1257}}
\affiliation{Center for Quantum Information, Institute for Interdisciplinary Information Sciences, Tsinghua University, Beijing, 100084 China}

\author{Zhou You~\orcidlink{0000-0002-6140-2092}}
\affiliation{\PKUCs}

\author{You Zhou~\orcidlink{0000-0003-0886-077X}}
\email{you\_zhou@fudan.edu.cn}
\affiliation{\FudanEMW}

\date{\today}

\begin{abstract}
Random quantum objects are powerful resources for quantum information processing, yet exact Haar randomness is costly and typically unnecessary. We introduce an explicit sparse commuting circuit ensemble on $n$ qubits that reproduces low-order Haar moments in the stringent relative-error sense. The circuit consists of a sparse Clifford phase layer followed by independent single-qubit Clifford gates. Acting on a simple product state, the resulting ensemble forms $\epsilon$-approximate projective $2$- and $3$-designs in relative error, with the required logarithmic interaction degree being asymptotically optimal within this circuit family. It admits an ancilla-free implementation of quantum depth $O(\log(n/\epsilon))$ on an all-to-all architecture, as well as an adaptive constant-depth implementation---in fact, depth seven---using $O(n\log(n/\epsilon))$ ancilla qubits. Departing from existing shallow-design paradigms, our analysis exploits the intrinsic moment structure of commuting phase circuits; at third order, this requires a new block decomposition and combinatorial analysis that also suggests a route toward higher-order shallow designs. Our results show that precise Haar-like statistics can emerge from sparse commuting dynamics with remarkably low quantum resources, with applications to randomized characterization, quantum metrology, quantum algorithms, and many-body physics.
\end{abstract}
\maketitle

\emph{Introduction.}---
Random quantum states and operations are fundamental resources in quantum information science~\cite{mele2024introduction}, with applications ranging from quantum tomography~\cite{gebhart2023learning}, quantum benchmarking~\cite{helsen2022general,eisert2020quantum,kliesch2021theory} and randomized measurements~\cite{huang2020predicting,elben2023randomized} to quantum advantage~\cite{hangleiter2023computational}, quantum algorithms~\cite{cerezo2021cost,cerezo2021variational}, and quantum many-body dynamics~\cite{hayden2007black}.  
Exact Haar randomness is typically unnecessary; quantum designs reproduce its relevant low-order moments at lower cost.
A central challenge in the field is to minimize the quantum resources required to generate such quantum randomness. 
The recent frameworks~\cite{schuster2024random,laracuente2026approximate} have shown that finite-order Haar statistics can already emerge at logarithmic depth, well before full scrambling~\cite{brandao2016local,haferkamp2022random,brandao2021models}.  Subsequent studies have further explored all-to-all connectivity~\cite{lee2026shallow}, non-Clifford resources~\cite{leone2026non,bittel2026adaptively}, and adaptive circuit models~\cite{foxman2025random}.

Pushing the depth below logarithmic scale, however, has so far come with weaker notions of randomness: known sublogarithmic~\cite{cui2025unitary} or constant-depth constructions~\cite{foxman2025random} generally provide additive or measurable-error guarantees. Relative error formalizes an ultra-precise notion of design accuracy: it requires multiplicative agreement with Haar even on exponentially small moment sectors, and therefore automatically implies weaker additive and measurable-error guarantees. Such fine-grained control is particularly important for moment inversion and variance bounds, as in classical shadows~\cite{hu2023classical,bertoni2024shallow,king2025triply,zhang2025robust} 
and quantum metrology~\cite{zhou2026randomized,du2026complexity,mao2026near}. 
\emph{A key open question is therefore whether constant-depth architectures can realize quantum designs with vanishing relative error.}
\begin{figure}[h]
\centering \includegraphics[width=\linewidth]{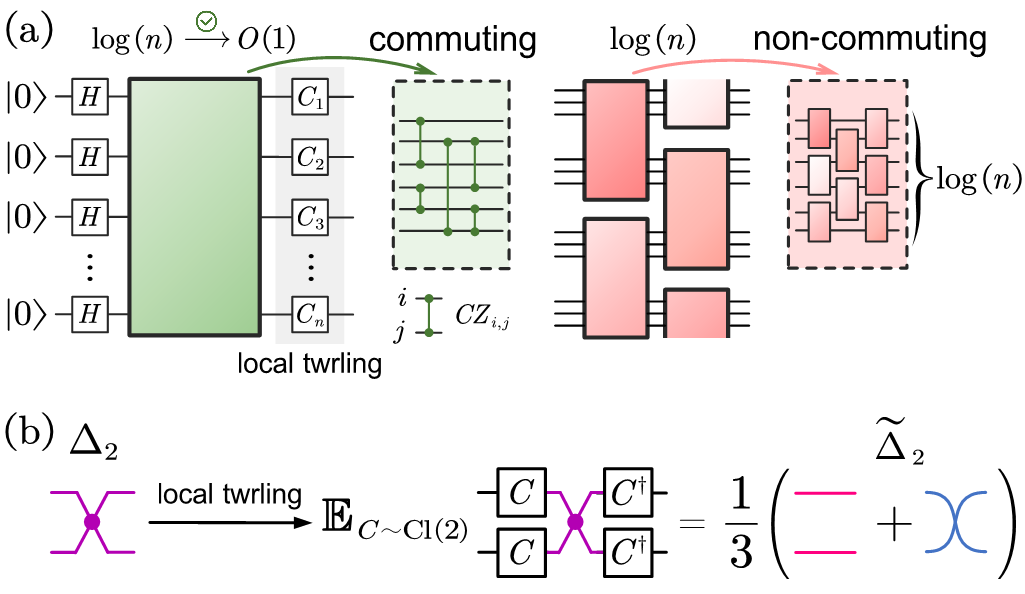}  
    \caption{Schematic illustration of random circuit ensemble.
(a) Structural contrast between our shallow phase-circuit construction (left) and generic shallow random circuits based on "gluing" non-commuting $O(\log n)$-qubit patches \cite{schuster2024random} or non-commuting local bricks~\cite{haferkamp2022random} (right).
(b) Our moment analysis makes the deviation from Haar transparent. For example, at second order, the relative-error obstruction is carried by the residual term $\Delta_2^{\otimes n}$. Local Clifford twirling maps $\Delta_2$ to $\widetilde{\Delta}_2$, exponentially suppressing this contribution, as quantified around \cref{eq:twirled-delta2}.
}
\label{fig:battlehuang}
\end{figure}

Here we resolve this question affirmatively with a simple commuting architecture, which also admits a constant-depth adaptive realization. We consider a sparse Clifford phase circuit followed by independent single-qubit Clifford gates, with each controlled-$Z$ ($CZ$) edge included independently with probability $p=\gamma\log n/n$.
We prove that the resulting ensemble forms approximate projective $2$- and $3$-designs with vanishing relative error, and further show that the resulting logarithmic-degree sparsity threshold is asymptotically optimal within this circuit family.

\Cref{fig:battlehuang}(a) summarizes our circuit architecture and its key properties. Unlike generic brickwork random circuits~\cite{haferkamp2022random,dalzell2022random,schuster2024random}, where randomness is built up through repeated layers of generally non-commuting gates, our circuit structure is built almost entirely from mutually commuting diagonal circuits. Remarkably, we show that repeated non-commuting dynamics are unnecessary for generating precise Haar-like statistics. The moment deviation of commuting phase circuits from Haar is confined to transparent terms. A single layer of independent local Clifford gates suffices to suppress these terms and achieve vanishing relative error, even at the optimal logarithmic interaction degree.
Establishing this persistence at third order is substantially more involved, requiring a new block decomposition and combinatorial analysis beyond the simultaneous diagonalization available at second order.

In our construction, commutativity is not an obstacle to Haar-like randomness but a resource for realizing it in a highly parallel form.
Using adaptive measurements and classical feedforward~\cite{buhrman2024state,zi2025constant,liu2025state}, the $CZ$ layer can be compressed to $O(1)$ quantum depth with $O(n\log(n/\epsilon))$ ancillas, while an ancilla-free implementation on all-to-all architectures has depth $O(\log(n/\epsilon))$. Importantly, the use of measurements and feedforward places our construction beyond the nonadaptive circuit setting underlying earlier depth barriers for shallow designs \cite{schuster2024random,cui2025unitary}.

\emph{Problem setup and main result.}---
We mainly consider the following ensemble of random unitaries on an $n$-qubit system, denoted by the \emph{shallow phase ensemble} $\mc E_{\mathrm{shallow}}^{(\gamma)}$. 
Each unitary has the form $U=U_{\mathrm{loc}}U_SU_{CZ}H^{\otimes n}$ in \cref{fig:battlehuang}(a): a fixed Hadamard layer, followed by a sparse diagonal Clifford layer generated by phase gates $S=\operatorname{diag}(1,i)$ and controlled-$Z$ gates $CZ=\operatorname{diag}(1,1,1,-1)$, and finally a local Clifford layer. Here $U_{\mathrm{loc}}=\bigotimes_{q\in[n]}U_q$, with each $U_q$ sampled independently from the single-qubit Clifford group $\mathrm{Cl}(2)$. 

The diagonal layer is $U_S=\prod_{q\in[n]}S_q^{A_{q,q}}$ and $U_{CZ}=\prod_{i<j}CZ_{i,j}^{A_{i,j}}$, where $A$ is a random symmetric matrix specifying both local phases and the $CZ$ interaction graph. The diagonal entries $A_{q,q}$ are independent and uniform in $\{0,1,2,3\}$, while the off-diagonal entries are independent Bernoulli variables with $p:=\Pr(A_{i,j}=1)=\gamma\log n/n$. Thus, $\gamma$ controls the entangling sparsity, with expected degree $\gamma\log n$. Unless stated otherwise, $\gamma>0$ is constant sparsity parameter independent of $n$.
All gates in $U_SU_{CZ}$ are diagonal and mutually commuting, reflecting the IQP structure~\cite{bremner2011classical}. This structure is also naturally compatible with platforms supporting long-range entangling interactions, such as trapped-ion and Rydberg-atom architectures~\cite{liu2025certified,bluvstein2026faulttolerant,evered2023high}. Although the single-qubit phase layer $U_S$ could be absorbed into $U_{\mathrm{loc}}$ without changing the ensemble distribution, we keep it explicit to expose the phase structure used in the moment analysis.

Moreover, we use three reference ensembles for comparison. The dense phase ensemble
\(\mc E_{\mathrm{phase}}\) has the same circuit structure as $\mc E_{\mathrm{shallow}}^{(\gamma)}$, including
\(U_{\mathrm{loc}}\), but with each \(CZ\) edge chosen independently with
balanced probability \(p=1/2\). The bare phase ensemble
\(\mc E_{\mathrm{bare}}\) is obtained from
\(\mc E_{\mathrm{phase}}\) by removing \(U_{\mathrm{loc}}\). It is used only to isolate the obstruction present before local Clifford twirling. Finally, \(\mc E_{\mathrm{Haar}}\) denotes the ideal Haar ensemble. 

For any unitary ensemble $\mc E$, sampling $U\sim\mc E$ and applying it to the fixed reference state $\ket{\mathbf 0}=\ket{0}^{\otimes n}$ induces the pure-state ensemble $\{U\ket{\mathbf 0}\}_{U\sim\mc E}$, whose $k$-th moment is $\mb M_{\mc E}^{(k)}:=\mathbb E_{U\sim\mc E}(U\ket{\mathbf 0}\!\bra{\mathbf 0}U^\dagger)^{\otimes k}$. For convenience, we use the same notation for a unitary ensemble and its induced state ensemble. In our construction, the latter can be viewed as sparse random graph states~\cite{hein2004multiparty,ghosh2025random,zhou2022entanglement} 
followed by independent local Clifford twirling. For the Haar ensemble, $\mb M_{\mc E_{\mathrm{Haar}}}^{(k)}=\Pi_{\mathrm{sym}}^{(k)}/r_{\mathrm{sym}}^{(k)}$, where $\Pi_{\mathrm{sym}}^{(k)}:=\frac{1}{k!}\sum_{\pi\in S_k}V_n(\pi)$ projects onto the fully symmetric subspace~\cite{mele2024introduction}. Here $V_n(\pi)$ denotes the unitary operator that permutes the $k$ copies of the $n$-qubit Hilbert space according to $\pi$, and factorizes as $V_n(\pi)=V_1(\pi)^{\otimes n}$. The symmetric-subspace dimension is $r_{\mathrm{sym}}^{(k)}=\binom{D+k-1}{k}$, with $D=2^n$.

\comments{For each unitary ensemble \(\mc E\) introduced above, applying a sampled unitary to the reference state \(\ket{\mathbf 0}\) induces the corresponding projective state ensemble. Its \(k\)-th state moment operator is defined as
$\mb M_{\mc E}^{(k)}
:=
\mathbb E_{U\sim \mc E}
\Bigl(U\ket{\mathbf 0}\!\bra{\mathbf 0}U^\dagger\Bigr)^{\otimes k}.$
For example, the \(k\)-th moment operator of the Haar ensemble is $\mb M_{\mc E_{\mathrm{Haar}}}^{(k)}
=
\frac{1}{r_{\mathrm{sym}}^{(k)}}
\Pi_{\mathrm{sym}}^{(k)}$~\cite{mele2024introduction}. Here, $\Pi_{\mathrm{sym}}^{(k)}:= \frac{1}{k!}\sum_{\pi\in S_k}V_n(\pi)$ denotes the projector onto the fully symmetric subspace of the \(k\)-copy
\(n\)-qubit Hilbert space, where \(V_n(\pi)\) permutes the \(k\) copies
according to \(\pi\in S_k\). Its rank is $r_{\mathrm{sym}}^{(k)}
:=\binom{D+k-1}{k},$ with $D:=2^n.$}

We measure the deviation from Haar randomness in relative error~\cite{mele2024introduction}: $\mc E$ is an $\epsilon$-approximate projective state $k$-design if
\begin{equation}
(1-\epsilon)\,\mb M_{\mc E_{\mathrm{Haar}}}^{(k)}
\preceq
\mb M_{\mc E}^{(k)}
\preceq
(1+\epsilon)\,\mb M_{\mc E_{\mathrm{Haar}}}^{(k)} .
\label{eq:rel-design-def}
\end{equation}
Relative error provides a stringent multiplicative notion of randomness and guarantees indistinguishability from Haar under arbitrary $k$-query protocols, including adaptive ones. It directly implies additive-error bounds, whereas converting an additive bound into relative error incurs an exponentially large prefactor \cite{brandao2016local}. Our main result establishes relative-error state $2$- and $3$-designs from shallow phase circuits with a single local Clifford layer, summarized below.

\begin{theorem}[Informal: Relative-error design from shallow phase circuits]
\label{thm:main-informal}
For every target accuracy $\epsilon\in(0,1)$, the $n$-qubit shallow phase ensemble $\mc E_{\mathrm{shallow}}^{(\gamma)}$ forms $\epsilon$-approximate projective $2$- and $3$-designs in relative error.

With overwhelmingly high probability, the ensemble can be implemented either
\begin{itemize}
\item[(i)] in $O(1)$ quantum depth using $O\!\left(n\log\frac{n}{\epsilon}\right)$ ancilla qubits;
\item[(ii)] without ancillas in depth $O\!\left(\log\frac{n}{\epsilon}\right)$.
\end{itemize}
\end{theorem}
~\cref{thm:main-informal} provides a qualitatively different route to ultra-precise shallow designs. The commuting structure of the circuits in $\mc E_{\mathrm{shallow}}^{(\gamma)}$ enables an adaptive implementation with measurements and classical feedforward in quantum depth seven. The construction in Theorem~\ref{thm:main-informal} includes each $CZ$ interaction independently with probability $p=\gamma\log n/n$. We next show that this selection of $p$ is asymptotically optimal within the phase circuit family.
\begin{prop}[Informal: Optimality of logarithmic sparsity]
\label{prop:log-sparsity-optimal-main}
It requires $p=\Omega\left(\log n/n\right)$ to form $\epsilon$-approximate projective designs in relative error within the phase circuit family.
\end{prop}
Throughout the main text, we focus on the main results and underlying mechanisms, while complete proofs and further technical details are collected in the Supplemental Material. 
We next uncover the moment mechanisms underlying the construction, beginning with the second moment.

\emph{Second moment: removing the phase circuit obstruction.}---
The second moment already captures the central mechanism of our construction. Even in the dense limit, the bare phase ensemble $\mc E_{\mathrm{bare}}$ retains a constant relative-error obstruction. Local Clifford twirling turns this obstruction into an exponential contraction, and we further show that the second-order relative error vanishes when the $CZ$ layer is sparsified to logarithmic expected degree.
We first present a general formulation of relative error that will be used throughout the analysis. Since both the phase ensemble and Haar state moments are supported on the symmetric subspace, the relative-error condition in~\cref{eq:rel-design-def} is equivalent to a renormalized operator-norm bound. 
We quantify the relative error of ensemble $\mc E$ at order $k$ by
\begin{equation}
\epsilon_k(\mc E)
=
r_{\mathrm{sym}}^{(k)}
\left\|
\mb M_{\mc E}^{(k)}
-
\mb M_{\mc E_{\mathrm{Haar}}}^{(k)}
\right\|_\infty .
\label{lem:relative-error-opnorm:main}
\end{equation}
In particular, for constant $k$, one has $r_{\mathrm{sym}}^{(k)}=O(D^k)$.
This characterization reduces the relative-error analysis at any fixed order to controlling the operator-norm distance from the Haar moment, with $r_{\mathrm{sym}}^{(k)}$ setting the intrinsic Haar scale. 

We first use it to expose the obstruction of the bare phase ensemble. Consider the dense phase circuit without the final local Clifford twirl, where each $CZ$ edge is sampled independently with probability $p=1/2$. According to Refs.~\cite{zhang2025robust,nechita2021graphical}, the second state moment is $\mb M_{\mathrm{bare}}^{(2)}
=
D^{-2}
\left(
\id_4^{\otimes n}
+
\sw_2^{\otimes n}
-
\Delta_2^{\otimes n}
\right),$
where $\sw_2$ is the two-copy swap operator on one qubit and $\Delta_2:=\ket{0,0}\!\bra{0,0}+\ket{1,1}\!\bra{1,1}$. By contrast, the second moment operator of the Haar ensemble is $\mb M_{\mc E_{\mathrm{Haar}}}^{(2)}
=
\frac{1}{D(D+1)}
\left(
\id_4^{\otimes n}
+
\sw_2^{\otimes n}
\right)$~\cite{mele2024introduction}.

\begin{prop}[Constant relative-error obstruction of the bare phase ensemble]
\label{prop:bare-phase-obstruction}
The bare phase ensemble $\mc E_{\mathrm{bare}}$ remains a constant relative-error distance from a projective $2$-design. More precisely, $\epsilon_2(\mc E_{\mathrm{bare}})\ge 1/2-o(1)$ as $n\to\infty$.
\end{prop}

The origin of this obstruction is transparent from the residual term $-D^{-2}\Delta_2^{\otimes n}$ in $\mb M_{\mathrm{bare}}^{(2)}$. In additive norm this correction is exponentially small, of order $D^{-2}$, but the Haar second moment itself has eigenvalue $1/r_{\mathrm{sym}}^{(2)}=2/[D(D+1)]=\Theta(D^{-2})$. Thus, once measured relative to the intrinsic Haar scale, the residual contribution produces an order-one distortion rather than a vanishing error. Importantly, this obstruction already persists in the fully dense $CZ$ circuit, so increasing the interaction density alone cannot remove it. Consequently, the bare phase ensemble cannot approach a relative-error state $2$-design, and therefore cannot approach any higher-order relative-error design either.

Surprisingly, this obstruction can be overcome by applying a single layer of local Clifford twirling. As shown in~\cref{fig:battlehuang}(b), the local diagonal projector $\Delta_2$ under this twirling is replaced by
\begin{equation}
\widetilde{\Delta}_2
:=
\mathbb E_{C\sim \mathrm{Cl}(2)}
C^{\otimes 2}\Delta_2 C^{\dagger\otimes 2}
=
\frac{\id_4+\sw_2}{3}.
\label{eq:twirled-delta2}
\end{equation}
Both \(\id_4\) and \(\sw_2\) are invariant under local Clifford twirling, so the twirl modifies only the term $D^{-2}\Delta_2$ responsible for the obstruction. While \(\|\Delta_2^{\otimes n}\|_\infty=1\),~\cref{eq:twirled-delta2} gives \(\|\widetilde{\Delta}_2^{\otimes n}\|_\infty=\|\widetilde{\Delta}_2\|_\infty^n=(2/3)^n\). The unwanted
contribution is therefore exponentially suppressed. This leads to an exponentially small second-moment relative error for the dense phase ensemble.

\begin{prop}[Dense phase ensemble at second order]
\label{prop:k2-phase-story:main}
The dense phase ensemble \(\mc E_{\mathrm{phase}}\) satisfies
\(\epsilon_2(\mc E_{\mathrm{phase}})\le (2/3)^n\).
\end{prop}

We next compare $\mc E_{\mathrm{phase}}$ with its shallow counterpart, which is the most challenging part of the proof. The second moment is organized by ordered partitions $I_0\sqcup I_1\sqcup I_2=[n]$, where  $`\sqcup'$ denotes disjoint union.  
The three sets specify the sites assigned to the local two-copy operators $\widetilde{\Delta}_2$, $\id_4-\widetilde{\Delta}_2$, and $\sw_2-\widetilde{\Delta}_2$, respectively. We refer to these local operators as blocks. For a local block $A$, we write $A^{\otimes I}:=\bigotimes_{q\in I}A_q$, so that, for example, $\widetilde{\Delta}_2^{\otimes I_0}$ acts on the two-copy local spaces of the qubits in $I_0$. Define the corresponding global block by $ Q_{I_0,I_1,I_2}
:=
\widetilde{\Delta}_2^{\otimes I_0}
\otimes
(\id_4-\widetilde{\Delta}_2)^{\otimes I_1}
\otimes
(\sw_2-\widetilde{\Delta}_2)^{\otimes I_2}$.  
In this way, the second moment reads
\begin{equation}
\mb M_{\mc E_{\mathrm{shallow}}^{(\gamma)}}^{(2)}
=
D^{-2}
\sum_{I_0\sqcup I_1\sqcup I_2=[n]}
\left(1-\frac{2\gamma \log n}{n}\right)^{|I_1||I_2|} Q_{I_0,I_1,I_2}
,
\label{eq:twirled-sparse-second-moment}
\end{equation}
Here, each potential \(CZ\) edge between \(I_1\) and \(I_2\) contributes a factor \(1-2p\) under Bernoulli averaging, producing the total damping factor \((1-2p)^{|I_1||I_2|}\). 
In the dense phase ensemble, \(p=1/2\), so every configuration with \(|I_1||I_2|>0\) is eliminated, and ~\cref{eq:twirled-sparse-second-moment} reduces to the dense second-moment formula.

~\cref{eq:twirled-sparse-second-moment} also reveals the algebraic structure that makes the second-moment comparison tractable. The three local blocks $\widetilde{\Delta}_2$, $\id_4-\widetilde{\Delta}_2$, and $\sw_2-\widetilde{\Delta}_2$ commute with each other.  Hence, the shallow, dense-phase, and Haar second-moment operators are simultaneously diagonal. We can therefore evaluate their eigenvalues explicitly on each basis, reducing the operator-norm comparison to a scalar problem, yielding the sparse-to-dense estimate:
\begin{equation}
\left\|
\mb M_{\mc E_{\mathrm{shallow}}^{(\gamma)}}^{(2)}
-
\mb M_{\mc E_{\mathrm{phase}}}^{(2)}
\right\|_\infty
\le
D^{-2}\left(2n^{1-\gamma/4}+\mathrm{negl}(n)\right).
\label{eq:k2-sparse-to-dense}
\end{equation}
Throughout this work, $\mathrm{negl}(n)$ denotes a negligible function, i.e., $\mathrm{negl}(n)=o(n^{-c})$ for every constant $c>0$. Combining this estimate with the dense-to-Haar bound in~\cref{prop:k2-phase-story:main} and the normalized operator-norm characterization in~\cref{lem:relative-error-opnorm:main} yields the second-order relative-error result below.

\begin{theorem}[Second-order relative error]
\label{thm:second-moment}
For \(\gamma>4\), the shallow phase ensemble satisfies
\begin{equation}
\epsilon_2\!\left(\mc E_{\mathrm{shallow}}^{(\gamma)}\right)
\le
2n^{1-\gamma/4}
+
\mathrm{negl}(n).
\label{eq:second-moment-bound}
\end{equation}
In particular, for any constant \(\gamma>4\), the second-order relative error vanishes as \(n\to\infty\).
\end{theorem}

Here the exponentially small dense-to-Haar contribution $(2/3)^n$ from~\cref{prop:k2-phase-story:main} has been absorbed into $\mathrm{negl}(n)$, so the dominant correction arises only from sparsification. We therefore show that local Clifford twirling removes the constant relative-error obstruction of the bare phase circuit, while logarithmically sparse $CZ$ interactions still yield vanishing second-order relative error. This relative-error guarantee is, in particular, sufficient to imply anti-concentration, a key ingredient underlying quantum sampling advantage~\cite{hangleiter2023computational,dalzell2022random}. Thus, asymptotically exact second-order Haar statistics already emerge from a genuinely sparse commuting circuit.

A complementary perspective is provided by Ref.~\cite{heinrich2025anti}, which shows that sufficiently accurate anti-concentration can imply a relative-error state $2$-design for suitable locally invariant ensembles. This route requires quantitative control of the anti-concentration error. Our explicit second-moment expansion provides such control for sparse phase circuits, both with and without local Clifford twirling, and at the same time makes transparent how local twirling changes the resulting relative-error accuracy. This may offer further insight into the relation between anti-concentration and relative-error randomness, while our moment-based approach also extends naturally to the third-order analysis below.

\emph{Third moment: beyond simultaneous diagonalization.}---
Moving from second to third order introduces qualitatively new structure. Local Clifford twirling remains effective, but the second-order strategy no longer extends directly: relative-error control at third order requires a substantially different analysis.

As in the $k=2$ expansion of~\cref{eq:twirled-sparse-second-moment}, phase averaging decomposes the third moment into tensor products of local operator blocks weighted by interaction-dependent decay coefficients. The richer permutation structure of $S_{3}$ now gives ten local block types: $\Delta_{3}:=\ket{0,0,0}\bra{0,0,0}+\ket{1,1,1}\bra{1,1,1}$ and nine blocks $B(\pi_{1},\pi_{2})$, indexed by odd/even permutation pairs. To define them, we use computational-basis support intersections, which arise naturally under diagonal phase twirling~\cite{zhang2025robust}. At second order, for example, $\Delta_{2}$ is precisely the support intersection of the identity and swap operators, $\Delta_{2}=\id_4\cap\sw_2$. Analogously, we define $B(\pi_1,\pi_2):=V_1(\pi_1)\cap V_1(\pi_2)-\Delta_3$, and show that $\Delta_3$ together with the nine $B$-blocks partitions the local support generated by local phase averaging. Unlike the three second-order blocks, however, these third-order blocks are not simultaneously diagonalizable, and some are non-Hermitian. This loss of simultaneous diagonalization is the first major obstacle in the third-order analysis. Their explicit definitions and support structure are given in the End Matter.

As in~\cref{eq:twirled-sparse-second-moment}, an $n$-qubit block is specified by the position data $\mathfrak I:=\bigl(I^{(0)},\{I^{(\pi_{1},\pi_{2})}\}_{\pi_{1},\pi_{2}}\bigr)$, where these sets partition $[n]$, with corresponding operator $B_{\mathfrak I}:=\Delta_{3}^{\otimes I^{(0)}}\otimes\bigotimes_{\pi_{1},\pi_{2}}B(\pi_{1},\pi_{2})^{\otimes I^{(\pi_{1},\pi_{2})}}$. Its decaying coefficient is $(1-2p)^{N_{\mathfrak I}}$, where $N_{\mathfrak I}$ counts pairs of sites whose block labels differ in both permutation indices. Denoting the locally twirled block by $\widetilde{B}_{\mathfrak I}\coloneq \mbb{E}_{C\sim\mathrm{Cl}(2)^{\otimes n}} C^{\otimes 3}B_{\mathfrak I}C^{\dagger,\otimes 3}$, the third moment takes the compact form
\begin{equation}\label{eq:shallow3-main}
\mb M_{\mc E_{\mathrm{shallow}}^{(\gamma)}}^{(3)}
=
D^{-3}\sum_{\mathfrak I}
\left(1-\frac{2\gamma\log n}{n}\right)^{N_{\mathfrak I}}
\widetilde{B}_{\mathfrak I}.
\end{equation}
The precise block definitions and the derivation are deferred to the End Matter.

This richer block structure underlies our main third-order result, whose proof goes fundamentally beyond the simultaneous-diagonalization strategy available at second order. A more detailed proof sketch is provided in the End Matter; here we briefly outline the key ingredients.
\begin{theorem}[Third-order relative error]
\label{thm:third-moment}
There exist positive constants $A_{0},B_{0},\gamma_{0}$, such that the shallow phase ensemble satisfies
\begin{equation}
\epsilon_{3}\!\left(\mc E_{\mathrm{shallow}}^{(\gamma)}\right)
\le
n^{A_{0}-B_{0}\gamma}
+
\mathrm{negl}(n).
\label{eq:k3-relative}
\end{equation}
In particular, for any constant $\gamma>\gamma_{0}$, the third-order relative error vanishes as $n\to\infty$.
\end{theorem}
Similar to the second-order case, we use $\mc E_{\mathrm{phase}}$ as an intermediate ensemble. In the dense limit, all configurations with $N_{\mathfrak I}>0$ vanish, while the remaining non-Haar terms are exponentially suppressed by the local contraction induced by Clifford twirling, as in~\cref{eq:twirled-delta2}. This yields an exponentially small dense-to-Haar error, analogous to~\cref{prop:k2-phase-story:main}. 

The sparse-to-dense comparison presents two new difficulties. First, the simultaneous diagonalization available at $k=2$ no longer holds. Instead, all locally twirled blocks preserve a common decomposition into $S_{3}$ irreducible sectors,
\(
(\mathbb C^2)^{\otimes3}
=
\mc H_{\mathrm{sym}}
\oplus
\mc H_{\mathrm{std}}^{(0)}
\oplus
\mc H_{\mathrm{std}}^{(1)}.
\)
This confines the non-commuting structure to the last two sectors $\mc H_{\mathrm{std}}^{(0,1)}$, although some local blocks remain non-Hermitian. A global symmetrization then pairs such block configurations with their adjoints, restoring Hermiticity at the grouped level and providing the local foundation for the third-order norm analysis.

Second, even after the local operator structure is controlled, the number of block types grows from three to ten, leading to as many as $10^n$ position assignments. To tame this combinatorial growth, we organize the nine nontrivial blocks $B(\pi_{1},\pi_{2})$ into a $3\times3$ chessboard, with rows and columns indexed by odd and even permutations, respectively. This `bishop method' \footnote{We refer to this organization as the `bishop method' because \(N_{\mathfrak I}\) is obtained by summing products of nontrivial blocks over pairs lying in different rows and different columns, which may be viewed as oblique connections reminiscent of a bishop's diagonal moves on a chessboard. This representation makes both the interaction exponent \(N_{\mathfrak I}\) and the associated multiplicities transparent.} makes both the interaction exponent $N_{\mathfrak I}$ and the associated multiplicities transparent. The key observation is that the relevant counting can be reduced to three effective occupation numbers $m_1,m_2,m_3$ in the chessboard, rather than nine independent block populations. A refined finite-case analysis then establishes that, in every regime, the combinatorial growth is compensated either by local Clifford contraction or by the $CZ$-induced damping encoded in $N_{\mathfrak I}$, yielding~\cref{eq:k3-relative}.

\emph{Constant-depth implementation and applications.}---
The commuting diagonal structure of the phase circuit leads to efficient implementations. First, without using auxiliary systems, the $CZ$ gates can be scheduled by edge coloring the sampled interaction graph. With high probability, the complete circuit has depth $O(\gamma\log n)$~\cite{bremner2016average}.

More importantly, the diagonal structure enables a constant-depth implementation using adaptive circuits and classical feedforward. This is achieved by mapping the interaction graph to a bounded-degree graph-state extension~\cite{Hoyer2006resources}. More precisely, Vizing's theorem bounds the $CZ$ depth by four. Together with one initial Hadamard layer, one simultaneous measurement layer, and one layer of outcome-dependent single-qubit Clifford gates---into which the random phase gates and Pauli corrections can be absorbed---this gives a total quantum depth of at most \emph{seven}.  With overwhelming probability, this requires only $O(n\log n)$ measured ancillas and classical feedforward. Importantly, these adaptive ingredients allow our construction to circumvent the nonadaptive depth barriers for shallow designs~\cite{schuster2024random,cui2025unitary}. Compared with the current constant-depth implementation~\cite{foxman2025random}, our protocol upgrades the guarantee from measurable error to the stronger relative-error notion, while reducing the stated ancilla scaling from $n\,\mathrm{polylog}(n/\epsilon)$ to $O(n\log(n/\epsilon))$. Moreover, such a modest space-for-depth trade-off is particularly well-suited to scalable neutral-atom architectures with large qubit arrays and parallel measurements~\cite{bluvstein2024logical}.

Our explicit moment bounds translate a target relative error $\epsilon$ directly into circuit resources. We highlight two applications: global classical shadows and multiparameter quantum metrology.

Global classical shadows are powerful for estimating fidelities and other low-rank observables, but sufficiently global randomized measurements can require deep circuits~\cite{huang2020predicting,schuster2024random}. A shadow protocol applies a random unitary $V$ followed by a projective measurement, inducing measurement states $\Phi_{V,\mb b}=V^\dagger\ket{\mb b}\!\bra{\mb b}V$. Taking $V=U^\dagger$ with $U\sim\mathcal E_{\rm shallow}^{(\gamma)}$, the averaged basis moments $D^{-1}\sum_{\mb b}\mathbb E_U\Phi_{U^\dagger,\mb b}^{\otimes k}$ coincide, by computational-basis symmetry, with $\mb M_{\mathcal E_{\rm shallow}^{(\gamma)}}^{(k)}$ for $k=2,3$~\cite{zhang2025robust}. Thus our second- and third-moment bounds directly control the shadow bias and variance. For any positive observable $O$, a relative-error $\epsilon$-approximate $3$-design gives~\cite{schuster2024random}
$\operatorname{Var}[\tr(\widehat{\rho}O)]\le 3\tr(O^2)+10\epsilon\tr(O)^2$
and
$|\operatorname{Bias}[\tr(\widehat{\rho}O)]|\le 2\epsilon\tr(O)$. In particular, rank-one fidelity observables retain dimension-independent variance and $O(\epsilon)$ bias, achievable here with ancilla-free depth $O(\log(n/\epsilon))$ or adaptive quantum depth seven.

Our third-moment result also enables shallow randomized readout for multiparameter quantum metrology. Projective $3$-design measurements recover a constant fraction of the quantum Fisher information matrix $J$ for arbitrary pure-state encodings~\cite{zhou2026randomized,du2026complexity,mao2026near}, namely $I\succeq\kappa J$ with $\kappa=\Omega(1)$, where $I$ is the classical Fisher information matrix. While Ref.~\cite{du2026complexity} establishes this principle using approximate unitary $3$-designs, the argument depends only on the corresponding projective moments. Our projective $3$-design therefore suffices, reducing the adaptive all-to-all quantum depth from (sub)logarithmic to constant.

\emph{Conclusion and outlook.}---
Our work shows that sparse commuting phase dynamics, supplemented only by local Clifford twirling, can generate approximate projective $2$- and $3$-designs with vanishing relative error. Local twirling removes the relative-error obstruction of bare phase circuits, while logarithmically sparse $CZ$ interactions suffice to preserve Haar-like low-order moments. At third order, a new block decomposition and combinatorial analysis control the resulting non-commuting moment structure. Together with the constant-depth adaptive implementation, these results show that precise projective randomness can emerge from sparse commuting dynamics well before generic unitary scrambling.

Our results open several directions. First, the phase-circuit mechanism~\cite{nakata2014generating,ji2018pseudorandom,park2023resource,zhang2024minimal} may extend to higher-order designs through controlled non-Clifford or magic resources~\cite{haferkamp2023efficient,leone2026non, leone2021quantum,magni2025anticoncentration}, and potentially to unitary designs~\cite{liu2026almost}. 
Second, shallow relative-error designs may also find more applications in quantum characterization, algorithms and variational circuits~\cite{cerezo2021cost,zhao2022hamiltonian,yan2025variational}, 
and probes of information spreading and many-body chaos~\cite{roberts2017chaos,ho2022exact,cotler2023emergent}. Finally, extending circuit-generated randomness to Hamiltonian dynamics is particularly natural: commuting phase circuits already arise from Ising-type interactions~\cite{cao2026measurement,Diagonal2026shen,hou2026state}, suggesting a bridge between shallow design theory and experimentally native many-body dynamics~\cite{nakata2017efficient,bertini2026non}.

\section*{Acknowledgments}
As always, we emphasize that You Zhou and Zhou You are two distinct people—neither of whom is Hermitian.

We thank Jens Eisert, Jonas Haferkamp, Guoding Liu, Christopher Vairogs, and Francisca Vasconcelos for insightful discussions.
This work is supported by the National Natural Science Foundation of China(NSFC) Grant No.~12575012, Quantum Science and Technology-National Science and Technology Major Project Grant Nos.~2024ZD0301900 and 2021ZD0302000, the Shanghai QiYuan Innovation Foundation, the Shanghai Municipal Commission of Science and Technology with Grant No.~25511103200, the Shanghai Science and Technology Innovation Action Plan Grant No.~24LZ1400200, the Shanghai Pilot Program for Basic Research - Fudan University 21TQ1400100 (25TQ003), the CCF-Quantum CTek Superconducting Quantum Computing CCF-QC2025006.

\clearpage

%

\clearpage
\onecolumngrid
\begin{figure*}[!t]
\centering
\includegraphics[width=\linewidth]{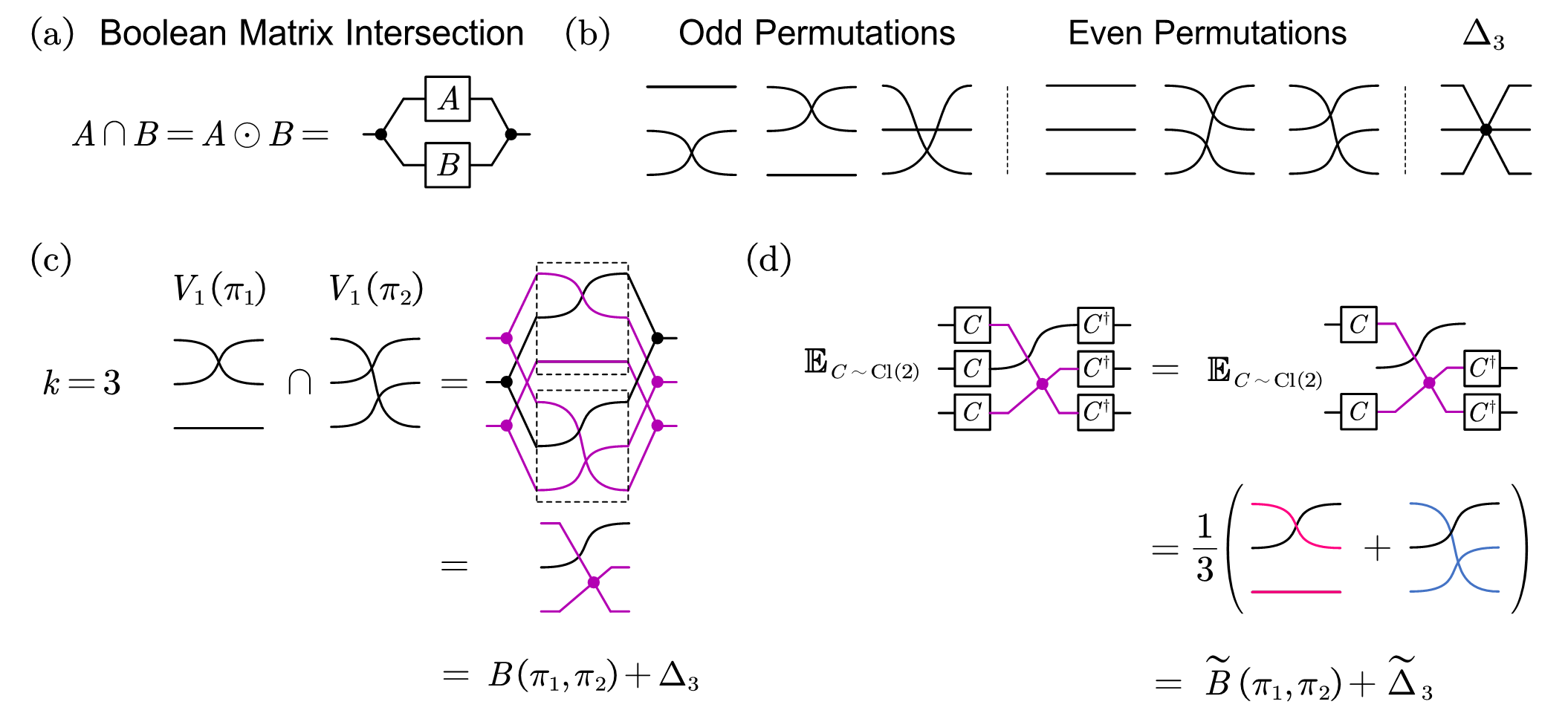}
\caption{Tensor-network diagram of the third-order Boolean-support operator blocks.
 Panel (d) illustrates a representative instance of~\cref{eq:em-B3-twirl} in such tensor-network notation. Panels (a)--(c) provide the corresponding background: (a) defines Boolean-support intersection; (b) displays the odd and even permutations and their common block $\Delta_3$; and (c) applies the intersection introduced in (a) to an odd--even permutation pair, yielding $V_1(\pi_1)\cap V_1(\pi_2)=B(\pi_1,\pi_2)+\Delta_3$.
}
\label{fig:k3-block-twirl}
\end{figure*}
\twocolumngrid
\textbf{End Matter}\\

\emph{Proof of~\cref{prop:bare-phase-obstruction}.---}
The vector \(\ket{\mathbf 0}^{\otimes 2}\) is invariant under the global swap and belongs
to the support of \(\Delta_2^{\otimes n}\). Therefore, we obtain that $\bra{\mb{0}}^{\otimes 2}
\mb M_{\mc E_{\mathrm{bare}}}^{(2)}
\ket{\mathbf 0}^{\otimes 2}=D^{-2}$ and $\bra{\mb{0}}^{\otimes 2}
\mb M_{\mc E_{\mathrm{Haar}}}^{(2)}
\ket{\mathbf 0}^{\otimes 2}=2/[D(D+1)]$. Then, it follows that
\[
\left|
\bra{\mb{0}}^{\otimes 2}
\left(
\mb M_{\mc E_{\mathrm{bare}}}^{(2)}
-
\mb M_{\mc E_{\mathrm{Haar}}}^{(2)}
\right)
\ket{\mathbf 0}^{\otimes 2}\right|
=
\frac{D-1}{D^{2}(D+1)}.
\]

Because the operator norm of a Hermitian operator is at least the absolute
value of its expectation on any normalized vector,
\[
\left\|
\mb M_{\mc E_{\mathrm{bare}}}^{(2)}
-
\mb M_{\mc E_{\mathrm{Haar}}}^{(2)}
\right\|_{\infty}
\ge
\frac{D-1}{D^{2}(D+1)}.
\]
Using
\(r_{\mathrm{sym}}^{(2)}=D(D+1)/2\)
and the normalized operator-norm characterization in
~\cref{lem:relative-error-opnorm:main}, we obtain
\[
\epsilon_2(\mc E_{\mathrm{bare}})
\ge
\frac{D(D+1)}{2}
\frac{D-1}{D^{2}(D+1)}
=
\frac{D-1}{2D}.
\]
Finally, \(D=2^n\), so the right-hand side converges to \(1/2\) as
\(n\to\infty\). Hence,
\[
\epsilon_2(\mc E_{\mathrm{bare}})
\ge
\frac12-o(1),
\]
which proves the proposition.

\emph{Background and proof sketch of \cref{thm:third-moment}}.---
Analogously to the second-order analysis, we use the dense phase ensemble as an intermediate reference. Controlling the sparse-to-dense operator-norm gap constitutes the main technical challenge. Rather than reproducing the full proof, we highlight two structural ingredients that may also be useful at higher orders: the Boolean-support block decomposition and the bishop method for organizing the resulting combinatorics. 

\paragraph{Boolean-support blocks and their Clifford twirls.}

The twirling of diagonal circuits naturally calls for a Boolean language. For Boolean-support operators \(A\) and \(B\), we use \(A\cap B\) and
\(A\cup B\) to denote their elementwise intersection and union in the
computational basis, respectively. 

At second order, we use \(\Delta_2\) to decompose the combined Boolean support \(\id_4\cup\sw_2\)  into the three disjoint blocks
\(\Delta_2\), \(\id_4-\Delta_2\), and
\(\sw_2-\Delta_2\).  At third order, the combined Boolean support is the Boolean union \(\bigcup_{\pi\in S_3}V_1(\pi)\). Its fully aligned part is
\(\Delta_3=\ket{0,0,0}\!\bra{0,0,0}
+\ket{1,1,1}\!\bra{1,1,1}\). Every remaining matrix unit belongs to
exactly two permutation supports, one odd and one even, motivating the
definition $B(\pi_1,\pi_2)
:=
V_1(\pi_1)\cap V_1(\pi_2)-\Delta_3$.
Thus, \(\Delta_3\) and the nine \(B(\pi_1,\pi_2)\) blocks form a disjoint
partition of the Boolean union.

We next evaluate the action of the local Clifford twirl
\(\mc T_3(X):=\mathbb E_{C\sim\mathrm{Cl}(2)}
C^{\otimes3}X(C^\dagger)^{\otimes3}\)
on the ten three-copy blocks. For \(\Delta_3\), a direct Weingarten
calculation shows that all six permutation components have the same
coefficient.
\begin{equation}
\widetilde\Delta_3
=
\mc T_3(\Delta_3)
=
\frac{1}{12}\sum_{\pi\in S_3}V_1(\pi).
\label{eq:em-delta3-twirl}
\end{equation}
The twirl of \(B(\pi_1,\pi_2)\) follows from the observation illustrated in~\cref{fig:k3-block-twirl}: the intersection
\(B(\pi_1,\pi_2)+\Delta_3\) is a permutation-dressed version of
\(\id_2\otimes\Delta_2\). Since copy permutations commute with the local
Clifford action, the calculation reduces to the second-order identity
\(\widetilde\Delta_2=(\id_4+\sw_2)/3\). Consequently,
\begin{equation}
\widetilde B(\pi_1,\pi_2)+\widetilde\Delta_3
=
\frac{1}{3}
\left[
V_1(\pi_1)+V_1(\pi_2)
\right],
\label{eq:em-B3-twirl}
\end{equation}

A complete derivation of~\cref{eq:em-B3-twirl} is provided in the Supplemental Material, while~\cref{fig:k3-block-twirl} represents the block construction and twirling identity diagrammatically using tensor-network notation.

\paragraph{Block expansion and the sparse-to-dense gap.}

\cref{eq:em-delta3-twirl} and
\cref{eq:em-B3-twirl} determine the local Clifford twirl of all ten
blocks. Since the Clifford gates are sampled independently on different
qubits, the twirl of a global block factorizes as $\widetilde B_{\mathfrak I}$ defined in the main text.
These identities therefore provide an explicit local description of every
tensor-product operator appearing in the third-moment expansion.

We next group the position assignments \(\mathfrak I\) according to their
count data
\(\mb I=[|I^{(0)}|;\{|I^{(\pi_1,\pi_2)}|\}_{\pi_1,\pi_2}]\).
Let \(\widetilde{\mc C}_{\mb I}\) denote the collection of all twirled
blocks with the same count data. Since the interaction exponent depends
only on these cardinalities,~\cref{eq:shallow3-main} becomes
\begin{equation}
\mb M_{\mc E_{\mathrm{shallow}}^{(\gamma)}}^{(3)}
=
D^{-3}
\sum_{\mb I}
\left(1-\frac{2\gamma\log n}{n}\right)^{N_{\mb I}}
\sum_{\widetilde B_{\mathfrak I}\in\widetilde{\mc C}_{\mb I}}
\widetilde B_{\mathfrak I}.
\label{eq:em-grouped-third-moment}
\end{equation}
Here, \(N_{\mb I}\) counts pairs of nontrivial blocks whose odd and even
permutation labels are both different. Each such pair contributes one
factor \(1-2p\) under the Bernoulli averaging of the corresponding
\(CZ\) gate.

For the dense phase ensemble, \(p=1/2\), and hence all configurations
with \(N_{\mb I}>0\) vanish. Subtracting the dense moment from the
shallow one therefore gives
\begin{equation}
\mb M_{\mc E_{\mathrm{shallow}}^{(\gamma)}}^{(3)}
-
\mb M_{\mc E_{\mathrm{phase}}}^{(3)}
=
D^{-3}
\sum_{\substack{\mb I,N_{\mb I}>0}}
\left(1-\frac{2\gamma\log n}{n}\right)^{N_{\mb I}}
\sum_{\widetilde B_{\mathfrak I}\in\widetilde{\mc C}_{\mb I}}
\widetilde B_{\mathfrak I}.
\label{eq:em-shallow-phase-gap}
\end{equation}

\paragraph{The bishop method.}
Although~\cref{eq:em-shallow-phase-gap} remains complicated, it
identifies a natural route to estimating its operator norm. As a simplified illustration of the proof strategy, we focus on
the contribution associated with a fixed count configuration \(\mb I\). Define
\(\mb T(\mb I):=
\left(1-\frac{2\gamma\log n}{n}\right)^{N_{\mb I}}
\sum_{\widetilde B_{\mathfrak I}\in\widetilde{\mc C}_{\mb I}}
\widetilde B_{\mathfrak I}\).
Using
\(\|\widetilde\Delta_3\|_\infty=1/2\) and
\(\|\widetilde B(\pi_1,\pi_2)\|_\infty=2/3\), a direct termwise estimate
gives
\begin{equation}
\|\mb T(\mb I)\|_\infty
\le
\left(1-\frac{2\gamma\log n}{n}\right)^{N_{\mb I}}
\#\widetilde{\mc C}_{\mb I}
\left(\frac23\right)^n .
\label{eq:k3-basic-count-bound}
\end{equation}
The problem is therefore to balance three quantities: the local
contraction \((2/3)^n\), the multiplicity
\(\#\widetilde{\mc C}_{\mb I}\), and the interaction exponent \(N_{\mb I}\). The remaining difficulty is that both the multiplicity and the interaction exponent depend on nine coupled block occupations in \(\mb I\). To organize these parameters and make their competition transparent, we introduce the \emph{bishop method}.

As shown in~\cref{fig:bishop}(a), the central device of the bishop method is a \(3\times3\) chessboard on which we organize the nine nontrivial block labels, where \(\pi_1,\pi_1',\pi_1''\) are the three odd permutations and
\(\pi_2,\pi_2',\pi_2''\) are the three even permutations. The
chessboard is defined up to independent row and column permutations and
transposition, all of which leave \(N_{\mb I}\) and
\(\#\widetilde{\mc C}_{\mb I}\) invariant. For each label
\(\mathfrak x\), let
\(n_{\mathfrak x}:=|I^{(\mathfrak x)}|\) denote its occupation.

Using the symmetries of the chessboard, we may label these
cells as \(\mathfrak a\) and \(\mathfrak e\), and define
\begin{equation}
m_1:=n_{\mathfrak a},
\qquad
m_2:=n_{\mathfrak e},
\qquad
m_1\ge m_2.
\label{eq:em-bishop-m12}
\end{equation}

\begin{figure}[t]
\centering \includegraphics[width=\linewidth]{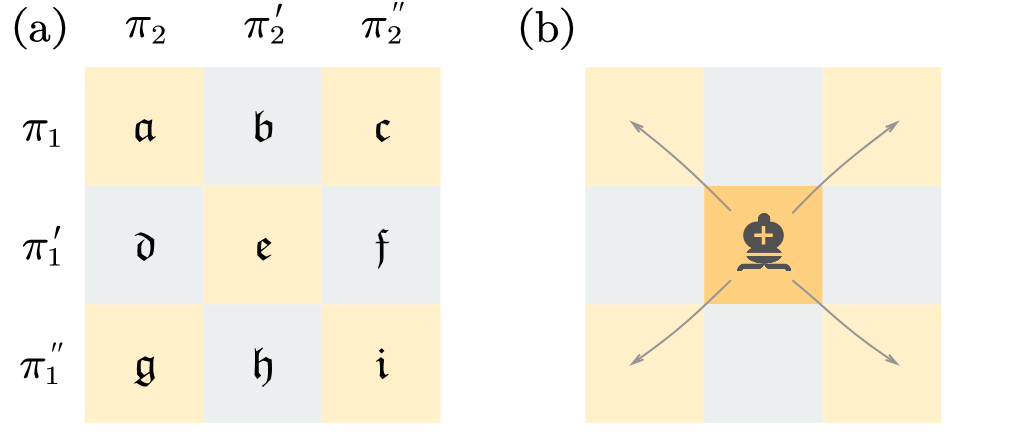} 
    
    \caption{The bishop method for third-order block counting.
(a) The nine nontrivial blocks $B(\pi_1,\pi_2)$ are arranged on a $3\times3$ board. For example, the cell $\mathfrak a$ corresponds to $B(\pi_1,\pi_2)$, and 
$|I^{(\mathfrak a)}|=|I^{(\pi_1,\pi_2)}|$.
(b) The bishop-like moves identify pairs in different rows and columns whose occupation products contribute to $N_{\mb I}$.
}

    \label{fig:bishop}
\end{figure}

We further define $m_3
:=
\max\{
n_{\mathfrak b},
n_{\mathfrak c},
n_{\mathfrak d},
n_{\mathfrak g}
\}$ as
the largest occupation sharing a row or column with the \(m_1\)-cell.
The maximality of \(m_1m_2\) and the definition of \(m_3\) imply that
all nine occupations $n_{\mathfrak x}$ are controlled by the three effective parameters
\(m_1,m_2,m_3\), together with the derived ratio
\(m_1m_2/m_3\). The precise assignment of these bounds to the
individual cells is given in the Supplemental Material.

This three-parameter reduction, which we call the `bishop method', is a key step in our proof. Here, $N_{\mb I}$ sums the products of occupations in cells lying in different rows and columns, as visualized by the bishop-like moves in~\cref{fig:bishop}(b). The bishop method extracts the simple lower bound \(N_{\mb I}\ge m_1m_2\) while simultaneously controlling the multiplicity \(\#\widetilde{\mc C}_{\mb I}\) through the same three parameters. After this reduction, the remaining estimates become tractable.

The remainder of the proof involves more detailed operator and
combinatorial estimates, including block-combination bounds, Hermitian pairing, and a finite case analysis over \(m_1,m_2,m_3\). The detailed estimates are provided in the Supplemental Material. They show that the combinatorial multiplicity is always compensated either by the local Clifford contraction or by the exponent $\left(1-\frac{2\gamma\log n}{n}\right)^{N_{\mb I}}$.
Consequently, combining the resulting sparse-to-dense bound with the exponentially small dense-to-Haar estimate and~\cref{lem:relative-error-opnorm:main} proves the third-moment result in~\cref{thm:third-moment}.

\clearpage
\appendix
\onecolumngrid
\startappendixtoc
\vspace{3em}
\begin{center}
\textbf{\large Supplementary Material: \smallskip Ultra-Precise Quantum Projective Designs in Constant Depth}
\end{center}

All complete proofs and technical details are collected in the Supplemental Material. ~\cref{sec:prelim-relative} introduces the ensemble notation and the relative-error formulation. \cref{ap:mom-func} derives the second- and third-order moment operators, providing the foundation for the subsequent relative-error analysis. \cref{Ap:2twirled} and \cref{Ap:3twirled} establish the second- and third-order relative-error bounds, respectively. \cref{Ap:3twirled} also develops the Boolean-support block decomposition and the bishop method used to control the third-order combinatorics. Finally, \cref{ap:log-sparsity-necessary} proves that $p=\Omega(\log n/n)$ is necessary for bounded second-order relative error within phase circuit architecture, thereby establishing the asymptotic optimality of our logarithmic sparsity scaling.

The correspondence between the main-text results and their complete proofs is as follows. The formal statement and proof of \cref{thm:main-informal} are given in \cref{th:constdepth}. The second-order result in \cref{thm:second-moment} follows by combining \cref{prop:k2-phase-haar-norm} and \cref{prop:k2-shallow-phase-norm} through the triangle inequality and the normalized operator-norm characterization in \cref{lem:relative-error-opnorm}. Likewise, the third-order result in \cref{thm:third-moment} follows from \cref{prop:k3-phase-haar-opnorm}, \cref{prop:twirled-third-moment-main}, and \cref{lem:relative-error-opnorm}. Finally, \cref{prop:log-sparsity-optimal-main} is proved in \cref{ap:log-sparsity-necessary}.
\appendixtableofcontents
\clearpage
\begin{appendix}

\section{Preliminaries for relative-error analysis}
\label{sec:prelim-relative}

In ~\cref{sec:prelim-relative}, we collect the ensemble definitions and the state-moment
notation used in the relative-error analysis for $k=2$ and $k=3$.
Throughout the estimates for the shallow phase ensemble, we work in the parameter
regime
\[
0<p<\frac12,
\qquad\text{equivalently}\qquad
0<1-2p<1.
\]
In particular, for \(p=\gamma\log n/n\), our asymptotic statements concern
fixed \(\gamma>0\) and sufficiently large \(n\). The endpoint \(p=1/2\) is
used only to define the dense phase ensemble and is treated separately.

\subsection{Notation and conventions}
\label{ap:notation}

In this subsection, we collect the notation and mathematical conventions used
throughout the appendices. We write
\[
[n]:=\{0,1,\ldots,n-1\},
\qquad
\mc H_n:=(\mathbb C^2)^{\otimes n},
\qquad
D:=\dim\mc H_n=2^n,
\qquad
\ket{\mb 0}:=\ket{0}^{\otimes n}.
\]
For a local operator \(A\) and a subset \(I\subseteq[n]\), we write
$A^{\otimes I}:=\bigotimes_{j\in I} A_j,$
where \(A_j\) denotes a copy of \(A\) acting on the $j$-th qubit site.
We write
\(I_1\sqcup\cdots\sqcup I_s\) for a disjoint union. We use
\(\id_r\) and \(\mathbf 0_r\) for the \(r\times r\) identity and zero
matrices, respectively. For an operator \(X\), \(\|X\|_\infty\) denotes its
operator norm, and \(X\preceq Y\) means that \(Y-X\) is positive semidefinite.
For a Pauli string \(P=\bigotimes_{q\in[n]}P_q\), with
\(P_q\in\{\id_2,X,Y,Z\}\), its Pauli weight is
\(\mathrm{wt}(P):=\bigl|\{q\in[n]:P_q\neq\id_2\}\bigr|\).

Let \(S_k\) denote the symmetric group on \(k\) elements, and let
\(S_{k,\mathrm{even}}\) and \(S_{k,\mathrm{odd}}\) denote its even and odd
permutations, respectively. For \(\pi\in S_k\), we denote by \(V_n(\pi)\)
the corresponding permutation operator on \(\mc H_n^{\otimes k}\), defined by
\[
V_n(\pi)
\bigl(
\ket{\psi_1}\otimes\cdots\otimes\ket{\psi_k}
\bigr)
=
\ket{\psi_{\pi^{-1}(1)}}\otimes\cdots\otimes
\ket{\psi_{\pi^{-1}(k)}}.
\]
Under the sitewise identification
\(\mc H_n^{\otimes k}\cong(\mathbb C^2)^{\otimes kn}\), one has $V_n(\pi)=V_1(\pi)^{\otimes n}.$
The projector onto the fully symmetric subspace is $\Pi_{\mathrm{sym}}^{(k)}
:=
\frac{1}{k!}\sum_{\pi\in S_k}V_n(\pi).$
We denote its rank by
\begin{equation}\label{eq:rsym-def}
r_{\mathrm{sym}}^{(k)}
:=
\dim\operatorname{Sym}^k(\mc H_n)
=
\binom{D+k-1}{k}.
\end{equation}

For an operator \(X\) acting on the single-site \(k\)-copy space, define the
single-qubit Clifford twirling channel by
\[
\mc T_k(X)
:=
\mathbb E_{C\sim\mathrm{Cl}(2)}
\left[
C^{\otimes k}X(C^\dagger)^{\otimes k}
\right].
\]
Throughout the appendices, a tilde indicates the corresponding locally
twirled operator; that is, $\widetilde X:=\mc T_k(X),$
with the copy number \(k\) clear from context. Independent local Clifford
twirling on \(n\) qubit sites is therefore represented by
\(\mc T_k^{\otimes n}\).

Finally, \(\mathrm{negl}(n)\) denotes a negligible function of \(n\), namely
a function satisfying $\mathrm{negl}(n)=o(n^{-c})$
for every constant \(c>0\). An event is said to hold with overwhelming
probability if its failure probability is negligible in \(n\).

\subsection{State ensembles and moment operators}\label{ap:sparsedef}

We first define the three ensembles considered in this work.

\paragraph{Shallow phase ensemble.}
The random unitaries in the shallow phase ensemble are of the form $U=U_{\mathrm{loc}}\,U_SU_{CZ} H^{\otimes n}\in\mc E_{\mathrm{shallow}}^{(\gamma)},$
where the local Clifford layer is
\begin{equation}\label{eq:prelim-local}
U_{\mathrm{loc}}=\bigotimes_{q\in[n]} U_q,
\end{equation}
and each $U_q$ is sampled independently and uniformly from the single-qubit Clifford group $\mathrm{Cl}(2)$. The diagonal layers are given by
\begin{equation}\label{eq:prelim-sparse-ensemble}
U_S=\prod_{q\in[n]} S_q^{A_{q,q}},
\qquad
U_{CZ}=\prod_{0\le i<j\le n-1} CZ_{i,j}^{A_{i,j}}.
\end{equation}
Here $A$ is a random symmetric matrix: the diagonal entries $A_{q,q}$ are sampled independently and uniformly from $\{0,1,2,3\}$, while the off-diagonal entries $A_{i,j}$ for $i<j$ are independent and identically distributed (i.i.d.) Bernoulli$(p)$ with $p=\frac{\gamma\log n}{n},$
where $\gamma>0$ is a constant.
\paragraph{Dense phase ensemble.}
The dense phase ensemble is the corresponding dense version of the shallow phase  ensemble. Concretely, \(\mc E_{\mathrm{phase}}\) consists of unitaries $U=U_{\mathrm{loc}}U_SU_{CZ}H^{\otimes n}$,
with the same circuit structure and the same distributions of
\(U_{\mathrm{loc}}\) and \(U_S\) as in
\cref{eq:prelim-local,eq:prelim-sparse-ensemble}. The only difference is that
each two-qubit interaction variable \(A_{i,j}\), \(0\le i<j\le n-1\), is sampled
i.i.d. from \(\operatorname{Bernoulli}(1/2)\). Thus,
\(\mc E_{\mathrm{phase}}\) is obtained from the same parametrized family by
setting the interaction probability to \(p=1/2\), instead of
\(p=\gamma\log n/n\).

\paragraph{Haar ensemble.}
The Haar reference ensemble \(\mc E_{\mathrm{Haar}}\) consists of
Haar-random~\cite{mele2024introduction} unitaries \(U\in U(D)\). Its
associated state ensemble $\mc E_{\mathrm{Haar}}$ is obtained by applying \(U\) to the reference state \(\ket{\mb 0}\).

\medskip

For any ensemble $\mc E$ of $n$-qubit unitaries, we associate the random pure
state \(\ket{\psi_U}:=U\ket{\mb 0}\) and its density matrix
\(\rho_U:=\ket{\psi_U}\!\bra{\psi_U}
=U\ket{\mb 0}\!\bra{\mb 0}U^\dagger\).
The $k$-th moment operator of $\mc E$ is then defined by
\begin{equation}\label{eq:prelim-moment-def-new}
\mb M_{\mc E}^{(k)}
:=
\mathbb E_{U\sim\mc E}\,
\rho_U^{\otimes k}
=
\mathbb E_{U\sim\mc E}
\bigl(U\ket{\mb 0}\!\bra{\mb 0}U^\dagger\bigr)^{\otimes k}.
\end{equation}
This is the notion of moment used throughout the paper. In particular,
\begin{equation}
\mb M_{\mc E_{\mathrm{shallow}}^{(\gamma)}}^{(k)}
=
\mathbb E_{U\sim\mc E_{\mathrm{shallow}}^{(\gamma)}}
\bigl(U\ket{\mb 0}\!\bra{\mb 0}U^\dagger\bigr)^{\otimes k},
\end{equation}
and similarly for $\mc E_{\mathrm{phase}}$ and the Haar ensemble $\mc E_{\mathrm{Haar}}$.

For the Haar state ensemble on $n$-qubit systems, the $k$-th moment operator is supported on the fully symmetric subspace. The moment operator of the Haar random ensemble is then
\begin{equation}\label{eq:haar-moment-sym-proj}
\mb M_{\mc E_{\mathrm{Haar}}}^{(k)}
=
{\Pi_{\mathrm{sym}}^{(k)}}/r_{\mathrm{sym}}^{(k)}
\;.
\end{equation}
Thus, \(\mb M_{\mc E_{\mathrm{Haar}}}^{(k)}\) is a normalized projector, and all of its nonzero eigenvalues are equal
to \(1/r_{\mathrm{sym}}^{(k)}\).

\subsection{Relative-error approximate state $k$-designs}\label{ap:relative_design}

We next recall the relative-error notion of an approximate state design used
throughout this work.

\begin{definition}[Relative-error approximate state $k$-design]\label{defi:approx-design-relative}
An ensemble $\mc E$ is called an $\epsilon$-approximate state $k$-design in relative error if 
\begin{equation}\label{eq:prelim-relative-design}
(1-\epsilon)\,\mb M_{\mc{E}_\mathrm{Haar}}^{(k)}
\preceq
\mb M_{\mc E}^{(k)}
\preceq
(1+\epsilon)\,\mb M_{\mc{E}_\mathrm{Haar}}^{(k)}.
\end{equation}
\end{definition}
Equivalently, for every positive semidefinite operator $O\succeq 0$,
\begin{equation}\label{eq:prelim-relative-design-psd}
\left|
\tr\!\left[
\bigl(\mb M_{\mc E}^{(k)}-\mb M_{\mc{E}_\mathrm{Haar}}^{(k)}\bigr)O
\right]
\right|
\le
\epsilon\,
\tr\!\left[\mb M_{\mc{E}_\mathrm{Haar}}^{(k)}O\right].
\end{equation}
Correspondingly, one can define the relative error of $\mc E$ at order $k$ by
\begin{equation}\label{eq:prelim-relative-error}
\epsilon_k(\mc E)
:=
\sup_{\substack{O\succeq 0\\ \tr(\mb M_{\mc{E}_\mathrm{Haar}}^{(k)}O)\neq 0}}
\frac{
\left|
\tr\!\left[
\bigl(\mb M_{\mc E}^{(k)}-\mb M_{\mc{E}_\mathrm{Haar}}^{(k)}\bigr)O
\right]
\right|
}{
\tr\!\left[\mb M_{\mc{E}_\mathrm{Haar}}^{(k)}O\right]
}.
\end{equation}
Thus, \(\mc E\) is an \(\epsilon\)-approximate state \(k\)-design in
relative error if and only if \(\epsilon_k(\mc E)\le\epsilon\).

\begin{lemma}[Relative error as a normalized operator norm]
\label{lem:relative-error-opnorm}
The relative error of a state ensemble \(\mc E\) at order \(k\) is exactly
\begin{equation}\label{eq:relative-error-exact-opnorm}
\epsilon_k(\mc E)
=
r_{\mathrm{sym}}^{(k)}
\left\|
\mb M_{\mc E}^{(k)}
-
\mb M_{\mc E_{\mathrm{Haar}}}^{(k)}
\right\|_\infty .
\end{equation}
Consequently, \(\mc E\) is an \(\epsilon\)-approximate state \(k\)-design in
relative error whenever the right-hand side is at most \(\epsilon\).
\end{lemma}

\begin{proof}
Both \(\mb M_{\mc E}^{(k)}\) and
\(\mb M_{\mc E_{\mathrm{Haar}}}^{(k)}\) are supported on the fully symmetric
subspace. Hence, in~\cref{eq:prelim-relative-error}, replacing any
\(O\succeq0\) by
\(\Pi_{\mathrm{sym}}^{(k)}O\Pi_{\mathrm{sym}}^{(k)}\) leaves both the
numerator and the denominator unchanged. It therefore suffices to optimize
over positive semidefinite operators supported on
\(\operatorname{Sym}^k(\mc H_n)\).

On this subspace,
\(\mb M_{\mc E_{\mathrm{Haar}}}^{(k)}
=\Pi_{\mathrm{sym}}^{(k)}/r_{\mathrm{sym}}^{(k)}\), and hence
\(\tr[\mb M_{\mc E_{\mathrm{Haar}}}^{(k)}O]
=\tr(O)/r_{\mathrm{sym}}^{(k)}\). Therefore,
\begin{equation}
\epsilon_k(\mc E)
=
r_{\mathrm{sym}}^{(k)}
\sup_{\substack{O\succeq0\\ \tr(O)>0}}
\frac{
\left|
\tr\!\left[
\left(
\mb M_{\mc E}^{(k)}
-
\mb M_{\mc E_{\mathrm{Haar}}}^{(k)}
\right)O
\right]
\right|
}{
\tr(O)
}.
\end{equation}
Since the quotient is homogeneous in \(O\), we can restrict the supremum to
operators satisfying \(\tr(O)=1\). Moreover,
\[
\left|
\tr\!\left[
\left(
\mb M_{\mc E}^{(k)}
-
\mb M_{\mc E_{\mathrm{Haar}}}^{(k)}
\right)O
\right]
\right|
\le
\left\|
\mb M_{\mc E}^{(k)}
-
\mb M_{\mc E_{\mathrm{Haar}}}^{(k)}
\right\|_\infty .
\]
Equality is attained by choosing \(O\) as the rank-one projector onto an
eigenvector corresponding to an eigenvalue of largest absolute value.
Consequently,
\begin{equation}\label{eq:rsym-app}
    \epsilon_k(\mc E)
=
r_{\mathrm{sym}}^{(k)}
\left\|
\mb M_{\mc E}^{(k)}
-
\mb M_{\mc E_{\mathrm{Haar}}}^{(k)}
\right\|_\infty ,
\end{equation}
which proves the claim.
\end{proof}

This lemma plays a critical role in the subsequent analysis. In the following sections, we first present explicit formulas for the moment operators $\mb M_{\mc E_{\mathrm{phase}}}^{(k)}$ and $\mb M_{\mc E_{\mathrm{shallow}}^{(\gamma)}}^{(k)}$ in~\cref{ap:mom-func}, and then estimate the relative error of $\mc E_{\mathrm{shallow}}^{(\gamma)}$ through the operator-norm differences
\begin{equation}\label{eq:OpGap}
\bigl\|\mb M_{\mc E_{\mathrm{shallow}}^{(\gamma)}}^{(k)}-\mb M_{\mc E_{\mathrm{phase}}}^{(k)}\bigr\|_\infty
\quad\text{and}\quad
\bigl\|\mb M_{\mc E_{\mathrm{phase}}}^{(k)}-\mb M_{\mc E_{\mathrm{Haar}}}^{(k)}\bigr\|_\infty.    
\end{equation}

Thus the relative-error analysis is reduced to proving quantitative bounds for these two operator norm gaps. Finally, the main result (\cref{thm:main-informal}) of this work is formally presented as follows.

\begin{theorem}[Relative-error design from shallow
phase circuits]\label{th:constdepth}
For every fixed $\epsilon\in(0,1)$ and all sufficiently large \(n\), the
shallow phase ensemble over $n$ qubits forms an $\epsilon$-approximate
state $2$-design and $3$-design in relative error, with overwhelming
probability, in
\begin{itemize}
\item $O(1)$-depth, using
    $O\!\left(n\log\frac{n}{\epsilon}\right)$ ancilla qubits,
\item $O\!\left(\log\frac{n}{\epsilon}\right)$-depth, without ancillas.
\end{itemize}
\end{theorem}

\begin{proof}
{For each \(k\in\{2,3\}\), the triangle inequality gives}
\begin{align}
\left\|
\mb M_{\mc E_{\mathrm{shallow}}^{(\gamma)}}^{(k)}
-
\mb M_{\mc E_{\mathrm{Haar}}}^{(k)}
\right\|_\infty
&{\le
\left\|
\mb M_{\mc E_{\mathrm{shallow}}^{(\gamma)}}^{(k)}
-
\mb M_{\mc E_{\mathrm{phase}}}^{(k)}
\right\|_\infty}
\nonumber+
\left\|
\mb M_{\mc E_{\mathrm{phase}}}^{(k)}
-
\mb M_{\mc E_{\mathrm{Haar}}}^{(k)}
\right\|_\infty.
\label{eq:thm-main-triangle}
\end{align}
{Combining Propositions~\ref{prop:k2-phase-haar-norm} and~\ref{prop:k2-shallow-phase-norm} for \(k=2\), and Propositions~\ref{prop:k3-phase-haar-opnorm} and~\ref{prop:twirled-third-moment-main} for \(k=3\), with the normalized operator-norm characterization in~\cref{lem:relative-error-opnorm}, we obtain the following bound. For every
\(k\in\{2,3\}\), there exist absolute constants
\(A_k>0\), \(B_k>0\), and \(\gamma_k>0\) such that}
\begin{equation}\label{eq:thm-main-unified-opnorm}
\left\|
\mb M_{\mc E_{\mathrm{shallow}}^{(\gamma)}}^{(k)}
-
\mb M_{\mc E_{\mathrm{Haar}}}^{(k)}
\right\|_\infty
\le
2^{-kn}
\Bigl(
n^{A_k-B_k\gamma}
+
\mathrm{negl}(n)
\Bigr)
\end{equation}
for every fixed \(\gamma>\gamma_k\) and all sufficiently large \(n\).

Define $\gamma_\star
:=
\max_{k\in\{2,3\}}
\left\{
\gamma_k,\,
\frac{A_k}{B_k}
\right\}$,
{and fix any constant \(\gamma>\gamma_\star\).}
Then we denote $\delta
:=
\min_{k\in\{2,3\}}
\bigl(B_k\gamma-A_k\bigr)
>0,$ so \(n^{A_k-B_k\gamma}\le n^{-\delta}\) for each \(k\in\{2,3\}\).
For sufficiently large \(n\), both \(n^{-\delta}\) and the negligible term
in~\cref{eq:thm-main-unified-opnorm} are bounded by \(\epsilon/2\). Hence,
by~\cref{lem:relative-error-opnorm},
\begin{equation}
r_{\mathrm{sym}}^{(k)}
\left\|
\mb M_{\mc E_{\mathrm{shallow}}^{(\gamma)}}^{(k)}
-
\mb M_{\mc E_{\mathrm{Haar}}}^{(k)}
\right\|_\infty
\le
D^k
\left\|
\mb M_{\mc E_{\mathrm{shallow}}^{(\gamma)}}^{(k)}
-
\mb M_{\mc E_{\mathrm{Haar}}}^{(k)}
\right\|_\infty
\le
\epsilon.
\end{equation}
Here we used \(r_{\mathrm{sym}}^{(k)}\le D^k\).
By the definition of relative-error approximate state designs, this proves
that \(\mc E_{\mathrm{shallow}}^{(\gamma)}\) is an
\(\epsilon\)-approximate state \(2\)-design and \(3\)-design in relative
error. To achieve a prescribed accuracy $\epsilon$, it suffices to choose $\gamma$ such that $n^{A_k-B_k\gamma}\le \epsilon/2$ for $k=2,3$, which requires only an expected interaction degree $\gamma\log n=O(\log(n/\epsilon))$.

The implementation bounds are established separately in the following \cref{prop:implementation-bounds} of \cref{ap:efficient-implementation}. Combining those bounds with the relative-error estimates above completes the proof.
\end{proof}

\subsection{Adaptive and ancilla-free implementations}
\label{ap:efficient-implementation}

We now establish the implementation bounds stated in
\cref{th:constdepth}. Unlike the moment analysis, these bounds follow
directly from the graph-state structure of the sampled phase circuit.

\begin{prop}[Implementation bounds for the shallow phase ensemble]
\label{prop:implementation-bounds}
Let \(p=\gamma\log n/n\), with \(\gamma>0\) fixed. A circuit sampled from
\(\mc E_{\mathrm{shallow}}^{(\gamma)}\) admits

\begin{itemize}
\item an adaptive implementation of quantum depth at most seven using
\(O(n\log \frac{n}{\epsilon})\) measured ancillas, with overwhelming probability;
\item an ancilla-free implementation of depth \(O(\log \frac{n}{\epsilon})\) on an
all-to-all architecture, with arbitrarily high polynomial probability.
\end{itemize}
\end{prop}

\begin{proof}
Let \(G=([n],E)\sim G(n,p)\) denote the interaction graph of the sampled
\(CZ\) layer. We first consider an ancilla-assisted adaptive
implementation using mid-circuit measurements and classical feedforward.
It suffices to prepare
\[
U_{CZ}H^{\otimes n}\ket{\mathbf 0}=\ket{G},
\]
where \(\ket{G}\) is the graph state associated with \(G\).
Ref.~\cite{Hoyer2006resources} shows that \(\ket{G}\) can be obtained by
measuring a graph-state extension of maximum degree three containing
\(O(n+|E|)\) qubits. By Vizing's theorem, the edges of this extension can
be partitioned into at most four matchings, so all of its \(CZ\) gates
can be applied in four quantum layers.

All auxiliary qubits can then be measured simultaneously. The resulting
Pauli corrections can be absorbed into the final random local Clifford
layer, together with the single-qubit phase gates in \(U_S\). The complete
sampled circuit therefore has quantum depth at most seven: one Hadamard
layer, four \(CZ\) layers, one measurement layer, and one
outcome-dependent local Clifford layer.

It remains to control the number of edges. Since
\[
|E|\sim
\operatorname{Bin}\!\left(
\binom{n}{2},\frac{\gamma\log n}{n}
\right),
\qquad
\mu_E:=\mathbb E|E|
=
\frac{\gamma(n-1)\log n}{2},
\]
the multiplicative Chernoff bound gives
\begin{equation}
\Pr\!\left(|E|\ge 2\mu_E\right)
\le
\exp\!\left(-\frac{\mu_E}{3}\right)
=
\exp\!\left[-\frac{\gamma(n-1)\log n}{6}\right].
\label{eq:implementation-edge-tail}
\end{equation}
Thus, with overwhelming probability, \(|E|=O(\gamma n\log n)\), and the
adaptive construction uses \(O(\gamma n\log n)=O(n\log\frac{n}{\epsilon})\) measured ancillas.

Without ancillas, the \(CZ\) gates can instead be scheduled directly by
edge coloring \(G\). Vizing's theorem bounds their depth by
\(\Delta(G)+1\). For every vertex \(q\),
\[
\deg_G(q)\sim\operatorname{Bin}(n-1,p),
\qquad
\mathbb E[\deg_G(q)]\le\gamma\log n.
\]
A Chernoff bound followed by a union bound shows that, for every fixed
\(c>0\), there exists a constant \(C_c>0\) such that
\[
\Pr\!\left[\Delta(G)>C_c\gamma\log n\right]\le n^{-c}.
\]
Hence the sampled circuit has ancilla-free depth
\(O(\gamma\log n)\) with arbitrarily high polynomial probability. The
Hadamard, phase, and final local Clifford layers contribute only constant
additional depth. Finally, for the fixed target accuracy \(\epsilon\) in
\cref{th:constdepth},  $C_c\gamma\log n = O(\log \frac{n}{\epsilon})$, which proves the claimed implementation bounds.
\end{proof}

\section{Moment operators of the shallow phase ensemble}\label{ap:mom-func}

In this section, we derive the second- and third-order moment operators of the shallow phase ensemble, $\mb{M}_{\mc{E}_{\text{shallow}}^{(\gamma)}}^{(2)}$ and $\mb{M}_{\mc{E}_{\text{shallow}}^{(\gamma)}}^{(3)}$. These moment formulas are the key analytical tools required for the relative-error operator-norm comparisons in the main text and in the subsequent sections of the appendix.

To this end, we first analyze the structural form of the ensemble $\mc{E}_{\text{shallow}}^{(\gamma)}$. According to~\cref{ap:sparsedef}, a unitary element $U\in \mc{E}_{\text{shallow}}^{(\gamma)}$ takes the form $U=U_{\mathrm{loc}}U_SU_{CZ}{H}^{\otimes n}.$
Since $U_S$ and $U_{CZ}$ are both diagonal, they commute. Thus, for $k\in\{2,3\}$, we can represent the $k$-th moment operator of the ensemble as
\begin{equation}\label{Eq:momentHSCZ}
\begin{split}
\mb{M}_{\mc{E}_{\text{shallow}}^{(\gamma)}}^{(k)} &= \mathbb{E}_{U\sim \mc{E}_{\text{shallow}}^{(\gamma)}} (U\ket{\mb{0}}\bra{\mb{0}}U^{\dagger})^{\otimes k}\\
    &=\mathbb E_{U_{\mathrm{loc}}}
    U_{\mathrm{loc}}^{\otimes k}
    \Bigl[
    \mathbb E_{U_{CZ}}
    U_{CZ}^{\otimes k}
    \bigl[
    \mathbb E_{U_S}
    (U_SH^{\otimes n}\ket{\mb{0}}\bra{\mb{0}}H^{\otimes n}U_S^{\dagger})^{\otimes k}
    \bigr]
    (U_{CZ}^{\dagger})^{\otimes k}
    \Bigr]
    (U_{\mathrm{loc}}^{\dagger})^{\otimes k}.
\end{split}   
\end{equation}
This decomposition conveniently separates the averages over the local Clifford layer $U_{\mathrm{loc}}$, the single-qubit diagonal phase layer $U_S$, and the two-qubit diagonal layer $U_{CZ}$.

The first averaging step is the average over the random single-qubit diagonal layer $U_S$. This step enforces a local phase-cancellation condition on each qubit across the $k$ copies. 
Specifically, twirling over the single-qubit phase layer $U_S$ acts as a filter on the computational-basis matrix units $\ket{\hat{\mb{x}}}\bra{\hat{\mb{y}}}$ in the $k$-copy, $n$-qubit Hilbert space, retaining only those whose local phase exponents cancel modulo $4$. Equivalently, following Ref.~\cite{zhang2025robust}, one obtains
\begin{equation}\label{Eq:momentHS}
\mathbb{E}_{U_S}
\left(
U_S H^{\otimes n}\ket{\mb{0}}\bra{\mb{0}}H^{\otimes n} U_S^{\dagger}
\right)^{\otimes k}
=
D^{-k}
\sum_{(\hat{\mb{x}},\hat{\mb{y}})\in C^{(k)}}
\ket{\hat{\mb{x}}}\bra{\hat{\mb{y}}},
\end{equation}
where the index set \(C^{(k)}\) specifies the \(k\)-copy
computational-basis matrix units retained by the single-qubit phase average.
Writing
\(\hat{\mb x}=(x_{q,t})_{q\in[n],\,t\in[k]}\) and
\(\hat{\mb y}=(y_{q,t})_{q\in[n],\,t\in[k]}\), it is given explicitly by
\begin{equation}\label{Eq:CmDef}
    C^{(k)} := 
    \big\{ (\hat{\mb{x}}, \hat{\mb{y}}) 
    \in (\{0,1\}^{kn})^2 
    \;\big|\;
    \forall\, q\in[n],\;
    \sum_{t=0}^{k-1} 
    (x_{q,t} - y_{q,t}) := 0 \pmod{4} 
    \big\}.
\end{equation}
That is, for each qubit index \(q\), the total phase-exponent difference
across all \(k\) copies must vanish modulo \(4\). This support-selection rule is naturally described at the level of
computational-basis support.

If \(A\) and \(B\) have entries in \(\{0,1\}\) in the computational basis, we
regard them as the sets of matrix units on which their entries equal one. The
symbols \(\cup\), \(\cap\), \(\subseteq\), and \(\in\) below refer exclusively
to these sets of supported matrix units. In particular, union and intersection
are the entrywise Boolean operations
\begin{equation}\label{eq:Boolean-support-notation}
\begin{split}
    C=A\cup B
    &:\quad c_{i,j}=a_{i,j}\vee b_{i,j},\\
    C=A\cap B
    &:\quad c_{i,j}=a_{i,j}\wedge b_{i,j}.
\end{split}
\end{equation}
Thus \(\ket{i}\!\bra{j}\in A\) means
\(\bra{i}A\ket{j}=1\), and we say $A\subseteq B$ if $B_{i,j}=1$ for all $(i,j)$ such that $A_{i,j}=1$.
Ordinary \(+\) and \(-\) continue to denote operator addition and subtraction.
If \(A\) and \(B\) are Boolean-support operators satisfying
\(B\subseteq A\), then \(A-B\) is also a Boolean-support operator, whose
Boolean support is precisely the set difference between the supports of
\(A\) and \(B\).

The remainder of this section is organized as follows. We first specialize the filtered expression above to $k=2$ and prove the expression for the second moment operator in~\cref{Ap:moment2}. We then use the Boolean-support language to organize the $k=3$ surviving computational-basis matrix units into three-copy blocks and derive the third-order moment operator in~\cref{Ap:moment3}.

\subsection{The second-order moment operator}\label{Ap:moment2}

Here, we give a convenient explicit expression for the second moment operator of the shallow phase ensemble $\mb{M}_{\mc{E}_{\text{shallow}}^{(\gamma)}}^{(2)}$. The formula below shows that the moment operator decomposes into three local two-copy Boolean-support blocks, with coefficients determined by the random \(CZ\) layer.

Throughout this subsection, sums over \(I_i,I_j,I_k\) are taken over ordered partitions
\(I_i\sqcup I_j\sqcup I_k=[n]\). We also define $\widetilde\Delta_2
:=
\frac{\id_4+\sw_2}{3},$
where \(\sw_2\) denotes the two-copy swap operator on a single qubit.
\begin{lemma}[Second moment of the shallow phase ensemble]\label{lem:sparse-second-moment}
The second moment operator of the shallow phase ensemble $\mc{E}_{\text{shallow}}^{(\gamma)}$ is
\begin{equation}
\mb{M}_{\mc{E}_{\text{shallow}}^{(\gamma)}}^{(2)}
=D^{-2} \sum_{I_i\sqcup I_j\sqcup I_k=[n]}
\Bigl(1-\frac{2\gamma \log n}{ n}\Bigr)^{|I_j|\cdot|I_k|}\,
\widetilde\Delta_2^{\otimes I_i}
\otimes (\id_4-\widetilde\Delta_2)^{\otimes I_j}
\otimes (\sw_2-\widetilde\Delta_2)^{\otimes I_k}.
\end{equation}
\end{lemma}

\begin{proof}
The proof consists of four steps. We first evaluate the average over the diagonal phase layer $U_S$, which restricts the retained two-copy computational-basis patterns to the phase-matching set $C^{(2)}$. We then average over the random $CZ$ layer $U_{CZ}$ directly at the level of computational-basis matrix units. After this raw expression is obtained, we reorganize the retained matrix units into three local Boolean-support blocks. Finally, we apply the local Clifford twirling $U_{\mathrm{loc}}$ site by site, which replaces the untwirled local blocks by their Clifford-averaged counterparts and yields the claimed formula.

We now specialize the general framework to the $k=2$ case. Write
\[
\ket{\hat{\mb x}}
=
\ket{\mb x}\otimes\ket{\mb w}
=
\ket{\mb{x,w}},
\qquad
\bra{\hat{\mb y}}
=
\bra{\mb y}\otimes\bra{\mb s}
=
\bra{\mb{y,s}},
\]
where $\mb{x},\mb{w},\mb{y},\mb{s}\in\{0,1\}^n$. Since the local sums take values in $\{0,1,2\}$, the congruence condition modulo $4$ is equivalent to ordinary equality. Therefore, \cref{Eq:CmDef} becomes
\begin{equation}\label{Eq:CmDef2}
C^{(2)}
=
\Bigl\{
(\mb{x},\mb{w},\mb{y},\mb{s})
\in(\{0,1\}^n)^4
\;\Big|\;
x_l+w_l=y_l+s_l,\ \forall l\in[n]
\Bigr\},
\end{equation}
with $|C^{(2)}|=6^n$. Substituting \cref{Eq:CmDef2} into \cref{Eq:momentHS} gives
\begin{equation}\label{Eq:moment2final}
\mathbb{E}_{U_S}
\left(
U_S H^{\otimes n}\ket{\mb{0}}\bra{\mb{0}}H^{\otimes n} U_S^{\dagger}
\right)^{\otimes 2}
=
D^{-2}
\sum_{(\mb{x},\mb{w},\mb{y},\mb{s})\in C^{(2)}}
\ket{\mb{x,w}}\bra{\mb{y,s}}.
\end{equation}
We next average over the random $CZ$ layer. It suffices to consider a single edge $(a,b)$. For a computational-basis matrix unit
$\ket{\mb{x,w}}\bra{\mb{y,s}}$, define
\begin{equation}\label{eq:Tij2-def}
T_{a,b}^{(2)}
:=
x_ax_b+w_aw_b-y_ay_b-s_as_b
\pmod 2
\in\{0,1\}.
\end{equation}
Averaging over the Bernoulli random variable deciding whether the gate $CZ_{a,b}$ is present gives
\begin{equation}\label{eq:single-cz-average-k2}
\begin{split}
{\mathbb E_{A_{a,b}}
\left(CZ_{a,b}^{A_{a,b}}\right)^{\otimes2}
\ket{\mb{x,w}}\bra{\mb{y,s}}
\left(CZ_{a,b}^{\dagger,A_{a,b}}\right)^{\otimes2}}&
=
(1-p)\ket{\mb{x,w}}\bra{\mb{y,s}}
+p\,CZ_{a,b}^{\otimes2}
\ket{\mb{x,w}}\bra{\mb{y,s}}
\left(CZ_{a,b}^{\dagger}\right)^{\otimes2}
\\
&
=
\bigl(1-2p\,T_{a,b}^{(2)}\bigr)
\ket{\mb{x,w}}\bra{\mb{y,s}}.
\end{split}
\end{equation}
Since the different $CZ$ edges are sampled independently, their contributions multiply. Using
\(p=\gamma\log n/n\), we obtain that
\begin{equation}\label{eq:CZ-raw-second}
\mathbb E_{U_{CZ}}
U_{CZ}^{\otimes2}
\ket{\mb{x,w}}\bra{\mb{y,s}}
\left(U_{CZ}^{\dagger}\right)^{\otimes2}
=
\Bigl(1-\frac{2\gamma\log n}{n}\Bigr)^{
N^{(2)}_{\mb{x},\mb{w},\mb{y},\mb{s}}
}
\ket{\mb{x,w}}\bra{\mb{y,s}},
\end{equation}
where
\begin{equation}\label{eq:N2-raw-def}
N^{(2)}_{\mb{x},\mb{w},\mb{y},\mb{s}}
:=
\sum_{a<b}T_{a,b}^{(2)}.
\end{equation}
Substituting this raw $CZ$-averaged expression into the moment formula gives
\begin{equation}\label{eq:channel-noise-raw}
\begin{split}
\mb{M}_{\mc{E}_{\text{shallow}}^{(\gamma)}}^{(2)}
&=
D^{-2}
\mathbb E_{U_{\mathrm{loc}}}
U_{\mathrm{loc}}^{\otimes2}
\Biggl[
\sum_{(\mb{x},\mb{w},\mb{y},\mb{s})\in C^{(2)}}
\Bigl(1-\frac{2\gamma\log n}{n}\Bigr)^{
N^{(2)}_{\mb{x},\mb{w},\mb{y},\mb{s}}
}
\ket{\mb{x,w}}\bra{\mb{y,s}}
\Biggr]
\left(U_{\mathrm{loc}}^{\dagger}\right)^{\otimes2}.
\end{split}
\end{equation}

We now reorganize the retained computational-basis matrix units in terms of three local two-copy Boolean-support blocks. At each qubit, the six retained local matrix units are partitioned into
\begin{equation}\label{eq:k2-local-blocks}
\Delta_2
=
\ket{0,0}\bra{0,0}
+
\ket{1,1}\bra{1,1},
\end{equation}
\begin{equation}
\id_4-\Delta_2
=
\ket{0,1}\bra{0,1}
+
\ket{1,0}\bra{1,0},
\end{equation}
and
\begin{equation}
\sw_2-\Delta_2
=
\ket{0,1}\bra{1,0}
+
\ket{1,0}\bra{0,1}.
\end{equation}
Thus each retained $n$-qubit computational-basis matrix unit
$\ket{\mb{x,w}}\bra{\mb{y,s}}$ belongs to a unique tensor product Boolean-support block of the form
\[
\Delta_2^{\otimes I_i}
\otimes
(\id_4-\Delta_2)^{\otimes I_j}
\otimes
(\sw_2-\Delta_2)^{\otimes I_k},
\qquad
I_i\sqcup I_j\sqcup I_k=[n].
\]
For a pair of qubits $(a,b)$, the value $T_{a,b}^{(2)}=1$ if and only if one of the two local matrix units lies in the block $\id_4-\Delta_2$ and the other lies in the block $\sw_2-\Delta_2$. In all other cases, $T_{a,b}^{(2)}=0$. Therefore, we use the Boolean-support notation in~\cref{eq:Boolean-support-notation} to show that
\begin{equation}\label{eq:N2-block-value}
N^{(2)}_{\mb{x},\mb{w},\mb{y},\mb{s}}
=
|I_j|\cdot|I_k|,
\end{equation}
if $\ket{\mb{x,w}}\bra{\mb{y,s}}
\in
\Delta_2^{\otimes I_i}
\otimes
(\id_4-\Delta_2)^{\otimes I_j}
\otimes
(\sw_2-\Delta_2)^{\otimes I_k}$. The raw expression in~\cref{eq:channel-noise-raw} becomes
\begin{equation}\label{Eq:channel_noise}
\begin{split}
\mb{M}_{\mc{E}_{\text{shallow}}^{(\gamma)}}^{(2)}
&=
D^{-2}
\mathbb E_{U_{\mathrm{loc}}}
U_{\mathrm{loc}}^{\otimes2}
\Biggl[
\sum_{I_i\sqcup I_j\sqcup I_k=[n]}
\Bigl(1-\frac{2\gamma\log n}{n}\Bigr)^{|I_j||I_k|}
\Delta_2^{\otimes I_i}
\otimes
(\id_4-\Delta_2)^{\otimes I_j}
\otimes
(\sw_2-\Delta_2)^{\otimes I_k}
\Biggr]
\left(U_{\mathrm{loc}}^{\dagger}\right)^{\otimes2}.
\end{split}
\end{equation}

It remains to average over the local Clifford layer. Since
$U_{\mathrm{loc}}=\bigotimes_{q\in[n]}U_q$ is sampled independently on each qubit, the local Clifford twirl factorizes site by site. In particular,
\begin{equation}\label{eq:twirl-Delta2}
\mathbb E_{U_q\sim \mathrm{Cl}(2)}
\,U_q^{\otimes 2}\Delta_2\left(U_q^\dagger\right)^{\otimes 2}
=
\frac{\id_4+\sw_2}{3}
:=
\widetilde{\Delta}_2.
\end{equation}
Moreover, since both $\id_4$ and $\sw_2$ are invariant under local Clifford conjugation,
\begin{equation}\label{eq:twirl-k2-other-blocks}
\mathbb E_{U_q\sim \mathrm{Cl}(2)}
\,U_q^{\otimes 2}(\id_4-\Delta_2)\left(U_q^\dagger\right)^{\otimes 2}
=
\id_4-\widetilde{\Delta}_2,
\qquad
\mathbb E_{U_q\sim \mathrm{Cl}(2)}
\,U_q^{\otimes 2}(\sw_2-\Delta_2)\left(U_q^\dagger\right)^{\otimes 2}
=
\sw_2-\widetilde{\Delta}_2.
\end{equation}
Applying these identities to each tensor factor in~\cref{Eq:channel_noise}, we obtain
\[
\mb{M}_{\mc{E}_{\text{shallow}}^{(\gamma)}}^{(2)}
=
D^{-2}
\sum_{I_i\sqcup I_j\sqcup I_k=[n]}
\Bigl(1-\frac{2\gamma \log n}{n}\Bigr)^{|I_j||I_k|}
\,
\widetilde\Delta_2^{\otimes I_i}
\otimes
(\id_4-\widetilde\Delta_2)^{\otimes I_j}
\otimes
(\sw_2-\widetilde\Delta_2)^{\otimes I_k},
\]
which is exactly the claimed formula. Thus, the proof of~\cref{lem:sparse-second-moment} is completed.
\end{proof}

\subsection{The third-order moment operator}\label{Ap:moment3}

We now turn to the case \(k=3\). The overall averaging architecture is the same as in the second-order case. we first average over the single-qubit phase layer \(U_S\), then over the random \(CZ\) layer \(U_{CZ}\), and finally over the local Clifford layer \(U_{\mathrm{loc}}\). The main difference is the organization of the retained computational-basis matrix units. For \(k=2\), the retained local matrix units decompose into three local two-copy Boolean-support blocks. For \(k=3\), the retained local matrix units carry a richer \(S_3\)-structure, and the correct analogue is an \(S_3\)-adapted block decomposition.

The derivation is organized as follows. First, we obtain a raw computational-basis expression for the third moment after the \(U_S\)- and \(U_{CZ}\)-averages, while keeping the local Clifford twirling explicit. This is the direct analogue of the raw second-order expression before regrouping into Boolean-support blocks. Next, we introduce the single-qubit three-copy Boolean-support blocks. These blocks play the same role as \(\Delta_2\), \(\id_4-\Delta_2\), and \(\sw_2-\Delta_2\) in the second-order derivation. Finally, we regroup the raw expression according to these blocks and apply the local Clifford twirling site by site.

\paragraph{A computational-basis expression of the third-order moment operator.}

We first derive the third-order analogue of the raw second-order expression
in~\cref{eq:channel-noise-raw}. As in the \(k=2\) case, this step does not yet
use the local block decomposition. It only records the action of the
single-qubit phase average and the random \(CZ\) average on each retained
computational-basis matrix unit, while keeping the local Clifford twirling
explicit.

\begin{lemma}[Third moment operator in the computational basis]\label{lem:third-moment-preblock}
The third moment operator of the shallow phase ensemble
\(\mc E_{\mathrm{shallow}}^{(\gamma)}\) satisfies
\begin{equation}\label{eq:moment3-preblock}
\mb M_{\mc E_{\mathrm{shallow}}^{(\gamma)}}^{(3)}
=
D^{-3}\,
\mathbb E_{U_{\mathrm{loc}}}
U_{\mathrm{loc}}^{\otimes 3}
\Biggl[
\sum_{(\mb{x},\mb{w},\mb{z},\mb{y},\mb{s},\mb{t})\in C^{(3)}}
\left(1-\frac{2\gamma\log n}{n}\right)^{
N^{(3)}_{\mb{x},\mb{w},\mb{z},\mb{y},\mb{s},\mb{t}}
}
\ket{\mb{x,w,z}}\bra{\mb{y,s,t}}
\Biggr]
\left(U_{\mathrm{loc}}^\dagger\right)^{\otimes 3},
\end{equation}
where $N^{(3)}_{\mb{x},\mb{w},\mb{z},\mb{y},\mb{s},\mb{t}}
:=
\sum_{a<b}
T_{a,b}^{(3)},$
with
\begin{equation}\label{eq:defT3}
T_{a,b}^{(3)}
:=
x_ax_b+w_aw_b+z_az_b-y_ay_b-s_as_b-t_at_b
\mod 2 \in \{0,1\}.
\end{equation}
\end{lemma}
\begin{proof}
We write the three-copy computational-basis vectors as
\[
\ket{\hat{\mb x}}
=
\ket{\mb x}\otimes\ket{\mb w}\otimes\ket{\mb z}
=
\ket{\mb{x,w,z}},
\qquad
\bra{\hat{\mb y}}
=
\bra{\mb y}\otimes\bra{\mb s}\otimes\bra{\mb t}
=
\bra{\mb{y,s,t}},
\]
where
\(\mb{x},\mb{w},\mb{z},\mb{y},\mb{s},\mb{t}\in\{0,1\}^n\).
For \(k=3\), the local sums take values in \(\{0,1,2,3\}\), so the phase-cancellation condition modulo \(4\) is equivalent to ordinary equality. Thus \cref{Eq:CmDef} becomes
\begin{equation}\label{Eq:CmDef3}
C^{(3)}
=
\Bigl\{
(\mb{x},\mb{w},\mb{z},\mb{y},\mb{s},\mb{t})
\in(\{0,1\}^n)^6
\;\Big|\;
x_l+w_l+z_l=y_l+s_l+t_l,\ \forall l\in[n]
\Bigr\},
\end{equation}
with \( |C^{(3)}|=20^n \). Indeed, at each qubit, there are
\(\sum_{r=0}^3\binom3r^2=20\) retained local matrix units.

Therefore, the single-qubit phase average retains exactly the
computational-basis matrix units indexed by \(C^{(3)}\). Equivalently,
\cref{Eq:momentHS} specializes to
\begin{equation}\label{eq:US-average-k3}
\mathbb{E}_{U_S}
\left(
U_S H^{\otimes n}\ket{\mb 0}\bra{\mb 0}H^{\otimes n}U_S^\dagger
\right)^{\otimes 3}
=
D^{-3}
\sum_{(\mb{x},\mb{w},\mb{z},\mb{y},\mb{s},\mb{t})\in C^{(3)}}
\ket{\mb{x,w,z}}\bra{\mb{y,s,t}}.
\end{equation}
We next average over the random \(CZ\) layer, still at the level of
computational-basis matrix units. Consider a single edge \((a,b)\). Since
\(CZ_{a,b}\ket{\mb u}=(-1)^{u_au_b}\ket{\mb u}\), the conjugation by
\(CZ_{a,b}^{\otimes3}\) acts on
\(\ket{\mb{x,w,z}}\bra{\mb{y,s,t}}\) by the relative phase $(-1)^{
x_ax_b+w_aw_b+z_az_b-y_ay_b-s_as_b-t_at_b
}.$
Thus \(T_{a,b}^{(3)}\in\{0,1\}\) is precisely the indicator that the edge \((a,b)\)
produces a relative sign \(-1\) on this matrix unit. In other words,
\(T_{a,b}^{(3)}=0\) means that the edge contributes no sign, while
\(T_{a,b}^{(3)}=1\) means that it contributes a factor \(-1\).

Averaging over the Bernoulli random variable deciding whether \(CZ_{a,b}\) is
present gives
\begin{equation}\label{eq:single-cz-average-k3}
\begin{split}
&{\mathbb E_{A_{a,b}}
\left(CZ_{a,b}^{A_{a,b}}\right)^{\otimes3}
\ket{\mb{x,w,z}}\bra{\mb{y,s,t}}
\left(CZ_{a,b}^{\dagger,A_{a,b}}\right)^{\otimes3}}
\\
&\qquad
=
(1-p)\ket{\mb{x,w,z}}\bra{\mb{y,s,t}}
+
p\,(-1)^{T_{a,b}^{(3)}}
\ket{\mb{x,w,z}}\bra{\mb{y,s,t}}
\\
&\qquad
=
\left(1-2p\,T_{a,b}^{(3)}\right)
\ket{\mb{x,w,z}}\bra{\mb{y,s,t}}.
\end{split}
\end{equation}
This is exactly the same single-edge mechanism as in the \(k=2\) proof: the
only difference is that the parity test now involves three copies instead of
two.

Since the different \(CZ\) edges are sampled independently, their scalar
contributions multiply. Therefore
\begin{equation}\label{eq:CZ-raw-third}
\mathbb E_{U_{CZ}}
U_{CZ}^{\otimes 3}
\ket{\mb{x,w,z}}\bra{\mb{y,s,t}}
\left(U_{CZ}^{\dagger}\right)^{\otimes 3}
=
\left(1-2p\right)^{
N^{(3)}_{\mb{x},\mb{w},\mb{z},\mb{y},\mb{s},\mb{t}}
}
\ket{\mb{x,w,z}}\bra{\mb{y,s,t}}.
\end{equation}
Finally, using \(p=\gamma\log n/n\), this becomes
\[
\mathbb E_{U_{CZ}}
U_{CZ}^{\otimes 3}
\ket{\mb{x,w,z}}\bra{\mb{y,s,t}}
\left(U_{CZ}^{\dagger}\right)^{\otimes 3}
=
\left(1-\frac{2\gamma\log n}{n}\right)^{
N^{(3)}_{\mb{x},\mb{w},\mb{z},\mb{y},\mb{s},\mb{t}}
}
\ket{\mb{x,w,z}}\bra{\mb{y,s,t}}.
\]
Substituting the \(U_S\)- and \(U_{CZ}\)-averaged expression into the definition
of the third moment operator, and leaving the local Clifford average explicit,
gives~\cref{eq:moment3-preblock}.
\end{proof}

The lemma above is the exact third-order analogue of the raw second-order
formula in~\cref{eq:channel-noise-raw}. It gives a computational-basis
expansion before any block regrouping is performed. The next step is to
organize the retained matrix units in \(C^{(3)}\) into local three-copy
Boolean-support blocks. For \(k=2\), this regrouping used the three local blocks
\(\Delta_2\), \(\id_4-\Delta_2\), and \(\sw_2-\Delta_2\). For \(k=3\), these are
replaced by the \(S_3\)-adapted family consisting of the computational-basis block
\(\Delta_3\) and the nine odd-even overlap blocks \(B(\pi_1,\pi_2)\).

\paragraph{Single-qubit three-copy Boolean-support blocks.}

We now describe the local blocks into which the retained single-qubit three-copy matrix units decompose. Recall from~\cref{ap:notation} the definition of the permutation group $S_3\coloneq \{\pi_{()},\pi_{(23)},\pi_{(12)},\pi_{(13)},\pi_{(123)},\pi_{(132)}\},$
we obtain that in the three-copy regime
\begin{equation}
S_{3,\mathrm{odd}}
:=
\{\pi_{(23)},\pi_{(12)},\pi_{(13)}\},
\qquad
S_{3,\mathrm{even}}
:=
\{\pi_{()},\pi_{(123)},\pi_{(132)}\}.
\end{equation}
For \(\pi\in S_3\), let \(V_1(\pi)\) denote the single-qubit three-copy permutation operator. In the computational basis, \(V_1(\pi)\) is a Boolean-support operator whose matrix units pair each three-bit string with its copy-permuted version.

Define the three-copy computational-basis block
\begin{equation}\label{eq:Delta3-def}
\Delta_3
:=
\ket{0,0,0}\bra{0,0,0}
+
\ket{1,1,1}\bra{1,1,1}.
\end{equation}
This block records the fully aligned local configurations. It is the common part of all single-qubit permutation supports \(V_1(\pi)\).

For every \(\pi_1\in S_{3,\mathrm{odd}}\) and
\(\pi_2\in S_{3,\mathrm{even}}\), define the computational-basis block
\begin{equation}\label{eq:B-pi1-pi2-def}
B(\pi_1,\pi_2)
:=
V_1(\pi_1)\cap V_1(\pi_2)-\Delta_3,
\end{equation}
where \(\cap\) is the Boolean-support intersection introduced in~\cref{eq:Boolean-support-notation}. Since
\(\Delta_3\) is contained in this intersection, 
\(B(\pi_1,\pi_2)\) consists precisely of the matrix units in the Boolean support of \(V_1(\pi_1)\) and 
\(V_1(\pi_2)\) that lie outside \(\Delta_3\).

\begin{lemma}[Single-qubit three-copy Boolean-support decomposition]\label{obs:B-partition-m3}
The ten Boolean-support blocks
\begin{equation}\label{eq:k3Blocks}
    \{\Delta_3\}
\cup
\{B(\pi_1,\pi_2):
\pi_1\in S_{3,\mathrm{odd}},\
\pi_2\in S_{3,\mathrm{even}}\}
\end{equation}
satisfy the following properties:
\begin{enumerate}
\item Each block contains exactly two computational-basis matrix units.
\item They are pairwise disjoint in Boolean support.
\item For every \(\pi_1\in S_{3,\mathrm{odd}}\), $V_1(\pi_1)
=
\Delta_3
+
\sum_{\pi_2\in S_{3,\mathrm{even}}}
B(\pi_1,\pi_2).$
\item For every \(\pi_2\in S_{3,\mathrm{even}}\), $V_1(\pi_2)
=
\Delta_3
+
\sum_{\pi_1\in S_{3,\mathrm{odd}}}
B(\pi_1,\pi_2).$
\item The nonzero entries of every \(B(\pi_1,\pi_2)\) form a subset
of the nonzero entries of both \(V_1(\pi_1)\) and \(V_1(\pi_2)\).
Consequently, each row and each column of \(B(\pi_1,\pi_2)\) contains at
most one nonzero entry, and $\|B(\pi_1,\pi_2)\|_\infty\le 1.$
\end{enumerate}
\end{lemma}
\begin{proof}
The condition
\(x_l+w_l+z_l=y_l+s_l+t_l\) in~\cref{Eq:CmDef3} retains
exactly the local computational-basis matrix units whose ket and bra labels
have the same Hamming weight. Enumerating the four possible weights gives
\[
\sum_{r=0}^{3}\binom{3}{r}^{2}=1+9+9+1=20
\]
retained units. The two units of weights \(0\) and \(3\) form \(\Delta_3\).
An exhaustive enumeration of the remaining \(18\) units over the
\(3\times3\) odd--even pairs shows that each belongs to exactly one
\(B(\pi_1,\pi_2)\), with two units in every such block. This establishes
properties 1 and 2.

Each single-qubit permutation support \(V_1(\pi)\) contains eight
computational-basis matrix units: the two fully aligned units in
\(\Delta_3\) and six non-fully aligned units. In the same enumeration, the
six latter units of a fixed odd permutation split into its three disjoint
two-unit intersections with the even permutations; the analogous statement
holds with odd and even interchanged. Summing these disjoint supports gives
properties 3 and 4.

Finally, by definition, the nonzero entries of
\(B(\pi_1,\pi_2)\)  form a subset of the nonzero entries
of each permutation matrix \(V_1(\pi_1)\) and \(V_1(\pi_2)\).
Since every row and every column of a permutation matrix contains exactly
one nonzero entry, every row and every column of \(B(\pi_1,\pi_2)\)
contains at most one nonzero entry. Consequently,
\(\|B(\pi_1,\pi_2)\|_\infty\le 1\), which proves property~5.
\end{proof}

\paragraph{Blockwise organization of the third-order moment operator.}

We now return to the computational-basis expression of the moment operator in~\cref{eq:moment3-preblock}
and organize its retained matrix units according to the local three-copy blocks introduced above. This is the third-order analogue of the \(k=2\) step, where each retained two-copy matrix unit was assigned to one of the three local blocks
\(\Delta_2\), \(\id_4-\Delta_2\), and \(\sw_2-\Delta_2\).
We write
\begin{equation}\label{eq:m3decom}
\ket{\mb{x,w,z}}\bra{\mb{y,s,t}}
=
\bigotimes_{l\in[n]}
\ket{x_l,w_l,z_l}\bra{y_l,s_l,t_l}.
\end{equation}
For each site \(l\), the local matrix unit
\(\ket{x_l,w_l,z_l}\bra{y_l,s_l,t_l}\)
belongs to exactly one of the ten Boolean-support blocks in~\cref{eq:k3Blocks},
by~\cref{obs:B-partition-m3}. Therefore every retained global matrix unit in
\(C^{(3)}\) belongs to a unique tensor-product block.

We encode this assignment by the \emph{position data}
\begin{equation}\label{eq:postion-data}
\mathfrak I
:=
\Bigl(
I^{(0)},
\{I^{(\pi_1,\pi_2)}\}_{\pi_1\in S_{3,\mathrm{odd}},\,\pi_2\in S_{3,\mathrm{even}}}
\Bigr),
\end{equation}
where $I^{(0)}
\sqcup
\Bigl(
\bigsqcup_{\pi_1\in S_{3,\mathrm{odd}},\,\pi_2\in S_{3,\mathrm{even}}}
I^{(\pi_1,\pi_2)}
\Bigr)
=
[n].$
Here \(I^{(0)}\) records the sites where the local matrix unit belongs to
\(\Delta_3\), while \(I^{(\pi_1,\pi_2)}\) records the sites where it belongs to
\(B(\pi_1,\pi_2)\). The corresponding tensor-product block is
\begin{equation}\label{eq:def-B-frakI-m3}
B_{\mathfrak I}
:=
\Delta_3^{\otimes I^{(0)}}
\otimes
\bigotimes_{\pi_1\in S_{3,\mathrm{odd}},\,\pi_2\in S_{3,\mathrm{even}}}
B(\pi_1,\pi_2)^{\otimes I^{(\pi_1,\pi_2)}},
\end{equation}
and each retained $n$-qubit computational-basis matrix unit
$\ket{\mb{x,w,z}}\bra{\mb{y,s,t}}$ belongs to a unique tensor-product Boolean-support block $B_{\mathfrak I}$.

We now translate \(T_{a,b}^{(3)}\) from~\cref{eq:defT3}
into the language of the local block labels encoded by \(B_{\mathfrak I}\).
The key point is that, once the position data \(\mathfrak I\) is fixed, all
global computational-basis matrix units $\ket{\mb{x,w,z}}\bra{\mb{y,s,t}}$ in the Boolean support of $ B_{\mathfrak I}$  have the same exponent $N^{(3)}_{\mb{x},\mb{w},\mb{z},\mb{y},\mb{s},\mb{t}}$, and this exponent can be computed purely from the block sizes.

\begin{lemma}[The exponent
\(N^{(3)}_{\mb{x},\mb{w},\mb{z},\mb{y},\mb{s},\mb{t}}\) for three-copy
blocks]
\label{lem:N-frakI-blockwise}
Fix position data $\mathfrak I$
and the corresponding tensor-product block \(B_{\mathfrak I}\) in~\cref{eq:def-B-frakI-m3}.
Then every retained computational-basis matrix unit $\ket{\mb{x,w,z}}\bra{\mb{y,s,t}}$
in the Boolean support of \(B_{\mathfrak I}\) has the same value of
\(N^{(3)}_{\mb{x},\mb{w},\mb{z},\mb{y},\mb{s},\mb{t}}\). This common value is
\begin{equation}\label{eq:def-N-frakI-m3}
N_{\mathfrak I}
=
\frac12
\sum_{\pi_1\in S_{3,\mathrm{odd}},\,\pi_2\in S_{3,\mathrm{even}}}
|I^{(\pi_1,\pi_2)}|
\sum_{\substack{
\pi_1'\in S_{3,\mathrm{odd}}\setminus\{\pi_1\}\\
\pi_2'\in S_{3,\mathrm{even}}\setminus\{\pi_2\}}}
|I^{(\pi_1',\pi_2')}|.
\end{equation}
\end{lemma}

\begin{proof}
By~\cref{eq:defT3}, the exponent
\(N^{(3)}_{\mb{x},\mb{w},\mb{z},\mb{y},\mb{s},\mb{t}}\) is the sum of the edge
indicators \(T_{a,b}^{(3)}\) over all pairs \(a<b\). Hence, it suffices to decide
\(T_{a,b}^{(3)}\) from the two local block labels at sites \(a\) and \(b\).

For retained one-qubit matrix units
\[
E_a=\ket{x_a,w_a,z_a}\bra{y_a,s_a,t_a},
\qquad
E_b=\ket{x_b,w_b,z_b}\bra{y_b,s_b,t_b},
\]
the following finite local identity holds:
\begin{equation}\label{eq:k3-local-compatibility}
T_{a,b}^{(3)}=0
\quad\Longleftrightarrow\quad
\text{there exists \(\pi\in S_3\) such that }
E_a,E_b\in V_1(\pi).
\end{equation}
Indeed, the phase constraint leaves \(20\) possible matrix units at each
site. Exhaustively evaluating the parity in~\cref{eq:defT3} on these
\(20^2\) pairs gives exactly the common-permutation-support criterion in
\cref{eq:k3-local-compatibility}. Organizing this finite table by the ten
blocks in~\cref{eq:k3Blocks} yields the blockwise rule used below.

Now use the block partition in~\cref{obs:B-partition-m3}. A local matrix unit in
\(\Delta_3\) belongs to every permutation support \(V_1(\pi)\). Therefore, if
either site \(a\) or site \(b\) lies in \(I^{(0)}\), then the two local matrix
units share a common permutation support and \(T_{a,b}^{(3)}=0\).

It remains to consider two matrix units from operator blocks. Suppose $
a\in I^{(\pi_1,\pi_2)},
b\in I^{(\pi_1',\pi_2')}.$
A local matrix unit in \(B(\pi_1,\pi_2)\) belongs precisely to the two
permutation supports \(V_1(\pi_1)\) and \(V_1(\pi_2)\). Hence the two sites
share a common permutation support \textit{if and only if} $\pi_1=\pi_1'
\text{ or }
\pi_2=\pi_2'.$
Equivalently,
\[
T_{a,b}^{(3)}=1
\quad\Longleftrightarrow\quad
\pi_1\neq\pi_1'
\ \text{ and }\
\pi_2\neq\pi_2'.
\]
Thus \(T_{a,b}^{(3)}\) depends only on the block labels of sites \(a\) and \(b\),
not on the particular matrix units chosen inside those blocks. Consequently,
the total exponent is constant on the Boolean support of \(B_{\mathfrak I}\).

For each block type
\((\pi_1,\pi_2)\), there are \(|I^{(\pi_1,\pi_2)}|\) choices for the first site,
and the second site must lie in a block
\((\pi_1',\pi_2')\) with
\(\pi_1'\neq\pi_1\) and \(\pi_2'\neq\pi_2\). Summing over all block types counts
each unordered pair twice, so we divide by \(2\). This gives the stated formula
for \(N_{\mathfrak I}\).
\end{proof}

Since the expression in~\cref{eq:def-N-frakI-m3} depends only on the cardinalities
of the position sets, we pass from position data to \textit{count data}:
\begin{equation}\label{eq:defI}
\mb I
:=
\Bigl[
|I^{(0)}|;
\,|I^{(\pi_1,\pi_2)}|
\text{ for all }
(\pi_1,\pi_2)\in S_{3,\mathrm{odd}}\times S_{3,\mathrm{even}}
\Bigr].
\end{equation}
We call such count data admissible. Equivalently, \(\mb I\) is a
nonnegative ten-component vector whose entries sum to \(n\). All sums over \(\mb I\) below range over admissible count data.

For fixed count data \(\mb I\), let $\mc C_{\mb I}
:=
\{B_{\mathfrak I}:\ \mathfrak I \text{ has count data } \mb I\}.$
If \(\mathfrak I\) has count data \(\mb I\), then the value
\(N_{\mathfrak I}\) in~\cref{eq:def-N-frakI-m3} depends only on \(\mb I\), not
on the actual positions of the sets in \(\mathfrak I\). We denote this common
value by \(N_{\mb I}\). 
\begin{equation}\label{eq:def-N-I-m3}
N_{\mb I}
=
\frac12
\sum_{\pi_1\in S_{3,\mathrm{odd}},\,\pi_2\in S_{3,\mathrm{even}}}
|I^{(\pi_1,\pi_2)}|
\sum_{\substack{
\pi_1'\in S_{3,\mathrm{odd}}\setminus\{\pi_1\}\\
\pi_2'\in S_{3,\mathrm{even}}\setminus\{\pi_2\}}}
|I^{(\pi_1',\pi_2')}|.
\end{equation}
Equivalently, \(N_{\mb I}\) is given by the same counting
formula as \(N_{\mathfrak I}\), now written only in terms of the entries of the
count data \(\mb I\).

Using this blockwise value of the exponent in the computational-basis formula
\cref{eq:moment3-preblock}, we obtain
\begin{equation}\label{eq:moment3-block-organized}
\mb M_{\mc E_{\mathrm{shallow}}^{(\gamma)}}^{(3)}
=
D^{-3}\,
\mathbb E_{U_{\mathrm{loc}}}
U_{\mathrm{loc}}^{\otimes 3}
\Biggl[
\sum_{\mb I}
\left(1-\frac{2\gamma\log n}{n}\right)^{N_{\mb I}}
\sum_{B_{\mathfrak I}\in \mc C_{\mb I}}
B_{\mathfrak I}
\Biggr]
\left(U_{\mathrm{loc}}^\dagger\right)^{\otimes 3}.
\end{equation}

\paragraph{Applying the local Clifford twirling.}

It remains to average over the local Clifford layer. This step is completely
parallel to the \(k=2\) case: since
\(U_{\mathrm{loc}}=\bigotimes_{q\in[n]}U_q\) is sampled independently across
qubits, the local Clifford twirl factorizes site by site. We use the
single-qubit three-copy channel \(\mc T_3\) defined in
\cref{ap:notation}.
We write  $\widetilde\Delta_3
:=
\mc T_3(\Delta_3),
\widetilde B(\pi_1,\pi_2)
:=
\mc T_3\!\left(B(\pi_1,\pi_2)\right).$ Let
\(\Pi_{\mathrm{sym},1}^{(3)}
:=\frac16\sum_{\pi\in S_3}V_1(\pi)\)
denote the projector onto the fully symmetric subspace of the single-qubit
three-copy space.
The following lemma records the explicit single-qubit twirled blocks.
\begin{lemma}[Single-qubit Clifford twirls of the three-copy blocks]
\label{obs:k3-single-site-twirls}
The twirled single-qubit three-copy blocks are
\begin{equation}\label{eq:twirl-Delta3}
\widetilde{\Delta}_3
=
\frac12\Pi_{\mathrm{sym},1}^{(3)}
=
\frac{1}{12}\sum_{\pi\in S_3}V_1(\pi), \qquad \widetilde B(\pi_1,\pi_2)
=
\frac13V_1(\pi_1)+\frac13V_1(\pi_2)-\widetilde\Delta_3,
\end{equation}
for every odd-even pair
\((\pi_1,\pi_2)\in S_{3,\mathrm{odd}}\times S_{3,\mathrm{even}}\).
\end{lemma}

\begin{proof}
We first evaluate the twirl of \(\Delta_3\). Since the single-qubit
Clifford group is an exact unitary \(3\)-design, its three-copy twirling
coincides with the standard Haar moment formula, or
equivalently Weingarten calculus~\cite{mele2024introduction}, gives
\begin{equation}
\widetilde\Delta_3
=
\mc T_3(\Delta_3)
=
\frac12\Pi_{\mathrm{sym},1}^{(3)}
=
\frac1{12}\sum_{\pi\in S_3}V_1(\pi).
\label{eq:proof-Delta3-twirl}
\end{equation}

We next evaluate the twirl of \(B(\pi_1,\pi_2)\). Let
\(s=(23)\in S_3\). At the level of computational-basis support, one obtains $V_1(e)\cap V_1(s) = \id_2\otimes\Delta_2,$
where \(\Delta_2\) acts on the second and third copies.

For an odd-even pair \((\pi_1,\pi_2)\), the relative permutation
\(\pi_1^{-1}\pi_2\) is odd and hence is a transposition. Since all transpositions in \(S_3\) are conjugate, there exists \(r\in S_3\) such that $r^{-1}sr=\pi_1^{-1}\pi_2$. Define $P_R:=V_1(r),$ and $
P_L:=V_1(\pi_1r^{-1}).$
These permutation operators satisfy $P_LP_R=V_1(\pi_1), P_LV_1(s)P_R=V_1(\pi_2),$ respectively.

Left and right multiplication by permutation matrices merely permutes
the rows and columns of a Boolean-support operator and therefore
preserves elementwise intersections. We then obtain
\begin{equation}
V_1(\pi_1)\cap V_1(\pi_2)=
\bigl[P_LV_1(e)P_R\bigr]
\cap
\bigl[P_LV_1(s)P_R\bigr]=
P_L
\bigl[V_1(e)\cap V_1(s)\bigr]
P_R
=
P_L(\id_2\otimes\Delta_2)P_R.
\label{eq:B-permutation-dressed}
\end{equation}

According to ~\cref{eq:B-pi1-pi2-def},
\(B(\pi_1,\pi_2)+\Delta_3
=V_1(\pi_1)\cap V_1(\pi_2)\).
Combining this identity with~\cref{eq:B-permutation-dressed} gives
\begin{equation}
B(\pi_1,\pi_2)+\Delta_3
=
P_L(\id_2\otimes\Delta_2)P_R.
\label{eq:B-plus-Delta3-permutation-dressed}
\end{equation}
Three-copy permutations commute with \(C^{\otimes3}\), so they can be pulled
through the Clifford twirling. By the linearity of \(\mc T_3\), we therefore
have
\begin{equation}
\begin{split}
\widetilde B(\pi_1,\pi_2)+\widetilde\Delta_3
=\mc T_3\!\left(B(\pi_1,\pi_2)+\Delta_3\right)
=P_L\,\mc T_3(\id_2\otimes\Delta_2)\,P_R.
\end{split}
\label{eq:twirl-pull-through-permutations}
\end{equation}

The first copy is unaffected, while the second and third copies reduce
to the second-order Clifford twirl. Using
\(\widetilde\Delta_2=(\id_4+\sw_2)/3\), we have
\begin{equation}
\mc T_3(\id_2\otimes\Delta_2)=
\id_2\otimes\mc T_2(\Delta_2)=
\frac13\left[V_1(e)+V_1(s)\right].
\label{eq:canonical-block-twirl}
\end{equation}

Combining
~\cref{eq:twirl-pull-through-permutations}
and~\cref{eq:canonical-block-twirl} gives
\begin{equation}
\widetilde B(\pi_1,\pi_2)+\widetilde\Delta_3=
\frac13P_L
\left[V_1(e)+V_1(s)\right]
P_R=
\frac13V_1(\pi_1)
+
\frac13V_1(\pi_2).
\label{eq:B-plus-Delta3-twirl-proof}
\end{equation}

Finally, subtracting \(\widetilde\Delta_3\) from both sides
of~\cref{eq:B-plus-Delta3-twirl-proof} yields
\begin{equation}
\widetilde B(\pi_1,\pi_2)
=
\frac13V_1(\pi_1)
+
\frac13V_1(\pi_2)
-
\widetilde\Delta_3,
\end{equation}
which completes the proof.
\end{proof}

Thus, the local Clifford twirl replaces each untwirled local block by its
twirled counterpart. In particular, for every tensor-product block
\(B_{\mathfrak I}\), we have
\begin{equation}
\mathbb E_{U_{\mathrm{loc}}}
\,U_{\mathrm{loc}}^{\otimes3}
B_{\mathfrak I}
\left(U_{\mathrm{loc}}^\dagger\right)^{\otimes3}
=
\widetilde B_{\mathfrak I},
\end{equation}
where $\widetilde B_{\mathfrak I}
:=
\widetilde\Delta_3^{\otimes I^{(0)}}
\otimes
\bigotimes_{\pi_1\in S_{3,\mathrm{odd}},\,\pi_2\in S_{3,\mathrm{even}}}
\widetilde B(\pi_1,\pi_2)^{\otimes I^{(\pi_1,\pi_2)}}.$
For fixed count data \(\mb I\), define $\widetilde{\mc C}_{\mb I}
:=
\{\widetilde B_{\mathfrak I}:\ \mathfrak I \text{ has count data } \mb I\}.$
Applying the local Clifford twirling to the block-organized expression
in~\cref{eq:moment3-block-organized}, and noting that the scalar coefficient
only depends on the count data \(\mb I\), we obtain
\begin{equation}\label{eq:block-expansion-twirled}
\mb M_{\mc E_{\mathrm{shallow}}^{(\gamma)}}^{(3)}
=
D^{-3}
\sum_{\mb I}
\left(1-\frac{2\gamma\log n}{n}\right)^{N_{\mb I}}
\sum_{\widetilde B_{\mathfrak I}\in \widetilde{\mc C}_{\mb I}}
\widetilde B_{\mathfrak I}.
\end{equation}
This is the desired third-order moment formula. It has the same structural form
as the second-order formula: after the computational-basis expression is
organized into local Boolean-support blocks, the local Clifford twirl replaces
each block by its twirled version. The new third-order feature is that the three
two-copy blocks are replaced by the \(S_3\)-adapted family
\(\Delta_3\) and \(B(\pi_1,\pi_2)\), while the second-order exponent
\(|I_j||I_k|\) is replaced by the block-count exponent \(N_{\mb I}\).

Equation~\cref{eq:block-expansion-twirled} is the general
block expansion for \(p=\gamma\log n/n\). Taking instead \(p=1/2\), which
corresponds to the dense phase ensemble, and adopting the convention \(0^0=1\), the
factor \((1-2p)^{N_{\mb I}}\) vanishes unless \(N_{\mb I}=0\). Therefore,
\begin{equation}\label{eq:block-expansion-phase}
\mb M_{\mc E_{\mathrm{phase}}}^{(3)}
=
D^{-3}
\sum_{\substack{\mb I\\N_{\mb I}=0}}
\sum_{\widetilde B_{\mathfrak I}\in\widetilde{\mc C}_{\mb I}}
\widetilde B_{\mathfrak I}.
\end{equation}
Consequently,
\begin{equation}\label{eq:block-expansion-shallow-phase-gap}
\mb M_{\mc E_{\mathrm{shallow}}^{(\gamma)}}^{(3)}
-
\mb M_{\mc E_{\mathrm{phase}}}^{(3)}
=
D^{-3}
\sum_{\substack{\mb I,  N_{\mb I}>0}}
\left(1-\frac{2\gamma\log n}{n}\right)^{N_{\mb I}}
\sum_{\widetilde B_{\mathfrak I}\in\widetilde{\mc C}_{\mb I}}
\widetilde B_{\mathfrak I}.
\end{equation}

\subsection{Moment operators of the dense phase ensemble}
The dense phase ensemble $\mc E_{\mathrm{phase}}$ is the special case of
$\mc E_{\mathrm{shallow}}^{(\gamma)}$ obtained by setting $p=\frac12$, equivalently $1-2p=0$. Therefore, its low-order moment operators can be read off directly from the general formulas derived above.

\begin{prop}[Second and third moments of $\mc E_{\mathrm{phase}}$]
\label{prop:phase-moments}
The dense phase ensemble satisfies
\begin{equation}\label{eq:phase-second-moment-app}
\mb{M}_{\mc{E}_{\mathrm{phase}}}^{(2)}
=
2^{-2n}
\left(
\id_4^{\otimes n}
+
\sw_2^{\otimes n}
-
\widetilde{\Delta}_2^{\otimes n}
\right),
\end{equation}
and
\begin{equation}\label{eq:phase-third-moment-app}
\mb{M}_{\mc{E}_{\mathrm{phase}}}^{(3)}
=
2^{-3n}\left[
\sum_{\pi\in S_3}V_n(\pi)
-\frac{1}{3^n}
\sum_{\substack{\pi_1\in S_{3,\mathrm{odd}}\\ \pi_2\in S_{3,\mathrm{even}}}}
\bigl(V_1(\pi_1)+V_1(\pi_2)\bigr)^{\otimes n}
+\frac{4}{12^n}\Bigl(\sum_{\pi\in S_3}V_1(\pi)\Bigr)^{\otimes n}
\right].
\end{equation}
\end{prop}
\begin{proof}
From the general second-moment formula for $\mc E_{\mathrm{shallow}}^{(\gamma)}$,
setting $p=\frac12$ gives $1-\frac{2\gamma\log n}{n}=1-2p=0.$
Hence only the terms with $|I_j||I_k|=0$ survive, namely the cases $I_j=\emptyset$ or $I_k=\emptyset$. Therefore, with the convention \(0^0=1\),
\begin{equation}
\begin{split}
\mb{M}_{\mc E_{\mathrm{phase}}}^{(2)}
&=
2^{-2n}
\sum_{I_i\sqcup I_j\sqcup I_k=[n]}
0^{|I_j||I_k|}
\widetilde{\Delta}_2^{\otimes I_i}
\otimes
(\id_4-\widetilde{\Delta}_2)^{\otimes I_j}
\otimes
(\sw_2-\widetilde{\Delta}_2)^{\otimes I_k}\\
&=
2^{-2n}
\left[
\sum_{I_i\sqcup I_j=[n]}
\widetilde{\Delta}_2^{\otimes I_i}\otimes(\id_4-\widetilde{\Delta}_2)^{\otimes I_j}
+
\sum_{I_i\sqcup I_k=[n]}
\widetilde{\Delta}_2^{\otimes I_i}\otimes(\sw_2-\widetilde{\Delta}_2)^{\otimes I_k}
-
\widetilde{\Delta}_2^{\otimes n}
\right]\\
&=
2^{-2n}
\left(
\id_4^{\otimes n}
+
\sw_2^{\otimes n}
-
\widetilde{\Delta}_2^{\otimes n}
\right),
\end{split}
\end{equation}
which proves \cref{eq:phase-second-moment-app}. The term \(\widetilde{\Delta}_2^{\otimes n}\) is subtracted because the overlap
\(I_j=I_k=\emptyset\), equivalently \(I_i=[n]\), is counted in both resulting
sums.

To characterize the terms with zero interaction exponent directly,
arrange the nine nontrivial blocks \(B(\pi_1,\pi_2)\) in a \(3\times3\)
array indexed by odd permutations in the rows and even permutations in the
columns. By~\cref{lem:N-frakI-blockwise,eq:k3-local-compatibility},
\(N_{\mb I}=0\) exactly when every two occupied nontrivial blocks lie in a
common row or a common column. Any pairwise compatible family of cells in
this array is contained in one row or one column. Properties~3 and~4
of~\cref{obs:B-partition-m3}, together with the fact that
\(\Delta_3\) is contained in every permutation support, then show that every
retained global matrix unit with \(N_{\mb I}=0\) belongs to one common
permutation support \(V_n(\pi)\). The converse follows immediately from the
same local compatibility relation. This gives the following Boolean-support identity, in agreement with Supplementary Note 3(A) of Ref.~\cite{zhang2025robust}
\begin{equation}\label{eq:phase-third-union-cited}
\sum_{\substack{
(\mb{x},\mb{w},\mb{z},\mb{y},\mb{s},\mb{t})\in C^{(3)},
N^{(3)}_{\mb{x},\mb{w},\mb{z},\mb{y},\mb{s},\mb{t}}=0
}}
\ket{\mb{x,w,z}}\bra{\mb{y,s,t}}
=
\bigcup_{\pi\in S_3} V_n(\pi),
\end{equation}
where the union is understood at the level of Boolean support. Equivalently,
\(N_{\mb{x},\mb{w},\mb{z},\mb{y},\mb{s},\mb{t}}^{(3)}=0\) holds precisely for the retained matrix units that belong to at
least one global three-copy permutation support \(V_n(\pi)\). Using~\cref{eq:phase-third-union-cited} in the raw third-moment formula
\cref{eq:moment3-preblock}, and noting that \(1-2p=0\) when \(p=\frac12\) with the convention \(0^0=1\), we get
\begin{equation}\label{eq:phase-third-untwirled-union}
\mb M_{\mc E_{\mathrm{phase}}}^{(3)}=
2^{-3n}
\mathbb E_{U_{\mathrm{loc}}}
U_{\mathrm{loc}}^{\otimes3}
\left[
\sum_{\substack{
(\mb{x},\mb{w},\mb{z},\mb{y},\mb{s},\mb{t})\in C^{(3)},
N^{(3)}_{\mb{x},\mb{w},\mb{z},\mb{y},\mb{s},\mb{t}}=0
}}
\ket{\mb{x,w,z}}\bra{\mb{y,s,t}}
\right]
\left(U_{\mathrm{loc}}^\dagger\right)^{\otimes3}=
2^{-3n}
\mathbb E_{U_{\mathrm{loc}}}
U_{\mathrm{loc}}^{\otimes3}
\left[
\bigcup_{\pi\in S_3}V_n(\pi)
\right]
\left(U_{\mathrm{loc}}^\dagger\right)^{\otimes3}.
\end{equation}

Following the Boolean-support inclusion--exclusion decomposition used in
Ref.~\cite{zhang2025robust}, we write it in the block notation introduced
above:
\begin{equation}\label{eq:phase-third-untwirled-IE}
\bigcup_{\pi\in S_3}V_n(\pi)
=
\sum_{\pi\in S_3}V_n(\pi)
-
\sum_{\substack{\pi_1\in S_{3,\mathrm{odd}}\\  \pi_2\in S_{3,\mathrm{even}}}}
[B(\pi_1,\pi_2)+\Delta_3]^{\otimes n}
+
4\,\Delta_3^{\otimes n}.
\end{equation}

Substituting~\cref{eq:phase-third-untwirled-IE} into
\cref{eq:phase-third-untwirled-union}, we get
\begin{equation}\label{eq:phase-third-after-IE-before-twirl}
\begin{split}
\mb M_{\mc E_{\mathrm{phase}}}^{(3)}
=
2^{-3n}
\mathbb E_{U_{\mathrm{loc}}}
U_{\mathrm{loc}}^{\otimes3}
\Biggl[
&
\sum_{\pi\in S_3}V_n(\pi)
-
\sum_{\substack{\pi_1\in S_{3,\mathrm{odd}}\\ \pi_2\in S_{3,\mathrm{even}}}}
[B(\pi_1,\pi_2)+\Delta_3]^{\otimes n}
+
4\,\Delta_3^{\otimes n}
\Biggr]
\left(U_{\mathrm{loc}}^\dagger\right)^{\otimes3}.
\end{split}
\end{equation}
We now apply the local Clifford twirl term by term. The permutation operators
\(V_n(\pi)\) are invariant under this twirl, since they commute with
\(U_{\mathrm{loc}}^{\otimes3}\). Moreover, because
\(U_{\mathrm{loc}}=\bigotimes_{q\in[n]}U_q\) is sampled independently across
qubits, the twirl factorizes over tensor products. Hence, 
\begin{equation}\label{eq:phase-third-before-final-substitution}
\mb M_{\mc E_{\mathrm{phase}}}^{(3)}
=
2^{-3n}
\left[
\sum_{\pi\in S_3}V_n(\pi)
-
\sum_{\substack{\pi_1\in S_{3,\mathrm{odd}}\\ \pi_2\in S_{3,\mathrm{even}}}}
[\widetilde B(\pi_1,\pi_2)+\widetilde\Delta_3]^{\otimes n}
+
4\,\widetilde\Delta_3^{\otimes n}
\right].
\end{equation}
Finally, by~\cref{obs:k3-single-site-twirls},
\[
\widetilde B(\pi_1,\pi_2)+\widetilde\Delta_3
=
\frac13\bigl(V_1(\pi_1)+V_1(\pi_2)\bigr),
\qquad
\widetilde\Delta_3
=
\frac1{12}\sum_{\pi\in S_3}V_1(\pi).
\]
Substituting these two identities into
\cref{eq:phase-third-before-final-substitution}, we obtain
\begin{equation}
\mb{M}_{\mc{E}_{\mathrm{phase}}}^{(3)}
=
2^{-3n}\left[
\sum_{\pi\in S_3}V_n(\pi)
-\frac{1}{3^n}
\sum_{\substack{\pi_1\in S_{3,\mathrm{odd}}\\ \pi_2\in S_{3,\mathrm{even}}}}
\bigl(V_1(\pi_1)+V_1(\pi_2)\bigr)^{\otimes n}
+\frac{4}{12^n}\Bigl(\sum_{\pi\in S_3}V_1(\pi)\Bigr)^{\otimes n}
\right],
\end{equation}
which proves~\cref{eq:phase-third-moment-app}.
\end{proof}

\section{The second-order relative-error analysis}\label{Ap:2twirled}
In this section, we study the case $k=2$. Throughout, \(D=2^n\), \(\gamma>0\)
is fixed independently of \(n\), and \(n\) is sufficiently large that
\[
0<p=\frac{\gamma\log n}{n}<\frac12.
\]
In particular, \(1-2p\in(0,1)\), as required in the estimates below.
By \cref{lem:relative-error-opnorm} and the triangle inequality, the relative-error analysis is reduced to controlling the two operator-norm gaps
\begin{equation}\label{eq:k2-2moment}
    \bigl\|\mb M_{\mc E_{\mathrm{shallow}}^{(\gamma)}}^{(2)}
-\mb M_{\mc E_{\mathrm{phase}}}^{(2)}\bigr\|_\infty
\quad\text{and}\quad
\bigl\|\mb M_{\mc E_{\mathrm{phase}}}^{(2)}
-\mb M_{\mc E_{\mathrm{Haar}}}^{(2)}\bigr\|_\infty.
\end{equation}

The derivation is organized as follows. In \cref{ap:k2-prelim}, we collect two mathematical lemmas that will be used later. In \cref{ap:k2-diag}, we show that the three moment operators appearing in~\cref{eq:k2-2moment} can be simultaneously diagonalized in a common tensor-product basis. This reduces the operator-norm comparison to an eigenvalue comparison. We then bound $\bigl\|\mb M_{\mc E_{\mathrm{phase}}}^{(2)}
-\mb M_{\mc E_{\mathrm{Haar}}}^{(2)}\bigr\|_\infty$ in \cref{ap:k2-phase-haar}, and $\bigl\|\mb M_{\mc E_{\mathrm{shallow}}^{(\gamma)}}^{(2)}
-\mb M_{\mc E_{\mathrm{phase}}}^{(2)}\bigr\|_\infty$ in \cref{ap:k2-shallow-phase}.

\subsection{Mathematical details}\label{ap:k2-prelim}

In this subsection, we introduce two technical lemmas that will be used repeatedly in the proof below.
\begin{lemma}[Combinatorial summation bound]\label{lem:comb_bound}
For \(0<a<1\) and any integer \(r\ge 1\), the following inequality holds:
\begin{equation}
    \sum_{j=0}^{r} \binom{r}{j} a^{j(r-j)}
    \le
    2\left(1+a^{r/2}\right)^r.
\end{equation}
\end{lemma}

\begin{proof}
Note that the terms in the summation are symmetric with respect to \(j\) and
\(r-j\). We can bound the sum by doubling the first half of the terms:
\begin{equation}
\sum_{j=0}^{r} \binom{r}{j} a^{j(r-j)}
\le
2 \sum_{j=0}^{\lfloor r/2 \rfloor} \binom{r}{j} a^{j(r-j)}.
\end{equation}

Consider the range \(0\le j\le r/2\). In this range, we have \(r-j\ge r/2\).
Since \(0<a<1\), a larger exponent gives a smaller value. Therefore, we have $a^{j(r-j)}
\le
a^{jr/2}
=
\left(a^{r/2}\right)^j.$

Substituting this bound into the partial sum gives
\begin{equation}
\begin{split}
2 \sum_{j=0}^{\lfloor r/2 \rfloor} \binom{r}{j} a^{j(r-j)}
&\le
2 \sum_{j=0}^{\lfloor r/2 \rfloor} \binom{r}{j}
\left(a^{r/2}\right)^j \le
2 \sum_{j=0}^{r} \binom{r}{j}
\left(a^{r/2}\right)^j
=
2\left(1+a^{r/2}\right)^r,
\end{split}
\end{equation}
where the last equality follows from the binomial theorem.
\end{proof}

\begin{lemma}[Sub-exponential bounds]\label{lem:sub_exp}
Let \(\gamma>0\) be a constant. For every integer \(1\le r\le n\), the
following bounds hold for sufficiently large \(n\):
\begin{enumerate}
\item[(A)]
Uniformly over \(1\le r\le n\),
\begin{equation}\label{eq:sub-exp-uniform}
\left(1+n^{-\frac{\gamma r}{n}}\right)^r
\le
\exp\!\left(\frac{n}{\gamma e\log n}\right)
=
\exp(o(n)).
\end{equation}

\item[(B)]
If \(r\ge n/4\) and \(\gamma>4\), then
\begin{equation}\label{eq:sub-exp-large-r}
\left(1+n^{-\frac{\gamma r}{n}}\right)^r
\le
1+\frac12 n^{1-\gamma/4}.
\end{equation}
\end{enumerate}
\end{lemma}

\begin{proof}

Using the inequality $1+x\le e^x$, we obtain $(1+n^{-\frac{\gamma r}{n}})^r\le \exp(r n^{-\frac{\gamma r}{n}})$. Hence,
we consider the function $f(r):= r\,n^{-\frac{\gamma r}{n}}
$. Let $a:=\frac{\gamma\log n}{n}>0$. Then $f(r)= r e^{-ar}$ and $f'(r)=e^{-ar}(1-ar)$. Viewing \(f\) as a function on \((0,\infty)\), its maximum is attained
at \(r^\star=1/a=n/(\gamma\log n)\), and
\begin{equation}
f(r)\le f(r^\star)=\frac{1}{a}e^{-1}=\frac{n}{\gamma(\log n)}\cdot \frac1e.
\end{equation}
Therefore, for all $r\ge 1$,
\begin{equation}
(1+n^{-\frac{\gamma r}{n}})^r\le \exp(f(r))\le
\exp\!\Big(\frac{n}{\gamma e\log n}\Big),
\end{equation}
which is sub-exponential in $n$.

Moreover, for sufficiently large \(n\),
\(n/4\ge r^\star=n/(\gamma\log n)\). Hence \(f\), and therefore
\(\exp(f)\), is decreasing on \([n/4,n]\). Thus, whenever \(r\ge n/4\),
\[\exp(r n^{-\frac{\gamma r}{n}})\le \exp(\frac 14 n^{1-\gamma/4}).\] 
As $e^x\le 1+2x$ when $0<x\le 1$, we obtain that $\exp(\frac 14 n^{1-\gamma/4})\leq 1+\frac 12 n^{1-\gamma/4}$ when $\gamma>4$ for sufficiently large $n$.
\end{proof}

\subsection{Second-moment formulas and simultaneous diagonalization}\label{ap:k2-diag}

The goal of this subsection is to turn the calculation of the second-order relative error into a scalar eigenvalue comparison. By~\cref{lem:relative-error-opnorm}, the
relative-error estimate is reduced to bounding operator-norm gaps between moment operators. Directly estimating these operator norms is inconvenient. Instead, we show that the relevant second-moment operators are simultaneously diagonalizable in a common tensor-product basis. Once this is done, each operator-norm gap is
the maximum absolute difference between the corresponding eigenvalues. The rest of the second-order analysis will therefore be reduced to bounding explicit scalar expressions.

We implement this reduction in two steps: first, we identify a common diagonalizing basis for the three
second-moment operators $\mb M_{\mc E_{\mathrm{shallow}}^{(\gamma)}}^{(2)},
\mb M_{\mc E_{\mathrm{phase}}}^{(2)},
\mb M_{\mc E_{\mathrm{Haar}}}^{(2)}$. Second, we compute the eigenvalues of the corresponding moment operators in this basis.

We first recall the three second-moment formulas. From~\cref{lem:sparse-second-moment}, the
second moment of the shallow phase ensemble is
\begin{equation}\label{eq:k2-shallow-moment}
\mb{M}_{\mc E_{\mathrm{shallow}}^{(\gamma)}}^{(2)}
=
2^{-2n}
\sum_{I_i\sqcup I_j\sqcup I_k=[n]}
\left(1-\frac{2\gamma \log n}{n}\right)^{|I_j||I_k|}
\widetilde{\Delta}_2^{\otimes I_i}
\otimes
(\id_4-\widetilde{\Delta}_2)^{\otimes I_j}
\otimes
(\sw_2-\widetilde{\Delta}_2)^{\otimes I_k}.
\end{equation}
Likewise, from~\cref{prop:phase-moments}, the second moment of the dense phase ensemble is
\begin{equation}\label{eq:k2-phase-moment}
\mb{M}_{\mc E_{\mathrm{phase}}}^{(2)}
=
2^{-2n}
\left(
\id_4^{\otimes n}
+
\sw_2^{\otimes n}
-
\widetilde{\Delta}_2^{\otimes n}
\right).
\end{equation}
Finally, the Haar second moment is
\begin{equation}\label{eq:k2-haar-moment}
\mb M_{\mc E_{\mathrm{Haar}}}^{(2)}
=
\frac{1}{D(D+1)}
\left(
\id_4^{\otimes n}
+
\sw_2^{\otimes n}
\right).
\end{equation}

\begin{lemma}[Common diagonal basis for the second moments]
\label{lem:k2-common-diagonal-basis}
Define the single-qubit two-copy basis
\begin{equation}\label{eq:basis}
\mc B_2
:=
\left\{
\ket{0,0},
\ket{1,1},
\frac{\ket{0,1}+\ket{1,0}}{\sqrt2},
\frac{\ket{0,1}-\ket{1,0}}{\sqrt2}
\right\}.
\end{equation}
Then  $\mb M_{\mc E_{\mathrm{shallow}}^{(\gamma)}}^{(2)},
\mb M_{\mc E_{\mathrm{phase}}}^{(2)},
\mb M_{\mc E_{\mathrm{Haar}}}^{(2)}$
are simultaneously diagonalizable in the tensor-product basis
\(\mc B_2^{\otimes n}\).
\end{lemma}

\begin{proof}
The basis \(\mc B_2\) separates the three-dimensional symmetric subspace from the one-dimensional antisymmetric subspace. On this basis,
\begin{equation}\label{eq:Iswap-basis}
[\id_4]_{\mc B_2}=\id_4,
\qquad
[\sw_2]_{\mc B_2}
=
\mathrm{diag}(1,1,1,-1).
\end{equation}
Since $\widetilde\Delta_2=\frac 13({\id_4+\sw_2}),$ we have
\begin{equation}\label{eq:k2block-basis}
[\widetilde{\Delta}_2]_{\mc B_2}
=
\mathrm{diag}\!\left(\frac23,\frac23,\frac23,0\right),\qquad[\id_4-\widetilde{\Delta}_2]_{\mc B_2}
=
\mathrm{diag}\!\left(\frac13,\frac13,\frac13,1\right),
\qquad
[\sw_2-\widetilde{\Delta}_2]_{\mc B_2}
=
\mathrm{diag}\!\left(\frac13,\frac13,\frac13,-1\right).
\end{equation}

Thus all local two-copy twirled blocks appearing in
\cref{eq:k2-shallow-moment,eq:k2-phase-moment,eq:k2-haar-moment}
are diagonal in the same local basis \(\mc B_2\). Since the global moment
operators are sums of tensor products of these local diagonal blocks, they are
diagonal in the tensor-product basis \(\mc B_2^{\otimes n}\).
\end{proof}

We now compute the corresponding eigenvalues. Let
\begin{equation}
    \mc B_{2,\mathrm{sym}}
:=
\left\{
\ket{0,0},
\ket{1,1},
\frac{\ket{0,1}+\ket{1,0}}{\sqrt2}
\right\},
\qquad
\ket{\psi_-}
:=
\frac{\ket{0,1}-\ket{1,0}}{\sqrt2}.
\end{equation}
For a tensor-product basis vector $\ket{\mb{z}}=\bigotimes_{q\in[n]}\ket{z_q}\in\mc B_2^{\otimes n},$ 
we denote $I_{\mathrm{sym}}^{(\mb{z})}
:=
\{q\in[n]:\ket{z_q}\in\mc B_{2,\mathrm{sym}}\}.$
Thus \(I_{\mathrm{sym}}^{(\mb{z})}\) records the sites whose local basis vector is one of the three symmetric basis vectors in $\mc B_{2,\mathrm{sym}}$, while \([n]\setminus I_{\mathrm{sym}}^{(\mb{z})}\) records the sites whose local basis vector is \(\ket{\psi_-}\). Since all local
eigenvalues above depend only on whether a site is symmetric or antisymmetric, this parameterization is sufficient to determine the eigenvalues of the three
global moment operators.

For \(\ket{\mb{z}}\in\mc B_2^{\otimes n}\), define
\[
\lambda_{\mc E_{\mathrm{phase}}}(\ket{\mb{z}})
:=
\bra{\mb{z}}
\left[\mb M_{\mc E_{\mathrm{phase}}}^{(2)}\right]_{\mc B_2^{\otimes n}}
\ket{\mb{z}},\qquad \lambda_{\mc E_{\mathrm{shallow}}^{(\gamma)}}(\ket{\mb{z}})
:=
\bra{\mb{z}}
\left[\mb M_{\mc E_{\mathrm{shallow}}^{(\gamma)}}^{(2)}\right]_{\mc B_2^{\otimes n}}
\ket{\mb{z}},
\qquad
\lambda_{\mc E_{\mathrm{Haar}}}(\ket{\mb{z}})
:=
\bra{\mb{z}}
\left[\mb M_{\mc E_{\mathrm{Haar}}}^{(2)}\right]_{\mc B_2^{\otimes n}}
\ket{\mb{z}}.
\]
\begin{lemma}[Eigenvalue formulas in the common basis]
\label{lem:k2-eigenvalues}
With the notation above, for every \(\ket{\mb z}\in\mc B_2^{\otimes n}\), the three second-moment operators in~\cref{eq:k2-shallow-moment,eq:k2-phase-moment,eq:k2-haar-moment} have the following eigenvalues:
\begin{align}
\lambda_{\mc E_{\mathrm{shallow}}^{(\gamma)}}(\ket{\mb z})
&=
2^{-2n}
\sum_{\substack{I_i\sqcup I_j\sqcup I_k=[n]\\ I_i\subseteq I_{\mathrm{sym}}^{(\mb z)}}}
\left(1-\frac{2\gamma \log n}{n}\right)^{|I_j||I_k|}
\left(\frac23\right)^{|I_i|}
\left(\frac13\right)^{|I_{\mathrm{sym}}^{(\mb z)}|-|I_i|}
(-1)^{|I_k\cap ([n]\setminus I_{\mathrm{sym}}^{(\mb z)})|},
\label{eq:k2-shallow-lambda}
\\
\lambda_{\mc E_{\mathrm{phase}}}(\ket{\mb z})
&=
2^{-2n}
\left[
1+(-1)^{n-|I_{\mathrm{sym}}^{(\mb z)}|}
-\delta_{|I_{\mathrm{sym}}^{(\mb z)}|,n}\left(\frac23\right)^n
\right],
\label{eq:k2-phase-lambda}
\\
\lambda_{\mc E_{\mathrm{Haar}}}(\ket{\mb z})
&=
\frac{1+(-1)^{n-|I_{\mathrm{sym}}^{(\mb z)}|}}{D(D+1)}.
\label{eq:k2-haar-lambda}
\end{align}
\end{lemma}

\begin{proof}
We first compute the eigenvalue $\lambda_{\mc E_{\mathrm{phase}}}(\ket{\mb{z}})$. Using~\cref{eq:Iswap-basis}, the identity term contributes $\bra{\mb{z}}[\id_4^{\otimes n}]_{\mc B_2^{\otimes n}}\ket{\mb{z}}=1$, and the swap term is $\bra{\mb{z}}[\sw_2^{\otimes n}]_{\mc B_2^{\otimes n}}\ket{\mb{z}}
=
(-1)^{n-|I_{\mathrm{sym}}^{(\mb{z})}|}.$
Besides, we obtain from~\cref{eq:k2block-basis} that
\[
\bra{\mb{z}}[\widetilde\Delta_2^{\otimes n}]_{\mc B_2^{\otimes n}}\ket{\mb{z}}
=
\begin{cases}
\left(\frac23\right)^n, & I_{\mathrm{sym}}^{(\mb{z})}=[n],\\
0, & \text{otherwise}.
\end{cases}
\]
Substituting these three contributions into~\cref{eq:k2-phase-moment} gives
\begin{equation}\label{eq:k2-phase-lambda-proof}
\lambda_{\mc E_{\mathrm{phase}}}(\ket{\mb{z}})
=
2^{-2n}\left[
1+(-1)^{n-|I_{\mathrm{sym}}^{(\mb{z})}|}
-\delta_{|I_{\mathrm{sym}}^{(\mb{z})}|,n}
\left(\frac23\right)^n
\right],
\end{equation}
which proves~\cref{eq:k2-phase-lambda}. Similarly,
\begin{equation}\label{eq:k2-haar-lambda-proof}
\lambda_{\mc E_{\mathrm{Haar}}}(\ket{\mb{z}})
=
\frac{\bra{\mb{z}}(\id_4^{\otimes n}+\sw_2^{\otimes n})\ket{\mb{z}}}{D(D+1)}
=
\frac{1+(-1)^{n-|I_{\mathrm{sym}}^{(\mb{z})}|}}{D(D+1)},
\end{equation}
which proves~\cref{eq:k2-haar-lambda}.

We next compute the eigenvalue $\lambda_{\mc E_{\mathrm{shallow}}^{(\gamma)}}(\ket{\mb{z}})$. By~\cref{eq:k2-shallow-moment},
\begin{equation}\label{eq:sparselambda-deriv-1}
\begin{split}
\lambda_{\mc E_{\mathrm{shallow}}^{(\gamma)}}(\ket{\mb{z}})
=
2^{-2n}
\sum_{I_i\sqcup I_j\sqcup I_k=[n]}
\left(1-\frac{2\gamma \log n}{n}\right)^{|I_j||I_k|}\times
\bra{\mb{z}}
\left[
\widetilde{\Delta}_2^{\otimes I_i}
\otimes
(\id_4-\widetilde{\Delta}_2)^{\otimes I_j}
\otimes
(\sw_2-\widetilde{\Delta}_2)^{\otimes I_k}
\right]_{\mc B_2^{\otimes n}}
\ket{\mb{z}} .
\end{split}
\end{equation}
Write \(\ket{\mb{z}}=\bigotimes_{q\in[n]}\ket{z_q}\). If \(q\in I_{\mathrm{sym}}^{(\mb{z})}\), then
\begin{equation}\label{eq:eigs-Sz}
\bra{z_q}[\widetilde{\Delta}_2]_{\mc B_2}\ket{z_q}=\frac23,
\qquad
\bra{z_q}[\id_4-\widetilde{\Delta}_2]_{\mc B_2}\ket{z_q}=\frac13,
\qquad
\bra{z_q}[\sw_2-\widetilde{\Delta}_2]_{\mc B_2}\ket{z_q}=\frac13.
\end{equation}
If \(q\in[n]\setminus I_{\mathrm{sym}}^{(\mb{z})}\), then
\begin{equation}\label{eq:eigs-Id}
\bra{z_q}[\widetilde{\Delta}_2]_{\mc B_2}\ket{z_q}=0,
\qquad
\bra{z_q}[\id_4-\widetilde{\Delta}_2]_{\mc B_2}\ket{z_q}=1,
\qquad
\bra{z_q}[\sw_2-\widetilde{\Delta}_2]_{\mc B_2}\ket{z_q}=-1.
\end{equation}
Thus, the matrix element in~\cref{eq:sparselambda-deriv-1} factorizes over qubits. If \(I_i\cap([n]\setminus I_{\mathrm{sym}}^{(\mb{z})})\neq\emptyset\), then this matrix
element vanishes because \(\widetilde\Delta_2\) has eigenvalue \(0\) in~\cref{eq:eigs-Id}. Hence, only terms with \(I_i\subseteq I_{\mathrm{sym}}^{(\mb{z})}\) survive.
For such terms, using~\cref{eq:eigs-Sz,eq:eigs-Id}, we obtain
\[
\begin{split}
&
\bra{\mb{z}}
\left[
\widetilde{\Delta}_2^{\otimes I_i}
\otimes
(\id_4-\widetilde{\Delta}_2)^{\otimes I_j}
\otimes
(\sw_2-\widetilde{\Delta}_2)^{\otimes I_k}
\right]_{\mc B_2^{\otimes n}}
\ket{\mb{z}} =
\left(\frac23\right)^{|I_i|}
\left(\frac13\right)^{|I_j\cap I_{\mathrm{sym}}^{(\mb{z})}|}
\left(\frac13\right)^{|I_k\cap I_{\mathrm{sym}}^{(\mb{z})}|}
(-1)^{|I_k\cap([n]\setminus I_{\mathrm{sym}}^{(\mb{z})})|}.
\end{split}
\]
Since \(I_i\sqcup I_j\sqcup I_k=[n]\) and \(I_i\subseteq I_{\mathrm{sym}}^{(\mb{z})}\),
\[
|I_j\cap I_{\mathrm{sym}}^{(\mb{z})}|
+
|I_k\cap I_{\mathrm{sym}}^{(\mb{z})}|
+
|I_i|
=
|I_{\mathrm{sym}}^{(\mb{z})}|.
\]
Substituting this into~\cref{eq:sparselambda-deriv-1} gives
\cref{eq:k2-shallow-lambda}.
\end{proof}

Therefore, the two operator-norm comparisons in~\cref{eq:k2-2moment} reduce to
the scalar eigenvalue comparisons
\[
\left\|
\mb M_{\mc E_{\mathrm{shallow}}^{(\gamma)}}^{(2)}
-
\mb M_{\mc E_{\mathrm{phase}}}^{(2)}
\right\|_\infty
=
\max_{\ket{\mb{z}}\in\mc B_2^{\otimes n}}
\left|
\lambda_{\mc E_{\mathrm{shallow}}^{(\gamma)}}(\ket{\mb{z}})
-
\lambda_{\mc E_{\mathrm{phase}}}(\ket{\mb{z}})
\right|,
\]
and
\[
\left\|
\mb M_{\mc E_{\mathrm{phase}}}^{(2)}
-
\mb M_{\mc E_{\mathrm{Haar}}}^{(2)}
\right\|_\infty
=
\max_{\ket{\mb{z}}\in\mc B_2^{\otimes n}}
\left|
\lambda_{\mc E_{\mathrm{phase}}}(\ket{\mb{z}})
-
\lambda_{\mc E_{\mathrm{Haar}}}(\ket{\mb{z}})
\right|.
\]
These scalar comparisons are carried out in the next two subsections.

\subsection{The norm gap between $\mc E_{\mathrm{phase}}$ and $\mc E_{\mathrm{Haar}}$}\label{ap:k2-phase-haar}

We next compare the second moment of $\mc E_{\mathrm{phase}}$ with the Haar second moment. This part is simpler, since both eigenvalue formulas are already explicit in the common basis $\mc B_2^{\otimes n}$.

\begin{prop}[Second-moment gap between the dense phase and Haar ensembles]
\label{prop:k2-phase-haar-norm}
One has
\begin{equation}
\left\|
\mb M_{\mc E_{\mathrm{phase}}}^{(2)}
-
\mb M_{\mc E_{\mathrm{Haar}}}^{(2)}
\right\|_\infty
\le
D^{-2}\left(\frac23\right)^n.
\end{equation}
In particular, the ratio of this gap to \(D^{-2}\) decays exponentially
with \(n\).
\end{prop}

\begin{proof}
Since both operators are diagonal in $\mc B_2^{\otimes n}$, it suffices to compare their eigenvalues.
From \cref{eq:k2-phase-lambda,eq:k2-haar-lambda}, for all basis states
$\ket{\mb z}$, there are three cases:
\[
\lambda_{\mc E_{\mathrm{Haar}}}(\ket{\mb z})
=
\begin{cases}
\displaystyle \frac{2}{D(D+1)}, & |I_{\mathrm{sym}}^{(\mb z)}|=n,\\[6pt]
\displaystyle \frac{2}{D(D+1)}, & n-|I_{\mathrm{sym}}^{(\mb z)}|>0 \text{ even},\\[6pt]
0, & n-|I_{\mathrm{sym}}^{(\mb z)}| \text{ odd},
\end{cases}
\]
while
\[
\lambda_{\mc E_{\mathrm{phase}}}(\ket{\mb z})
=
\begin{cases}
\displaystyle D^{-2}\left(2-\left(\frac23\right)^n\right), & |I_{\mathrm{sym}}^{(\mb z)}|=n,\\[6pt]
\displaystyle 2D^{-2}, & n-|I_{\mathrm{sym}}^{(\mb z)}|>0 \text{ even},\\[6pt]
0, & n-|I_{\mathrm{sym}}^{(\mb z)}| \text{ odd}.
\end{cases}
\]
Thus the eigenvalue difference is
\[
\left|
\lambda_{\mc E_{\mathrm{phase}}}(\ket{\mb z})-\lambda_{\mc E_{\mathrm{Haar}}}(\ket{\mb z})
\right|
=
\begin{cases}
\displaystyle
\displaystyle D^{-2}
\left|
\left(\frac23\right)^n-\frac{2}{D+1}
\right|, & |I_{\mathrm{sym}}^{(\mb z)}|=n,\\[10pt]
\displaystyle
\displaystyle D^{-2}\frac{2}{D+1},
& n-|I_{\mathrm{sym}}^{(\mb z)}|>0 \text{ even},\\[10pt]
0, & n-|I_{\mathrm{sym}}^{(\mb z)}| \text{ odd}.
\end{cases}
\]
For every \(n\ge1\),
\begin{equation}\label{eq:k2-explicit-constant}
\frac{2}{D+1}
=
\frac{2}{2^n+1}
\le
\left(\frac23\right)^n.
\end{equation}
Indeed, the convexity of \(x\mapsto x^n\) on \((0,\infty)\) gives
\(3^n\le(4^n+2^n)/2\), which is equivalent to~\cref{eq:k2-explicit-constant}.
It follows that each of the two nonzero differences above is at most
\(D^{-2}(2/3)^n\). Therefore,
\begin{equation}
\max_{\ket{\mb z}\in\mc B_2^{\otimes n}}
\left|
\lambda_{\mc E_{\mathrm{phase}}}(\ket{\mb z})-\lambda_{\mc E_{\mathrm{Haar}}}(\ket{\mb z})
\right|
\le
D^{-2}\left(\frac23\right)^n,
\end{equation}
which proves the claim.
\end{proof}

\subsection{The norm gap between $\mc E_{\mathrm{shallow}}^{(\gamma)}$ and $\mc E_{\mathrm{phase}}$}\label{ap:k2-shallow-phase}

In this subsection, we bound the operator norm of
$
\mb M_{\mc E_{\mathrm{shallow}}^{(\gamma)}}^{(2)}
-
\mb M_{\mc E_{\mathrm{phase}}}^{(2)}
$ as follows.

\begin{prop}[Second-moment gap between the shallow and dense ensembles]
\label{prop:k2-shallow-phase-norm}
For every fixed $\gamma>4$ and all sufficiently large \(n\),
\begin{equation}
\left\|
\mb M_{\mc E_{\mathrm{shallow}}^{(\gamma)}}^{(2)}
-
\mb M_{\mc E_{\mathrm{phase}}}^{(2)}
\right\|_\infty
\le
D^{-2}\Bigl(n^{1-\gamma/4}+\mathrm{negl}(n)\Bigr).
\end{equation}
\end{prop}

\begin{proof}
Since the two second-moment operators are simultaneously diagonalizable in
$\mc B_2^{\otimes n}$, their operator norm difference is the maximal eigenvalue difference:
\begin{equation}\label{eq:k2-opnorm-reduce}
\left\|
\mb M_{\mc E_{\mathrm{shallow}}^{(\gamma)}}^{(2)}
-
\mb M_{\mc E_{\mathrm{phase}}}^{(2)}
\right\|_\infty
=
\max_{\ket{\mb{z}}\in \mc B_2^{\otimes n}}
\left|
\lambda_{\mc E_{\mathrm{shallow}}^{(\gamma)}}(\ket{\mb{z}})
-
\lambda_{\mc E_{\mathrm{phase}}}(\ket{\mb{z}})
\right|.
\end{equation}

For convenience, we introduce the following notation. Let
\(n_1:=|I_{\mathrm{sym}}^{(\mb{z})}|\) and
\(n_2:=|[n]\setminus I_{\mathrm{sym}}^{(\mb{z})}|=n-n_1\).
We reparameterize the ordered partition \(I_i\sqcup I_j\sqcup I_k=[n]\) by
\begin{equation}\label{eq:a-def}
a_1:=|I_i|,
\qquad
a_2:=|I_j\cap I_{\mathrm{sym}}^{(\mb{z})}|,
\qquad
a_3:=|I_k\cap I_{\mathrm{sym}}^{(\mb{z})}|,
\qquad
a_4:=|I_j\cap([n]\setminus I_{\mathrm{sym}}^{(\mb{z})})|,
\qquad
a_5:=|I_k\cap([n]\setminus I_{\mathrm{sym}}^{(\mb{z})})|.
\end{equation}
These parameters satisfy $a_1+a_2+a_3=n_1, 
a_4+a_5=n_2.$
For fixed values \(a_1,\ldots,a_5\), the number of choices of
\(I_i,I_j,I_k\) is $\binom{n_1}{a_1}\binom{n_1-a_1}{a_2}\binom{n_2}{a_4}.$
Indeed, we first choose \(I_i\subseteq I_{\mathrm{sym}}^{(\mb{z})}\), then choose
\(I_j\cap I_{\mathrm{sym}}^{(\mb{z})}\) from the remaining sites in \(I_{\mathrm{sym}}^{(\mb{z})}\), and finally choose
\(I_j\cap([n]\setminus I_{\mathrm{sym}}^{(\mb{z})})\). All remaining sites are assigned to \(I_k\). The number of partitions $(I_i,I_j,I_k)$ with the parameters satisfying~\cref{eq:a-def} is therefore $\binom{n_1}{a_1}
\binom{n_1-a_1}{a_2}
\binom{n_2}{a_4}$.
In this way, we can rewrite~\cref{eq:k2-shallow-lambda} as 
\begin{equation}
\begin{split}
    \lambda_{\mc E_{\mathrm{shallow}}^{(\gamma)}}(\ket{\mb{z}})=
2^{-2n}
\sum_{\substack{
a_1+a_2+a_3=n_1\\
a_4+a_5=n_2
}}\binom{n_1}{a_1}
\binom{n_1-a_1}{a_2}
\binom{n_2}{a_4}
\left(1-\frac{2\gamma \log n}{n}\right)^{(a_2+a_4)(a_3+a_5)}
\left(\frac23\right)^{a_1}
\left(\frac13\right)^{a_2+a_3}
(-1)^{a_5},
\end{split}
\label{eq:k2-shallow-lambda-caseB}
\end{equation}
The contribution of $\mc E_{\mathrm{phase}}$ in~\cref{eq:k2-shallow-lambda-caseB} is the part of the expression with
\(|I_j||I_k|=0\), namely \(a_2+a_4=0\) or \(a_3+a_5=0\). Hence, after subtracting
the eigenvalues of $\mc E_{\mathrm{phase}}$, only the terms with $a_2+a_4>0,
a_3+a_5>0$ remain. Therefore,
\begin{equation}\label{eq:caseB-signed}
\begin{split}
&D^2
\left(
\lambda_{\mc E_{\mathrm{shallow}}^{(\gamma)}}(\ket{\mb{z}})
-
\lambda_{\mc E_{\mathrm{phase}}}(\ket{\mb{z}})
\right)
=
\sum_{\substack{
a_1+a_2+a_3=n_1\\
a_4+a_5=n_2\\
a_2+a_4>0,\ a_3+a_5>0
}}
\binom{n_1}{a_1}
\binom{n_1-a_1}{a_2}
\binom{n_2}{a_4}
\left(\frac23\right)^{a_1}
\left(\frac13\right)^{a_2+a_3}(-1)^{a_5}\times
\left(1-\frac{2\gamma\log n}{n}\right)^{(a_2+a_4)(a_3+a_5)}.
\end{split}
\end{equation}
Discarding the sign \((-1)^{a_5}\) gives
\begin{equation}\label{eq:caseB-abs}
\begin{split}
&D^2
\left|
\lambda_{\mc E_{\mathrm{shallow}}^{(\gamma)}}(\ket{\mb{z}})
-
\lambda_{\mc E_{\mathrm{phase}}}(\ket{\mb{z}})
\right|\leq \sum_{\substack{
a_1+a_2+a_3=n_1\\
a_4+a_5=n_2\\
a_2+a_4>0,\ a_3+a_5>0
}}
\binom{n_1}{a_1}
\binom{n_1-a_1}{a_2}
\binom{n_2}{a_4}
\left(\frac23\right)^{a_1}
\left(\frac13\right)^{a_2+a_3}
\left(1-\frac{2\gamma\log n}{n}\right)^{(a_2+a_4)(a_3+a_5)}.
\\
&=
3^{-n_1}
\sum_{\substack{
a_1+a_2+a_3=n_1\\
a_4+a_5=n_2\\
a_2+a_4>0,\ a_3+a_5>0
}}
\binom{n_1}{a_1}
\binom{n_1-a_1}{a_2}
\binom{n_2}{a_4}
2^{a_1}
\left(1-\frac{2\gamma\log n}{n}\right)^{(a_2+a_4)(a_3+a_5)}.
\end{split}
\end{equation}
Now set \(b:=a_2+a_4\) and \(r:=n-a_1\). Then
\(a_3+a_5=r-b\), and the two positivity conditions become
\(1\le b\le r-1\). For fixed \(a_1\) and \(b\), Vandermonde's identity,
with the admissible ranges made explicit, gives
\begin{equation}\label{eq:k2-vandermonde-identity}
\sum_{\substack{a_2+a_4=b\\
0\le a_2\le n_1-a_1\\
0\le a_4\le n_2}}
\binom{n_1-a_1}{a_2}
\binom{n_2}{a_4}
=
\binom{n-a_1}{b}.
\end{equation}
If \(r\in\{0,1\}\), there is no admissible \(b\), so the corresponding
contribution vanishes. Hence only
\(0\le a_1\le\min\{n_1,n-2\}\) needs to be retained. Applying
\cref{eq:k2-vandermonde-identity} to~\cref{eq:caseB-abs} gives
\begin{equation}\label{eq:caseB-vandermonde}
\begin{split}
&D^2
\left|
\lambda_{\mc E_{\mathrm{shallow}}^{(\gamma)}}(\ket{\mb{z}})
-
\lambda_{\mc E_{\mathrm{phase}}}(\ket{\mb{z}})
\right|
\le
3^{-n_1}
\sum_{a_1=0}^{\min\{n_1,n-2\}}
\binom{n_1}{a_1}2^{a_1}
\sum_{b=1}^{r-1}
\binom{r}{b}
\left(1-\frac{2\gamma\log n}{n}\right)^{b(r-b)}
\\
&\le
2\sum_{a_1=0}^{\min\{n_1,n-2\}}
3^{-n_1}
\binom{n_1}{a_1}
2^{a_1}
\left(
\left[
1+
\left(1-\frac{2\gamma\log n}{n}\right)^{\frac{r}{2}}
\right]^r
-1
\right)
\qquad
\text{(by~\cref{lem:comb_bound})}
\\
&\le
2\sum_{a_1=0}^{\min\{n_1,n-2\}}
3^{-n_1}
\binom{n_1}{a_1}
2^{a_1}
\left(
\left[
1+n^{-\frac{\gamma r}{n}}
\right]^r
-1
\right).
\end{split}
\end{equation}
In the last step, we used
\[
\left(1-\frac{2\gamma\log n}{n}\right)^{r/2}
\le
\exp\!\left(-\frac{\gamma r\log n}{n}\right)
=
n^{-\frac{\gamma r}{n}}.
\]

We split the last sum in~\cref{eq:caseB-vandermonde} according to the
effective size \(n-a_1\). 

First, consider the summation in the range
\(n-a_1\le n/4\), equivalently \(a_1\ge 3n/4\). This range is empty unless
\(n_1\ge 3n/4\). We use the binary entropy and binary relative entropy
\[
H_2(x):=-x\log_2x-(1-x)\log_2(1-x),
\qquad
D_2(x\|y):=x\log_2\frac{x}{y}+(1-x)\log_2\frac{1-x}{1-y},
\]
with the usual continuous conventions at \(x=0,1\). For every contributing
\(a_1\), let \(x:=a_1/n_1\ge 3/4\). Using the standard entropy bound
\(\binom{n_1}{a_1}\le 2^{n_1H_2(x)}\), we obtain
\[
\begin{split}
3^{-n_1}\binom{n_1}{a_1}2^{a_1}=
\binom{n_1}{a_1}
\left(\frac23\right)^{xn_1}
\left(\frac13\right)^{(1-x)n_1}\le
2^{n_1H_2(x)}
\left(\frac23\right)^{xn_1}
\left(\frac13\right)^{(1-x)n_1}=
2^{-n_1D_2(x\|2/3)}.
\end{split}
\]
Moreover, \(D_2(x\|2/3)\) is increasing for \(x\ge 2/3\). Since
\(x\ge3/4\), it follows that $2^{-n_1D_2(x\|2/3)}
\le
2^{-n_1D_2(3/4\|2/3)}.$
 Since there are at most \(n_1+1\) terms and
\(n_1\ge3n/4\),
\begin{equation}\label{eq:caseB-small-effective-weight}
\sum_{\substack{0\le a_1\le n_1\\ n-a_1\le n/4}}
3^{-n_1}\binom{n_1}{a_1}2^{a_1}
\le
(n_1+1)2^{-\frac{3n}{4}D_2(3/4\|2/3)}
\le
e^{-n/200}
\end{equation}
for sufficiently large \(n\). By part~(A) of~\cref{lem:sub_exp}, uniformly over this range,
\[
\left[
1+n^{-\frac{\gamma(n-a_1)}{n}}
\right]^{n-a_1}
= \exp(\frac{n}{\gamma e\log n}).
\]
Therefore, the contribution of the range \(n-a_1\le n/4\) to
\cref{eq:caseB-vandermonde} is bounded by $2e^{-n/200}\exp(\frac{n}{\gamma e\log n})
=
\mathrm{negl}(n).$

It remains to consider the range \(n-a_1>n/4\). By part~(B) of
\cref{lem:sub_exp}, for \(\gamma>4\) and sufficiently large \(n\), we have
\[
\left[
1+n^{-\frac{\gamma(n-a_1)}{n}}
\right]^{n-a_1}
-1
\le
\frac12 n^{1-\gamma/4}.
\]
Moreover, $\sum_{a_1=0}^{n_1}
3^{-n_1}
\binom{n_1}{a_1}
2^{a_1}
=
1.$
Therefore, the contribution from the range \(n-a_1>n/4\) is bounded by $n^{1-\gamma/4}.$

Combining the two ranges gives, uniformly over
\(\ket{\mb z}\in\mc B_2^{\otimes n}\),
\begin{equation}\label{eq:caseB-final}
D^2
\left|
\lambda_{\mc E_{\mathrm{shallow}}^{(\gamma)}}(\ket{\mb{z}})
-
\lambda_{\mc E_{\mathrm{phase}}}(\ket{\mb{z}})
\right|
\le
n^{1-\gamma/4}
+
\mathrm{negl}(n).
\end{equation}
Taking the maximum over \(\ket{\mb z}\) and using
\cref{eq:k2-opnorm-reduce} proves the proposition.
\end{proof}

\begin{corollary}[Second-order relative-error bound]
\label{cor:k2-relative-error}
For every fixed \(\gamma>4\) and all sufficiently large \(n\), the shallow
phase ensemble is an approximate state \(2\)-design in relative error, with
relative error at most
\begin{equation}\label{eq:k2-relative-error-final}
\frac{D+1}{2D}
\left[
n^{1-\gamma/4}
+
\left(\frac23\right)^n
+
\mathrm{negl}(n)
\right].
\end{equation}
In particular, this relative error vanishes as \(n\to\infty\).
\end{corollary}

\begin{proof}
By the triangle inequality and Propositions~\ref{prop:k2-phase-haar-norm} and \ref{prop:k2-shallow-phase-norm},
\[
\left\|
\mb M_{\mc E_{\mathrm{shallow}}^{(\gamma)}}^{(2)}
-
\mb M_{\mc E_{\mathrm{Haar}}}^{(2)}
\right\|_\infty
\le
D^{-2}
\left[
n^{1-\gamma/4}
+
\left(\frac23\right)^n
+
\mathrm{negl}(n)
\right].
\]
The claim follows from \cref{lem:relative-error-opnorm} and
\(r_{\mathrm{sym}}^{(2)}=\binom{D+1}{2}=D(D+1)/2\).
\end{proof}
We finally note a complementary second-order perspective from Ref.~\cite{heinrich2025anti}, which promotes sufficiently accurate anti-concentration to a relative-error state $2$-design for suitable locally invariant ensembles. The reduction is specific to the second moment and requires quantitative control of the anti-concentration error. Our moment expansion can in fact evaluate this anti-concentration directly for sparse phase circuits, both with and without the final local Clifford twirl, providing a possible bridge between the two approaches. 

\section{The third-order relative error analysis}\label{Ap:3twirled}

In this section, we establish the third-order relative-error bound using the
dense phase ensemble $\mc{E}_{\mathrm {phase}}$ as an intermediate reference. As in the \(k=2\) case,
it reduces the problem to controlling
\[
\left\|
\mb M_{\mc E_{\mathrm{phase}}}^{(3)}
-
\mb M_{\mc E_{\mathrm{Haar}}}^{(3)}
\right\|_\infty,
\qquad
\left\|
\mb M_{\mc E_{\mathrm{shallow}}^{(\gamma)}}^{(3)}
-
\mb M_{\mc E_{\mathrm{phase}}}^{(3)}
\right\|_\infty.
\]

The third-order analysis is substantially more involved than its second-order
counterpart. At \(k=2\), the three relevant local blocks \(\widetilde{\Delta}_2\),
\(\id_4-\widetilde{\Delta}_2\), and
\(\sw_2-\widetilde{\Delta}_2\) commute and are
simultaneously diagonalizable, so the global operator-norm problem reduces to
scalar estimates. At \(k=3\), by contrast, the local decomposition contains ten blocks in the
computational basis (see~\cref{eq:k3Blocks} for their explicit forms). These
blocks are no longer simultaneously diagonalizable, and the possible ways of
assigning them to the \(n\) qubit sites give rise to a substantially more complicated combinatorial counting problem.

The structure of this section is as follows. In
\cref{ap:k3-local-blocks}, we analyze the local third-order blocks by
decomposing them into several simple subspaces and deriving the norm bounds
needed later. In \cref{ap:k3-phase-haar}, we use these local estimates to
bound the gap between the dense phase and Haar ensembles. Finally, in
\cref{ap:k3-shallow-phase}, we bound the gap between the dense phase and shallow phase ensembles, which constitutes
the most technically demanding part of the entire proof. There, we introduce the chessboard technique and show that the combinatorial multiplicities associated with the nine nontrivial block types can be
controlled using only three parameters \(m_1,m_2,m_3\), reducing the proof to a finite case analysis.

\subsection{Local structure of the three-copy blocks}
\label{ap:k3-local-blocks}

In this subsection, we show that
\(\widetilde{\Delta}_3\) and
\(\widetilde B(\pi_1,\pi_2)\) preserve the common decomposition
\[
(\mathbb C^2)^{\otimes3}
=
\mc H_{\mathrm{sym}}
\oplus
\mc H_{\mathrm{std}}^{(0)}
\oplus
\mc H_{\mathrm{std}}^{(1)},
\]
and establish the exact local norm identities
\[
\|\widetilde{\Delta}_3\|_\infty=\frac12,
\qquad
\|\widetilde B(\pi_1,\pi_2)\|_\infty
=
\|\widetilde B(\pi_1,\pi_2)+\widetilde{\Delta}_3\|_\infty
=\frac23.
\]
These identities provide the local contraction used in the two third-moment
comparisons below. To derive these results, we write them in the following orthonormal basis.
\begin{equation}\label{eq:k3-local-basis}
\mc B_3
:=
\left\{
\ket{s_0},
\ket{s_1},
\ket{s_2},
\ket{s_3},
\ket{u_1},
\ket{u_2},
\ket{v_1},
\ket{v_2}
\right\},
\end{equation}
where
\begin{align}
\ket{s_0}
&=
\ket{0,0,0},
&
\ket{s_1}
&=
\frac{
\ket{0,0,1}
+
\ket{0,1,0}
+
\ket{1,0,0}
}{\sqrt{3}},
\nonumber\\
\ket{s_2}
&=
\frac{
\ket{1,1,0}
+
\ket{1,0,1}
+
\ket{0,1,1}
}{\sqrt{3}},
&
\ket{s_3}
&=
\ket{1,1,1},
\label{eq:k3-symmetric-basis}
\end{align}
and
\begin{align}
\ket{u_1}
&=
\frac{
\ket{0,0,1}
-
\ket{0,1,0}
}{\sqrt{2}},
&
\ket{u_2}
&=
\frac{
\ket{0,0,1}
+
\ket{0,1,0}
-
2\ket{1,0,0}
}{\sqrt{6}},
\nonumber\\
\ket{v_1}
&=
\frac{
\ket{1,1,0}
-
\ket{1,0,1}
}{\sqrt{2}},
&
\ket{v_2}
&=
\frac{
\ket{1,1,0}
+
\ket{1,0,1}
-
2\ket{0,1,1}
}{\sqrt{6}}.
\label{eq:k3-standard-basis}
\end{align}

The first four vectors span the fully symmetric subspace $\mc H_{\mathrm{sym}}
:=
\operatorname{Sym}^3(\mathbb C^2),$
while $\mc H_{\mathrm{std}}^{(0)}
:=
\operatorname{span}
\left\{
\ket{u_1},
\ket{u_2}
\right\},
\mc H_{\mathrm{std}}^{(1)}
:=
\operatorname{span}
\left\{
\ket{v_1},
\ket{v_2}
\right\}$
carry two copies of the standard irreducible two-dimensional representation of the
third-order permutation group \(S_3\).
Therefore,
\begin{equation}\label{eq:k3-local-space-decomposition}
(\mathbb C^2)^{\otimes 3}
=
\mc H_{\mathrm{sym}}
\oplus
\mc H_{\mathrm{std}}^{(0)}
\oplus
\mc H_{\mathrm{std}}^{(1)}.
\end{equation}
For later use, we also define the nonsymmetric subspace $\mc H_{\mathrm{ns}}
:=
\mc H_{\mathrm{std}}^{(0)}
\oplus
\mc H_{\mathrm{std}}^{(1)}.$

\begin{observation}[Single-qubit three-copy block forms]
\label{obs:k3-local-block-forms}
With respect to the basis \(\mc B_3\), the following statements hold.

\begin{enumerate}
\item For every \(\pi\in S_3\), there exists a real orthogonal matrix
\(A_\pi\in O(2)\) such that
\begin{equation}\label{eq:k3-local-permutation-block}
[V_1(\pi)]_{\mc B_3}
=
\id_4
\oplus
A_\pi
\oplus
A_\pi.
\end{equation}
Moreover, the matrices \(A_\pi\) form the standard two-dimensional
representation of \(S_3\).

\item The twirled diagonal block satisfies
\begin{equation}\label{eq:widetildedelta}
[\widetilde{\Delta}_3]_{\mc B_3}
=
\frac{1}{2}\id_4
\oplus
\mathbf{0}_2
\oplus
\mathbf{0}_2
\end{equation} 

\item For
\(\pi_1\in S_{3,\mathrm{odd}}\) and
\(\pi_2\in S_{3,\mathrm{even}}\), define $C(\pi_1,\pi_2)
:=
\frac{1}{3}
\left(
A_{\pi_1}
+
A_{\pi_2}
\right).$
Then
\begin{equation}\label{eq:widetildeB}
[\widetilde{B}(\pi_1,\pi_2)]_{\mc B_3}
=
\frac{1}{6}\id_4
\oplus
C(\pi_1,\pi_2)
\oplus
C(\pi_1,\pi_2).
\end{equation}
\end{enumerate}
\end{observation}

As a reminder, \(S_{3,\mathrm{odd}}\) and \(S_{3,\mathrm{even}}\) denote the
odd and even permutations in \(S_3\), respectively, while the notation \(\widetilde{(\,\cdot\,)}\) indicates the corresponding local block after
single-qubit Clifford twirling. See~\cref{Ap:moment3} for the precise definitions.

The first statement follows from a direct calculation of the permutation
operators in the basis \(\mc B_3\). The remaining two statements follow from
the single-qubit twirling identities in~\cref{obs:k3-single-site-twirls}.
For clarity, all of these identities admit an exact finite local
verification: one evaluates the six permutation matrices in the basis
\(\mc B_3\), performs the uniform average over the \(24\) single-qubit
Clifford operators, and compares the resulting \(4\oplus2\oplus2\) blocks.
This verification is independent of \(n\) and yields precisely
\cref{eq:k3-local-permutation-block,eq:widetildedelta,eq:widetildeB}.

\begin{lemma}[Operator norms of twirled local blocks]
\label{lem:k3-local-block-norms}
For every
\(\pi_1\in S_{3,\mathrm{odd}}\) and
\(\pi_2\in S_{3,\mathrm{even}}\), the matrix
\(C(\pi_1,\pi_2)\) has singular values \(2/3\) and \(0\). Consequently,
\begin{align}
\|\widetilde{\Delta}_3\|_\infty
&=
\frac{1}{2},
\label{eq:k3-delta-local-norm}
\\
\|\widetilde{B}(\pi_1,\pi_2)\|_\infty
&=
\frac{2}{3},
\label{eq:blockinfty2-m3}
\\
\|\widetilde{B}(\pi_1,\pi_2)+\widetilde{\Delta}_3\|_\infty
&=
\frac{2}{3}.
\label{eq:blockinfty3-m3}
\end{align}
\end{lemma}

\begin{proof}
The identity~\cref{eq:k3-delta-local-norm} follows immediately from
\cref{eq:widetildedelta}.

We next determine the singular values of \(C(\pi_1,\pi_2)\). Set $\tau:=\pi_1^{-1}\pi_2.$
Since \(\pi_1\) is odd and \(\pi_2\) is even, \(\tau\) is odd and hence is a
transposition. By the representation property, $A_{\pi_1}+A_{\pi_2}
=
A_{\pi_1}\bigl(\id_2+A_\tau\bigr).$
Because \(A_{\pi_1}\) is orthogonal, left multiplication by \(A_{\pi_1}\)
preserves singular values.

In the standard two-dimensional representation of \(S_3\), a transposition
acts as a reflection. Thus \(A_\tau\) has eigenvalues \(1\) and \(-1\), and
\(\id_2+A_\tau\) has singular values \(2\) and \(0\). It follows that $C(\pi_1,\pi_2)
=
\frac{1}{3}\bigl(A_{\pi_1}+A_{\pi_2}\bigr)$
has singular values \(2/3\) and \(0\).

By~\cref{eq:widetildeB}, the three blocks therefore have norms \(1/6\), \(2/3\), and \(2/3\),
respectively, proving~\cref{eq:blockinfty2-m3}.

Finally, combining~\cref{eq:widetildedelta,eq:widetildeB} gives
\[
[\widetilde{B}(\pi_1,\pi_2)+\widetilde{\Delta}_3]_{\mc B_3}
=
\frac{2}{3}\id_4
\oplus
C(\pi_1,\pi_2)
\oplus
C(\pi_1,\pi_2).
\]
Each diagonal block in $
[\widetilde{B}(\pi_1,\pi_2)+\widetilde{\Delta}_3]_{\mc B_3}$ has norm \(2/3\), which proves
\cref{eq:blockinfty3-m3}.
\end{proof}

The block forms above provide the local input for the remainder of the third-order analysis below.

\subsection{The norm gap between
\(\mc E_{\mathrm{phase}}\) and \(\mc E_{\mathrm{Haar}}\)}
\label{ap:k3-phase-haar}

We first compare the third moment of $\mc E_{\mathrm{phase}}$ with the Haar third moment.
The explicit moment formulas from the previous appendices give an exponentially
small normalized operator-norm gap.

\begin{prop}[Third-moment gap between the dense phase and Haar ensembles]
\label{prop:k3-phase-haar-opnorm}
One has
\begin{equation}\label{eq:k3-phase-haar-opnorm}
\left\|
\mb M_{\mc E_{\mathrm{phase}}}^{(3)}
-
\mb M_{\mc E_{\mathrm{Haar}}}^{(3)}
\right\|_\infty
\le
2^{-3n}
\left[
9\left(\frac{2}{3}\right)^n
+
22\cdot 2^{-n}
\right].
\end{equation}
In particular, $2^{3n}\left\|
\mb M_{\mc E_{\mathrm{phase}}}^{(3)}
-
\mb M_{\mc E_{\mathrm{Haar}}}^{(3)}
\right\|_\infty$ still decays exponentially
with \(n\).
\end{prop}

\begin{proof}
By~\cref{eq:phase-third-moment-app,eq:haar-moment-sym-proj}, specialized to
\(k=3\), and using \(D=2^n\), we have
\begin{align}
&
\mb M_{\mc E_{\mathrm{phase}}}^{(3)}
-
\mb M_{\mc E_{\mathrm{Haar}}}^{(3)}=
2^{-3n}
\left[
\sum_{\pi\in S_3}V_n(\pi)
-
\sum_{\substack{
\pi_1\in S_{3,\mathrm{odd}}\\
\pi_2\in S_{3,\mathrm{even}}
}}
\left(
\widetilde B(\pi_1,\pi_2)
+
\widetilde\Delta_3
\right)^{\otimes n}
+
4\widetilde\Delta_3^{\otimes n}
\right]
-
\frac{1}{D(D+1)(D+2)}
\sum_{\pi\in S_3}V_n(\pi)\\
&=
\left(
\frac{1}{D^3}
-
\frac{1}{D(D+1)(D+2)}
\right)
\sum_{\pi\in S_3}V_n(\pi)-
2^{-3n}
\sum_{\substack{
\pi_1\in S_{3,\mathrm{odd}}\\
\pi_2\in S_{3,\mathrm{even}}
}}
\left(
\widetilde B(\pi_1,\pi_2)
+
\widetilde\Delta_3
\right)^{\otimes n}
+
4\cdot 2^{-3n}\widetilde\Delta_3^{\otimes n}\\
&=
\left(
\frac{6}{D^3}
-
\frac{6}{D(D+1)(D+2)}
\right)
\Pi_{\mathrm{sym}}^{(3)}-
2^{-3n}
\sum_{\substack{
\pi_1\in S_{3,\mathrm{odd}}\\
\pi_2\in S_{3,\mathrm{even}}
}}
\left(
\widetilde B(\pi_1,\pi_2)
+
\widetilde\Delta_3
\right)^{\otimes n}
+
4\cdot 2^{-3n}\widetilde\Delta_3^{\otimes n}.
\label{eq:k3-phase-haar-difference}
\end{align}
Applying the triangle inequality to
\cref{eq:k3-phase-haar-difference}, we obtain
\begin{align}
&
\left\|
\mb M_{\mc E_{\mathrm{phase}}}^{(3)}
-
\mb M_{\mc E_{\mathrm{Haar}}}^{(3)}
\right\|_\infty\le
\left|
\frac{6}{D^3}
-
\frac{6}{D(D+1)(D+2)}
\right|
\left\|
\Pi_{\mathrm{sym}}^{(3)}
\right\|_\infty+
2^{-3n}
\sum_{\substack{
\pi_1\in S_{3,\mathrm{odd}}\\
\pi_2\in S_{3,\mathrm{even}}
}}
\left\|
\left(
\widetilde B(\pi_1,\pi_2)
+
\widetilde\Delta_3
\right)^{\otimes n}
\right\|_\infty
+
4\cdot 2^{-3n}
\left\|
\widetilde\Delta_3^{\otimes n}
\right\|_\infty\\
&=\left|
\frac{6}{D^3}
-
\frac{6}{D(D+1)(D+2)}
\right|+
2^{-3n}
\sum_{\substack{
\pi_1\in S_{3,\mathrm{odd}}\\
\pi_2\in S_{3,\mathrm{even}}
}}
\left\|
\widetilde B(\pi_1,\pi_2)
+
\widetilde\Delta_3
\right\|_\infty^n
+
4\cdot 2^{-3n}
\left\|
\widetilde\Delta_3
\right\|_\infty^n.
\label{eq:k3-phase-haar-triangle}
\end{align}
First, the normalization difference satisfies
\begin{equation}\label{eq:k3-haar-normalization-gap}
\begin{split}
\left|
\frac{6}{D^3}
-
\frac{6}{D(D+1)(D+2)}
\right|
&=
\frac{6(3D+2)}{D^3(D+1)(D+2)}\le
\frac{18}{D^4}
=
18\cdot 2^{-4n}.
\end{split}
\end{equation}

Moreover, we have used the multiplicativity of the operator norm under tensor
products together with~\cref{lem:k3-local-block-norms}, which gives local
norm bounds of \(2/3\) and \(1/2\) for the relevant twirled blocks. Since these
norms are strictly smaller than one, local Clifford twirling produces an exponential suppression in the number of qubits.

\begin{align*}
\left\|
\mb M_{\mc E_{\mathrm{phase}}}^{(3)}
-
\mb M_{\mc E_{\mathrm{Haar}}}^{(3)}
\right\|_\infty
&\le
18\cdot 2^{-4n}
+
2^{-3n}
\left[
9\left(\frac{2}{3}\right)^n
+
4\left(\frac{1}{2}\right)^n
\right]=
2^{-3n}
\left[
9\left(\frac{2}{3}\right)^n
+
22\cdot 2^{-n}
\right].
\end{align*}
This proves the proposition. 
\end{proof}

\subsection{The norm gap between $\mc E_{\mathrm{shallow}}^{(\gamma)}$ and $\mc E_{\mathrm{phase}}$}\label{ap:k3-shallow-phase}

We now estimate the third-order norm gap $\bigl\|
\mb M_{\mc E_{\mathrm{shallow}}^{(\gamma)}}^{(3)}
-
\mb M_{\mc E_{\mathrm{phase}}}^{(3)}
\bigr\|_\infty$ as follows.

\begin{prop}\label{prop:twirled-third-moment-main}
There exist absolute constants $A_0,B_0,\gamma_0>0$ such that
\begin{equation}\label{eq:twirled-third-moment-main}
\left\|
 {\mb{M}}_{\mc E_{\mathrm{shallow}}^{(\gamma)}}^{(3)}
- {\mb{M}}_{\mc E_{\mathrm{phase}}}^{(3)}
\right\|_\infty
\le2^{-3n}\left(n^{A_0-B_0\gamma}+\mathrm{negl}(n)\right),
\end{equation}
for every fixed constant $\gamma>\gamma_0$ and all sufficiently large \(n\).
\end{prop}


We begin with a fixed-cardinality estimate for sums of local twirled blocks. Then we prove~\cref{prop:twirled-third-moment-main} officially.

\begin{lemma}[Fixed-cardinality block-sum bound]
\label{lem:fixed-cardinality-block-sum}
For every finite position set \(\Lambda\) and every integer
\(0\le m\le |\Lambda|\), define $\mb S_m(\Lambda)
:=
\sum_{I\subseteq \Lambda,\,|I|=m}
\widetilde B(\pi_1,\pi_2)^{\otimes I}\otimes \widetilde\Delta_3^{\otimes(\Lambda\setminus I)},$
where $\pi_1\in S_{3,\mathrm{odd}}$ and $\pi_2\in S_{3,\mathrm{even}}$.
Then
\begin{equation}
\|\mb S_m(\Lambda)\|_\infty
\le
\|(\widetilde B(\pi_1,\pi_2)+\widetilde\Delta_3)^{\otimes |\Lambda|}\|_\infty=(2/3)^{|\Lambda|}.
\end{equation}
\end{lemma}

\begin{proof}
By \cref{eq:widetildedelta,eq:widetildeB}, the
single-qubit operators $\widetilde\Delta_3$ and $\widetilde B(\pi_1,\pi_2)$ in $\mc B_3$ are
simultaneously block diagonal with respect to the decomposition into the symmetric
subspace $\mc H_{\mathrm{sym}}$ and the two nonsymmetric
subspaces.
More precisely,
\begin{equation}
\widetilde\Delta_3
=
\frac12\id_4\oplus\mathbf 0_2\oplus\mathbf 0_2,
\qquad
\widetilde B(\pi_1,\pi_2)
=
\frac16\id_4\oplus C\oplus C,
\end{equation}
where $C:=\frac13(A_{\pi_1}+A_{\pi_2})$.
Hence the three-copy Hilbert space associated with the sites in
\(\Lambda\) admits the orthogonal decomposition
\begin{equation}\label{eq:KJsubspace}
\bigotimes_{q\in\Lambda}(\mathbb C^2)^{\otimes3}
=
\bigoplus_{J\subseteq\Lambda}\mc K_J,
\qquad
\mc K_J
:=
\bigotimes_{q\in J}\mc H_{\mathrm{sym}}
\otimes
\bigotimes_{q\in\Lambda\setminus J}\mc H_{\mathrm{ns}} .
\end{equation}

Here
\(\mc H_{\mathrm{ns}}
=
\mc H_{\mathrm{std}}^{(0)}
\oplus
\mc H_{\mathrm{std}}^{(1)}\),
so both copies of the standard representation are included in
\(\mc K_J\). For each \(J\subseteq\Lambda\), let \(P_J\) denote the
orthogonal projector onto \(\mc K_J\), and write
\(X|_{\mc K_J}:=P_JXP_J\). for the restriction of $X$ to $\mc K_J$.  Since every $\mc K_J$ is invariant under both operators
\begin{equation}\label{eq:block-max-norm}
\begin{aligned}
\|\mb S_m(\Lambda)\|_\infty
&=
\max_{J\subseteq\Lambda}
\|\mb S_m(\Lambda)|_{\mc K_J}\|_\infty,\\
\|(\widetilde B+\widetilde\Delta_3)^{\otimes|\Lambda|}\|_\infty
&=
\max_{J\subseteq\Lambda}
\left\|
(\widetilde B+\widetilde\Delta_3)^{\otimes|\Lambda|}
|_{\mc K_J}
\right\|_\infty .
\end{aligned}
\end{equation}
By the definition of \(\mc K_J\), the local three-copy space at each site
\(q\in J\) is restricted to \(\mc H_{\mathrm{sym}}\), while at each site
\(q\in\Lambda\setminus J\) it is restricted to $\mc H_{\mathrm{ns}}$. Set
\(r:=|J|\), so there are \(r\) symmetric sites and \(|\Lambda|-r\)
nonsymmetric sites. By~\cref{eq:widetildedelta}, $\widetilde\Delta_3$ vanishes on the
nonsymmetric subspace. Consequently, the restriction to \(\mc K_J\) of a term $\widetilde B(\pi_1,\pi_2)^{\otimes I}
\otimes
\widetilde\Delta_3^{\otimes(\Lambda\setminus I)}$
can be nonzero only if $\Lambda\setminus I\subseteq J,$
or equivalently, \(\Lambda\setminus J\subseteq I\). Thus all \(|\Lambda|-r\)
nonsymmetric sites must carry \(\widetilde B(\pi_1,\pi_2)\). Since
\(|I|=m\), the number of symmetric sites carrying
\(\widetilde B(\pi_1,\pi_2)\) is $j:=|I\cap J|=m-(|\Lambda|-r).$
It follows that
\(\mb S_m(\Lambda)|_{\mc K_J}=0\) whenever
\(j\notin\{0,\ldots,r\}\).
When \(0\le j\le r\), there are exactly \(\binom{r}{j}\) choices from $\mb S_m(\Lambda)$ for the
\(j\) symmetric sites carrying \(\widetilde B(\pi_1,\pi_2)\). Therefore,
\begin{equation}\label{eq:Sm-restriction}
\mb S_m(\Lambda) |_{\mc K_J}
=
\binom{r}{j}
\left(\frac16\right)^j
\left(\frac12\right)^{r-j}
\id_{\mc H_{\mathrm{sym}}^{\otimes r}}
\otimes
(C\oplus C)^{\otimes(|\Lambda|-r)}.
\end{equation}
Indeed, \(\widetilde B(\pi_1,\pi_2)\) acts as
\(\frac16\id_4\) on each of the \(j\) selected symmetric sites, while
\(\widetilde\Delta_3\) acts as \(\frac12\id_4\) on the remaining \(r-j\)
symmetric sites. On every nonsymmetric site, only
\(\widetilde B(\pi_1,\pi_2)\) contributes, with restriction
\(C\oplus C\) to \(\mc H_{\mathrm{ns}}\).

On the other hand,
\(\widetilde B(\pi_1,\pi_2)+\widetilde\Delta_3\) acts as
\(\frac23\id_4\) on \(\mc H_{\mathrm{sym}}\) and as \(C\) on each
standard-representation subspace. Hence
\begin{equation}\label{eq:combined-restriction}
\left.
\bigl(
\widetilde B(\pi_1,\pi_2)+\widetilde\Delta_3
\bigr)^{\otimes |\Lambda|}
\right|_{\mc K_J}
=
\left(\frac23\right)^r
\id_{\mc H_{\mathrm{sym}}^{\otimes r}}
\otimes
(C\oplus C)^{\otimes(|\Lambda|-r)}.
\end{equation}
Thus the two restrictions in
\cref{eq:Sm-restriction,eq:combined-restriction} are nonnegative scalar
multiples of the same operator.

Moreover,
\begin{equation}
\binom{r}{j}
\left(\frac16\right)^j
\left(\frac12\right)^{r-j}
\le
\sum_{\ell=0}^{r}
\binom{r}{\ell}
\left(\frac16\right)^\ell
\left(\frac12\right)^{r-\ell}
=
\left(\frac23\right)^r.
\end{equation}
Therefore, $\|\mb S_m(\Lambda)|_{\mc K_J}\|_\infty
\le
\left\|
\left.
\bigl(
\widetilde B(\pi_1,\pi_2)+\widetilde\Delta_3
\bigr)^{\otimes |\Lambda|}
\right|_{\mc K_J}
\right\|_\infty,$
for every invariant subspace \(\mc K_J\). Taking the maximum over all invariant subspaces yields the claimed bound
\(\|\mb S_m(\Lambda)\|_\infty\le (2/3)^{|\Lambda|}\).
\end{proof}

\textbf{We now prove~\cref{prop:twirled-third-moment-main}.}

\begin{proof}

Fix a constant \(\gamma>\gamma_0\), and take \(n\) sufficiently large that
the standing condition \(0<2\gamma\log n/n<1\) holds. In particular, every
factor \(\left(1-2\gamma\log n/n\right)^{N_{\mb I}}\) appearing below is
nonnegative and at most one.

We first recall the notation entering the twirled block
expansion~\cref{eq:block-expansion-twirled}. The position data \(\mathfrak I\) specifies a partition of the \(n\) qubit sites, as defined in
\cref{eq:postion-data}. It specifies which local block is assigned to
each site and therefore determines the twirled tensor-product block
\(\widetilde B_{\mathfrak I}\) in
\cref{eq:def-B-frakI-m3}. The associated count data $\mb I$
records only the cardinalities of these position sets (see
\cref{eq:defI}). For fixed count data \(\mb I\),
\(\widetilde{\mc C}_{\mb I}\) denotes the collection of all
\(\widetilde B_{\mathfrak I}\) obtained from position data with these
cardinalities. Finally, the interaction
exponent \(N_{\mathfrak I}\) defined in~\cref{eq:def-N-frakI-m3} depends only
on these cardinalities, and its common value within
\(\widetilde{\mc C}_{\mb I}\) is denoted by \(N_{\mb I}\); see
\cref{eq:def-N-I-m3}.

Consider the third-order moment formula in \cref{eq:block-expansion-twirled}. For fixed count data $\mb I$, define
\begin{equation}\label{eq:tI-new}
\mb T(\mb I)
:=
\left(1-\frac{2\gamma\log n}{n}\right)^{N_{\mb I}}
\sum_{\widetilde B_{\mathfrak I}\in \widetilde{\mc C}_{\mb I}}
\widetilde B_{\mathfrak I},
\qquad\text{so that}\qquad
\mb M_{\mc E_{\mathrm{shallow}}^{(\gamma)}}^{(3)}
-
\mb M_{\mc E_{\mathrm{phase}}}^{(3)}
=
2^{-3n}\sum_{\mb I:\,N_{\mb I}>0}\mb T(\mb I),
\end{equation}
due to~\cref{eq:block-expansion-shallow-phase-gap}. Since \(\mb I\) consists of ten nonnegative block counts whose sum is \(n\),
the number of possible count data is $\binom{n+9}{9}=O(n^9).$
In particular, the number of count data satisfying \(N_{\mb I}>0\) is also at
most \(O(n^9)\).

The proof proceeds in three steps. First, we exploit the adjoint symmetry of the block classes and regroup the above sum into Hermitian pieces. Next, for individual count data, we encode the nine nontrivial block multiplicities by a $3\times 3$ chessboard and distinguish four cases according to the sizes of the dominant entries. Finally, in each case, we bound the corresponding contribution and sum over all possible count data. 

\medskip
\noindent
\paragraph{\textbf{Step 1: Hermitian regrouping.}}
Starting from~\cref{eq:tI-new}, the most direct approach would be to apply the
triangle inequality term by term, which gives
$\left\|
\mb M_{\mc E_{\mathrm{shallow}}^{(\gamma)}}^{(3)}
-
\mb M_{\mc E_{\mathrm{phase}}}^{(3)}
\right\|_\infty
\le
2^{-3n}\sum_{\mb I:\,N_{\mb I}>0}\|\mb T(\mb I)\|_\infty. $
This estimate is sufficient in the simpler regimes considered later, but it
does not retain the Hermitian structure needed in Case~4 of the next step, since the individual
operators \(\mb T(\mb I)\) are not Hermitian in general. We therefore first
regroup the contributions according to their adjoints.

For the local twirled blocks, taking the adjoint replaces each permutation by
its inverse. Every odd permutation in \(S_3\) is a transposition and
hence self-inverse, whereas $\pi_{(123)}^{-1}=\pi_{(132)}$. Therefore, the adjoint operation exchanges the two columns indexed by
\(\pi_{(123)}\) and \(\pi_{(132)}\), and leaves the remaining column unchanged. In particular, for all $\pi_1\in S_{3,\mathrm{odd}}$, 
\begin{equation}
    \widetilde B(\pi_1,\pi_{(123)})^\dagger
=
\widetilde B(\pi_1,\pi_{(132)}),
\qquad \widetilde B(\pi_1,\pi_{()})^\dagger
=
\widetilde B(\pi_1,\pi_{()}),\qquad
\widetilde\Delta_3^\dagger=\widetilde\Delta_3.
\end{equation}
Accordingly, for each count data \(\mb I\) introduced in~\cref{eq:defI}, we denote \(\mb I^\dagger\) by
exchanging the two entries $\left|I^{(\pi_1,\pi_{(123)})}\right|
\quad\text{and}\quad
\left|I^{(\pi_1,\pi_{(132)})}\right|,$ while leaving all remaining entries unchanged.

The expression for \(N_{\mb I}\) in~\cref{eq:def-N-I-m3} is invariant under
this exchange. Hence
\begin{equation}\label{eq:NI-dagger}
N_{\mb I^\dagger}
=
N_{\mb I}.
\end{equation}

Moreover, taking adjoints maps the block class
\(\widetilde{\mc C}_{\mb I}\) bijectively onto
\(\widetilde{\mc C}_{\mb I^\dagger}\). Together
with~\cref{eq:NI-dagger}, this gives
\begin{equation}\label{eq:TIdagger}
\mb T(\mb I)^\dagger
=
\mb T(\mb I^\dagger).
\end{equation}

For each count data $\mb I$, define the Hermitian paired contribution
\begin{equation}\label{eq:hermitian-paired-block}
\widehat{\mb T}(\mb I)
:=\frac12[\mb T(\mb I)+\mb T(\mb I^\dagger)].
\end{equation}
By~\cref{eq:TIdagger}, each
\(\widehat{\mb T}(\mb I)\) is Hermitian. This replacement preserves the total sum. Indeed, for every pair
\(\mb I,\mb I^\dagger\),
\[
\widehat{\mb T}(\mb I)
+
\widehat{\mb T}(\mb I^\dagger)
=
\mb T(\mb I)
+
\mb T(\mb I^\dagger).
\]

Moreover, by~\cref{eq:NI-dagger}, the value of \(N_{\mb I}\) is invariant under
\(\mb I\mapsto\mb I^\dagger\). Therefore,~\cref{eq:tI-new} can be rewritten as
\begin{equation}\label{eq:proof-main-sum-hermitian}
\mb M_{\mc E_{\mathrm{shallow}}^{(\gamma)}}^{(3)}
-
\mb M_{\mc E_{\mathrm{phase}}}^{(3)}
=
2^{-3n}
\sum_{\mb I:\,N_{\mb I}>0}
\widehat{\mb T}(\mb I).
\end{equation}
Applying the triangle inequality gives
\begin{equation}\label{eq:triangle-hermitian}
\left\|
\mb M_{\mc E_{\mathrm{shallow}}^{(\gamma)}}^{(3)}
-
\mb M_{\mc E_{\mathrm{phase}}}^{(3)}
\right\|_\infty
\le
2^{-3n}
\sum_{\mb I:\,N_{\mb I}>0}
\left\|
\widehat{\mb T}(\mb I)
\right\|_\infty.
\end{equation}
We therefore fix count data \(\mb I\) with \(N_{\mb I}>0\) and study the individual contributions \(\|\widehat{\mb T}(\mb I)\|_\infty\).
\medskip
\noindent

\paragraph{\textbf{Step 2: Chessboard parametrization.}}

In this step, we focus on bounding a single
\(\|\widehat{\mb T}(\mb I)\|_\infty\) in~\cref{eq:triangle-hermitian}. We first derive a preliminary bound of $\|\widehat{\mb T}(\mb I)\|_\infty$ in terms
of the cardinality \(\#\widetilde{\mc C}_{\mb I}\), and the exponent \(N_{\mb I}\).  By~\cref{eq:TIdagger,eq:hermitian-paired-block}, $\|\widehat{\mb T}(\mb I)\|_\infty
\le
\|\mb T(\mb I)\|_\infty$ using the triangle inequality.

\begin{equation}\label{eq:k3-basic-count-bound}
\begin{split}
\|\widehat{\mb T}(\mb I)\|_\infty\leq \|\mb T(\mb I)\|_\infty
&\le
\left(1-\frac{2\gamma\log n}{n}\right)^{N_{\mb I}}
\sum_{\widetilde B_{\mathfrak I}\in\widetilde{\mc C}_{\mb I}}
\|\widetilde B_{\mathfrak I}\|_\infty \qquad\left(\text{using triangle inequality}\right)\\
&\le
\left(1-\frac{2\gamma\log n}{n}\right)^{N_{\mb I}}
\sum_{\widetilde B_{\mathfrak I}\in\widetilde{\mc C}_{\mb I}}
\|\widetilde\Delta_3^{\otimes I^{(0)}}
\otimes
\bigotimes_{\pi_1\in S_{3,\mathrm{odd}},\,\pi_2\in S_{3,\mathrm{even}}}
\widetilde B(\pi_1,\pi_2)^{\otimes I^{(\pi_1,\pi_2)}}\|_\infty \qquad\left(\text{using}~\cref{eq:def-B-frakI-m3}\right)\\
&=
\left(1-\frac{2\gamma\log n}{n}\right)^{N_{\mb I}}\#\widetilde{\mc C}_{\mb I}\times 
\|\widetilde\Delta_3\|_\infty^{|I^{(0)}|}\times
\prod_{\pi_1\in S_{3,\mathrm{odd}},\,\pi_2\in S_{3,\mathrm{even}}}
\|\widetilde B(\pi_1,\pi_2)\|_\infty^{|I^{(\pi_1,\pi_2)}|}\\
&=
\left(1-\frac{2\gamma\log n}{n}\right)^{N_{\mb I}}\#\widetilde{\mc C}_{\mb I}\times 
(\frac 12)^{|I^{(0)}|}\times (\frac 23)^{n-|I^{(0)}|}\qquad\left(\text{using}~\cref{eq:k3-delta-local-norm}\text{ and }\cref{eq:blockinfty2-m3}\right)\\
&\leq \left(1-\frac{2\gamma\log n}{n}\right)^{N_{\mb I}}\#\widetilde{\mc C}_{\mb I}\times (\frac 23)^{n}.
\end{split}
\end{equation}

To apply~\cref{eq:k3-basic-count-bound}, it remains to control the exponent \(N_{\mb I}\)
and the multiplicity
\(\#\widetilde{\mc C}_{\mb I}\). A direct difficulty is that the count data \(\mb I\) contain ten coupled parameters: one count $|I^{(0)}|$ associated with
\(\widetilde\Delta_3\) and nine counts $|I^{(\pi_1,\pi_2)}|$ associated with the nontrivial local
blocks \(\widetilde B(\pi_1,\pi_2)\). Fortunately, the dependence of the relevant bounds on these counts can
be controlled by only two or three suitably chosen dominant parameters. To
make this reduction transparent, we first organize the nine nontrivial counts
into a \(3\times3\) chessboard.
The rows are indexed by the three odd permutations, and the columns
by the three even permutations. For convenience, we write $n_{\mu,\nu}
:=
\left|I^{(\mu,\nu)}\right|,
\mu\in S_{3,\mathrm{odd}},
\nu\in S_{3,\mathrm{even}}.$ Relabeling the remaining rows and columns if
necessary, we may therefore arrange the nine blocks as
\begin{equation}\label{eq:chessboard-new}
\renewcommand{\arraystretch}{1.25}
\begin{array}{c|c|c}
\mathfrak a & \mathfrak b & \mathfrak c\\
\hline
\mathfrak d & \mathfrak e & \mathfrak f\\
\hline
\mathfrak g & \mathfrak h & \mathfrak i
\end{array}
=
\begin{array}{c|c|c}
(\pi_1,\pi_2)
&
(\pi_1,\pi_2')
&
(\pi_1,\pi_2'')
\\
\hline
(\pi_1',\pi_2)
&
(\pi_1',\pi_2')
&
(\pi_1',\pi_2'')
\\
\hline
(\pi_1'',\pi_2)
&
(\pi_1'',\pi_2')
&
(\pi_1'',\pi_2'')
\end{array}.
\end{equation}

For a chessboard label
\(\mathfrak x\in
\{\mathfrak a,\mathfrak b,\mathfrak c,
\mathfrak d,\mathfrak e,\mathfrak f,
\mathfrak g,\mathfrak h,\mathfrak i\}\)
corresponding to the pair \((\mu,\nu)\), we use the shorthand
\[
n_0:=|I^{(0)}|,
\qquad
I^{(\mathfrak x)}
:=
I^{(\mu,\nu)},
\qquad
n_{\mathfrak x}
:=
|I^{(\mathfrak x)}|,
\qquad
B(\mathfrak x)
:=
B(\mu,\nu),
\qquad
\widetilde B(\mathfrak x)
:=
\widetilde B(\mu,\nu).
\]

Among all pairs of counts lying in different rows and different columns,
choose a pair whose product is maximal. 
After relabeling the rows and columns
if necessary, we may assume that these counts are
\(n_{\mathfrak a}\) and \(n_{\mathfrak e}\). Define
\begin{equation}\label{eq:def-m1-m2}
m_1
:=
n_{\mathfrak a},
\qquad
m_2
:=
n_{\mathfrak e},
\qquad
m_1\ge m_2.
\end{equation}
By construction,
\begin{equation}\label{eq:m1m2max}
m_1m_2
=
\max_{\substack{
\mu\neq\mu'\\
\nu\neq\nu'
}}
n_{\mu,\nu}\,n_{\mu',\nu'}.
\end{equation}
We now use the chessboard parametrization to control the two quantities in
\cref{eq:k3-basic-count-bound}, beginning with the interaction exponent
\(N_{\mb I}\):
\begin{equation}\label{eq:NI-bound}
    N_{\mb I}\geq m_1m_2,
\end{equation}
since~\cref{eq:def-N-I-m3} contains the products of counts located in different rows and different columns. Thus, $N_{\mb I}$ is lower bounded using only the two dominant parameters.

We next use the chessboard parametrization to control the multiplicity
\(\#\widetilde{\mc C}_{\mb I}\) in~\cref{eq:k3-basic-count-bound}. For fixed
count data, choosing the nine nontrivial position sets successively gives
\begin{equation}\label{eq:exact-multiplicity-new}
\#\widetilde{\mc C}_{\mb I}
=
\binom{n}{n_{\mathfrak a}}
\binom{n-n_{\mathfrak a}}{n_{\mathfrak b}}
\binom{n-n_{\mathfrak a}-n_{\mathfrak b}}{n_{\mathfrak c}}
\cdots
\binom{
n-\sum_{\mathfrak x\prec\mathfrak i}n_{\mathfrak x}
}{
n_{\mathfrak i}
}.
\end{equation}
Here, \(\prec\) denotes the fixed ordering
\(\mathfrak a\prec\mathfrak b\prec\cdots\prec\mathfrak i\), so that
\(\sum_{\mathfrak x\prec\mathfrak i}n_{\mathfrak x}\) runs over all block
types preceding \(\mathfrak i\).

Although~\cref{eq:exact-multiplicity-new} is exact, it still depends on all
nine nontrivial counts. We now show that these counts can be controlled using
\(m_1,m_2\) and one additional parameter. Define
\[
m_3
:=
\max\left\{
n_{\mathfrak b},
n_{\mathfrak c},
n_{\mathfrak d},
n_{\mathfrak g}
\right\},
\]
so that \(m_3\) is the largest count sharing either the row or the column of
\(\mathfrak a\), excluding \(n_{\mathfrak a}\) itself.

Once the maximizing pair
\((n_{\mathfrak a},n_{\mathfrak e})=(m_1,m_2)\) is fixed, the four
positions entering the definition of \(m_3\) split into two symmetry
classes: $\{\mathfrak b,\mathfrak d\},
\{\mathfrak c,\mathfrak g\}.$
Interchanging the row and column roles maps \(\mathfrak d\) to
\(\mathfrak b\) and \(\mathfrak g\) to \(\mathfrak c\). Hence, after
this symmetry reduction, it is enough to treat the two representative
orientations
\[
m_3=n_{\mathfrak b}
\qquad\text{or}\qquad
m_3=n_{\mathfrak c}.
\]
Because the fixed count data satisfy \(N_{\mb I}>0\), the maximal
product obeys \(m_1m_2>0\), and in particular \(m_2>0\). In the
\(\mathfrak c\)-orientation, \(n_{\mathfrak c}=m_3\) and
\(n_{\mathfrak e}=m_2\) lie in different rows and columns. The
maximality of \(m_1m_2\) therefore gives
\[
m_3m_2
=
n_{\mathfrak c}m_2
\le
m_1m_2,
\qquad\text{and hence}\qquad
m_3\le m_1.
\]
Hence, the maximality properties
defining \(m_1,m_2,m_3\) are unchanged up to a relabeling of columns, and all
bounds below apply uniformly to both members of the Hermitian pair.

We now use the value of \(m_1,m_2,m_3\) to control the nine chessboard
counts. Since each of
\(\mathfrak c,\mathfrak g,\mathfrak i\) lies in a row and a column different
from those of \(\mathfrak e\),~\cref{eq:m1m2max} gives
\begin{equation}\label{eq:board-bound-1}
n_{\mathfrak a},\,
n_{\mathfrak c},\,
n_{\mathfrak g},\,
n_{\mathfrak i}
\le
m_1.
\end{equation}
By the definition of \(m_3\),
\begin{equation}\label{eq:board-bound-2}
n_{\mathfrak b},\,
n_{\mathfrak c},\,
n_{\mathfrak d},\,
n_{\mathfrak g}
\le
m_3.
\end{equation}
Similarly, each of
\(\mathfrak f,\mathfrak h,\mathfrak i\) lies in a row and a column different
from those of \(\mathfrak a\), and hence
\begin{equation}\label{eq:board-bound-3}
n_{\mathfrak e},\,
n_{\mathfrak f},\,
n_{\mathfrak h},\,
n_{\mathfrak i}
\le
m_2\le m_1.
\end{equation}
Finally, suppose that \(m_3=n_{\mathfrak b}>0\). Each of
\(\mathfrak d,\mathfrak f,\mathfrak g,\mathfrak i\) lies in a row
and a column different from those of \(\mathfrak b\). The maximality
in~\cref{eq:m1m2max} therefore implies
\begin{equation}\label{eq:board-bound-4}
n_{\mathfrak d},\,
n_{\mathfrak f},\,
n_{\mathfrak g},\,
n_{\mathfrak i}
\le
\frac{m_1m_2}{m_3}.
\end{equation}
If instead \(m_3=n_{\mathfrak c}>0\), the positions
\(\mathfrak d,\mathfrak e,\mathfrak g,\mathfrak h\) lie in rows and
columns different from those of \(\mathfrak c\). The same maximality
argument gives
\begin{equation}\label{eq:board-bound-4c}
n_{\mathfrak d},
n_{\mathfrak e},
n_{\mathfrak g},
n_{\mathfrak h}
\le
\frac{m_1m_2}{m_3}.
\end{equation}

Consequently,
\cref{eq:board-bound-1,eq:board-bound-2,eq:board-bound-3}, together
with \cref{eq:board-bound-4} in the \(\mathfrak b\)-orientation or
\cref{eq:board-bound-4c} in the \(\mathfrak c\)-orientation, reduce
the original nine-variable multiplicity estimate to bounds involving
only \(m_1,m_2,m_3\).

We now divide the remaining analysis into four cases according to the sizes of
the dominant count \(m_1\) and the auxiliary count \(m_3\). We choose constants $0<\alpha<\alpha'=\beta=\beta'<\frac12,$
whose explicit values will be given at the end of the proof. With these
choices, the following four cases exhaust all possible regimes of
\(m_1\) and \(m_3\).

\medskip
\noindent
\textbf{Case 1: \(m_1\le \alpha n\) and \(m_3\le \beta n\).}
In this regime, all relevant block counts are small, so the direct bound
in~\cref{eq:k3-basic-count-bound} is already sufficient.
\begin{equation}\label{eq:case1-starting-bound}
\left\|
\widehat{\mb T}(\mb I)
\right\|_\infty
\le
\#\widetilde{\mc C}_{\mb I}
\left(\frac23\right)^n = \binom{n}{n_{\mathfrak a}}
\binom{n-n_{\mathfrak a}}{n_{\mathfrak b}}
\binom{n-n_{\mathfrak a}-n_{\mathfrak b}}{n_{\mathfrak c}}
\cdots
\binom{
n-\sum_{\mathfrak x\prec\mathfrak i}n_{\mathfrak x}
}{
n_{\mathfrak i}
}\left(\frac23\right)^n
\leq \prod_{\mathfrak x\in \{\mathfrak a,...,\mathfrak i\}}\binom{n}{n_{\mathfrak x}}\left(\frac23\right)^n.
\end{equation}

Since \(m_2\le m_1\), the assumptions of Case~1 together
with~\cref{eq:board-bound-1,eq:board-bound-2,eq:board-bound-3} imply
\[
n_{\mathfrak a},
n_{\mathfrak c},
n_{\mathfrak e},
n_{\mathfrak f},
n_{\mathfrak g},
n_{\mathfrak h},
n_{\mathfrak i}
\le
\alpha n,
\qquad
n_{\mathfrak b},
n_{\mathfrak d}
\le
\beta n.
\]
Assuming \(0<\alpha,\beta<1/2\), the exact multiplicity formula
in~\cref{eq:exact-multiplicity-new} and the entropy bound $\binom{n}{k}
\le
2^{nH_2(k/n)}$, where \(H_2\) denotes the binary entropy,
give
\begin{equation}\label{eq:case1-combined-bound}
\left\|
\widehat{\mb T}(\mb I)
\right\|_\infty
\le\prod_{\mathfrak x\in \{\mathfrak a,...,\mathfrak i\}}\binom{n}{n_{\mathfrak x}}\left(\frac23\right)^n\leq
2^{\,n[7H_2(\alpha)+2H_2(\beta)]}
\left(\frac23\right)^n.
\end{equation}

Choose \(\alpha,\beta>0\) sufficiently small so that
\begin{equation}\label{eq:anc_const1}
7H_2(\alpha)+2H_2(\beta)
<
\log_2\!\left(\frac32\right).
\end{equation}
Then there exists a constant $\delta_1
:=
\log_2\!\left(\frac32\right)
-
7H_2(\alpha)
-
2H_2(\beta)
>
0$ such that
\begin{equation}\label{eq:case1-final-polished}
\left\|
\widehat{\mb T}(\mb I)
\right\|_\infty
\le
2^{-\delta_1 n}
=
\mathrm{negl}(n).
\end{equation}

\medskip
\noindent
\textbf{Case 2: \(m_1>\alpha n\) and \(m_3\le \beta' n\).}

In this regime, the direct multiplicity estimate used in Case~1 is no longer
sufficient. Indeed, \(m_1=n_{\mathfrak a}>\alpha n\), and therefore the
multiplicity term in~\cref{eq:k3-basic-count-bound} contains the factor $\binom{n}{n_{\mathfrak a}}.$
This factor can be exponentially large and cannot, in general, be absorbed by
the local decay \((2/3)^n\). We therefore use a block-combination argument to remove this
large combinatorial contribution.

For fixed choices of the pairwise disjoint position sets
\(I^{(\mathfrak b)},\ldots,I^{(\mathfrak i)}\), we set $\Lambda
=
[n]\setminus
\bigsqcup_{\mathfrak x\in\{\mathfrak b,\ldots,\mathfrak i\}}
I^{(\mathfrak x)}.$
Then, for every term in the inner sum, $\Lambda
=
I^{(0)}\sqcup I^{(\mathfrak a)}.$
For fixed choices of
\(I^{(\mathfrak b)},\ldots,I^{(\mathfrak i)}\), the contribution of the
\(\mathfrak a\)-block and the residual \(\widetilde\Delta_3\)-positions is $\mb{S}_{n_{\mathfrak a}}(\Lambda)\coloneq \sum_{\substack{
I^{(\mathfrak a)}\subseteq\Lambda\\
|I^{(\mathfrak a)}|=n_{\mathfrak a}
}}
\widetilde B(\mathfrak a)^{\otimes I^{(\mathfrak a)}}
\otimes
\widetilde\Delta_3^{\otimes
(\Lambda\setminus I^{(\mathfrak a)})}.$
Applying~\cref{lem:fixed-cardinality-block-sum} to this fixed-cardinality
sum gives
\begin{equation}\label{eq:case2-block-combination}
\|\mb{S}_{n_{\mathfrak a}}(\Lambda)\|_\infty\le
\left(\frac23\right)^{|\Lambda|}.
\end{equation}

The remaining eight local blocks satisfy $\|\widetilde B(\mathfrak x)\|_\infty
\le
\frac23,\text{ for }
\mathfrak x\in
\{\mathfrak b,\mathfrak c,\mathfrak d,\mathfrak e,
\mathfrak f,\mathfrak g,\mathfrak h,\mathfrak i\}$. Since $|\Lambda|
+
\sum_{\mathfrak x\in
\{\mathfrak b,\ldots,\mathfrak i\}}
n_{\mathfrak x}
=
n,$
each fixed outer configuration contributes at most \((2/3)^n\).
Therefore, we have 

\begin{equation}\label{eq:case2-T-decomposition}
\begin{split}
\mb T(\mb I)
&=
\left(1-\frac{2\gamma\log n}{n}\right)^{N_{\mb I}}
\sum_{\substack{
I^{(0)},I^{(\mathfrak a)},\ldots,I^{(\mathfrak i)}
\subseteq[n]\\
I^{(0)}\sqcup I^{(\mathfrak a)}\sqcup\cdots\sqcup
I^{(\mathfrak i)}=[n]\\
|I^{(0)}|=n_0,\;
|I^{(\mathfrak x)}|=n_{\mathfrak x}
}}
\widetilde\Delta_3^{\otimes I^{(0)}}
\otimes
\bigotimes_{\mathfrak x\in\{\mathfrak a,\ldots,\mathfrak i\}}
\widetilde B(\mathfrak x)^{\otimes I^{(\mathfrak x)}}
\\
&=
\left(1-\frac{2\gamma\log n}{n}\right)^{N_{\mb I}}
\sum_{\substack{
I^{(\mathfrak b)},\ldots,I^{(\mathfrak i)}\subseteq[n]\\
|I^{(\mathfrak x)}|=n_{\mathfrak x},
\ \mathfrak x\in\{\mathfrak b,\ldots,\mathfrak i\}\\
\text{pairwise disjoint}
}}
\left[
\sum_{\substack{
I^{(\mathfrak a)}\subseteq\Lambda\\
|I^{(\mathfrak a)}|=n_{\mathfrak a}
}}
\widetilde B(\mathfrak a)^{\otimes I^{(\mathfrak a)}}
\otimes
\widetilde\Delta_3^{
\otimes(\Lambda\setminus I^{(\mathfrak a)})}
\right]
\otimes
\bigotimes_{\mathfrak x\in\{\mathfrak b,\ldots,\mathfrak i\}}
\widetilde B(\mathfrak x)^{\otimes I^{(\mathfrak x)}}
\\
&=
\left(1-\frac{2\gamma\log n}{n}\right)^{N_{\mb I}}
\sum_{\substack{
I^{(\mathfrak b)},\ldots,I^{(\mathfrak i)}\subseteq[n]\\
|I^{(\mathfrak x)}|=n_{\mathfrak x},
\ \mathfrak x\in\{\mathfrak b,\ldots,\mathfrak i\}\\
\text{pairwise disjoint}
}}
\mb S_{n_{\mathfrak a}}(\Lambda)
\otimes
\bigotimes_{\mathfrak x\in\{\mathfrak b,\ldots,\mathfrak i\}}
\widetilde B(\mathfrak x)^{\otimes I^{(\mathfrak x)}}.
\end{split}
\end{equation}
and
\begin{equation}\label{eq:case2-master-bound}
\begin{split}
&\|\widehat{\mb T}(\mb I)\|_\infty\le\|{\mb T}(\mb I)\|_\infty\le\left(1-\frac{2\gamma\log n}{n}\right)^{N_{\mb I}}
\sum_{\substack{
I^{(\mathfrak b)},\ldots,I^{(\mathfrak i)}\subseteq[n]\\
|I^{(\mathfrak x)}|=n_{\mathfrak x},
\ \mathfrak x\in\{\mathfrak b,\ldots,\mathfrak i\}\\
\text{pairwise disjoint}
}} (\frac 23)^n\\
&\le\left(1-\frac{2\gamma\log n}{n}\right)^{N_{\mb I}}(\frac 23)^n\times
\binom{n}{n_{\mathfrak b}}
\binom{n-n_{\mathfrak b}}{n_{\mathfrak c}}
\binom{n-n_{\mathfrak b}-n_{\mathfrak c}}{n_{\mathfrak d}}
\cdots
\binom{
n-\sum_{\mathfrak x\in
\{\mathfrak b,\ldots,\mathfrak h\}}
n_{\mathfrak x}
}{
n_{\mathfrak i}
}.\\
\end{split}
\end{equation}

Compared with~\cref{eq:case1-starting-bound}, the bound
in~\cref{eq:case2-master-bound} no longer contains the binomial factor
associated with \(n_{\mathfrak a}=m_1\). This is precisely the gain obtained
from the block-combination estimate using~\cref{lem:fixed-cardinality-block-sum}, and it is
essential here because \(m_1>\alpha n\) may be macroscopically large. 
We now bound~\cref{eq:case2-master-bound} by distinguishing two subcases.

\medskip
\noindent
\emph{Subcase 2.1: \(m_1m_2\ge n^2/\log n\).} By~\cref{eq:m1m2max}, $N_{\mb I}\ge m_1m_2.$
Hence,
\begin{equation}\label{eq:case2-large-sparse-decay}
\left(1-\frac{2\gamma\log n}{n}\right)^{N_{\mb I}}
\le
\exp\left(
-\frac{2\gamma\log n}{n}m_1m_2
\right)
\le
e^{-2\gamma n}.
\end{equation}

The remaining multiplicity in~\cref{eq:case2-master-bound} contains only the
eight non-\(\mathfrak a\) block types. It is bounded by the total number of
assignments of these eight types and one residual class, hence $\binom{n}{n_{\mathfrak b}}
\binom{n-n_{\mathfrak b}}{n_{\mathfrak c}}
\cdots
\binom{
n-\sum_{\mathfrak x\in\{\mathfrak b,\ldots,\mathfrak h\}}
n_{\mathfrak x}
}{
n_{\mathfrak i}
}
\le9^n$.
Therefore, $\|\widehat{\mb T}(\mb I)\|_\infty
\le
e^{-2\gamma n}
\left(\frac23\right)^n
9^n=(6e^{-2\gamma})^n$.
For \(\gamma>1\), one has \(6e^{-2\gamma}<1\), and hence this bound is
exponentially small in \(n\).

\medskip
\noindent
\emph{Subcase 2.2: \(m_1m_2<n^2/\log n\).} Since \(m_1>\alpha n\), we obtain $m_2<
\frac{n}{\alpha\log n}$.
Together with
\cref{eq:board-bound-2,eq:board-bound-3},
\begin{equation}
    n_{\mathfrak b},
n_{\mathfrak c},
n_{\mathfrak d},
n_{\mathfrak g}
\le
\beta'n,\qquad n_{\mathfrak e},
n_{\mathfrak f},
n_{\mathfrak h},
n_{\mathfrak i}
\le
m_2
<
\frac{n}{\alpha\log n}.
\end{equation}
Assume \(0<\beta'<1/2\). By monotonicity of the binary entropy function on
\([0,1/2]\), the four counts bounded by \(\beta'n\) contribute at most
\(2^{4nH_2(\beta')}\).
For the remaining four counts, using $\binom nk\le
\left(\frac{en}{k}\right)^k,$
we have uniformly for
\(k\le n/(\alpha\log n)\), $\binom nk
\le
\exp\left(
O\left(\frac{n\log\log n}{\log n}\right)
\right)
=
2^{o(n)}.$
Hence the multiplicity factor in~\cref{eq:case2-master-bound} is bounded by $2^{4nH_2(\beta')+o(n)}.$
Consequently,
\begin{equation}\label{eq:case2-final-bound}
\|\widehat{\mb T}(\mb I)\|_\infty
\le
2^{4nH_2(\beta')+o(n)}
\left(\frac23\right)^n .
\end{equation}

Choose \(\beta'>0\) sufficiently small such that
\begin{equation}\label{eq:anc_const2}
4H_2(\beta')
<
\log_2\left(\frac32\right).
\end{equation}
Then there exists a constant \(\delta_2>0\), independent of \(n\), such that,
for sufficiently large \(n\),$\|\widehat{\mb T}(\mb I)\|_\infty
\le
2^{-\delta_2 n}
=
\mathrm{negl}(n).$

Both subcases therefore give an  exponentially small bound. Consequently,
\begin{equation}\label{eq:case2-final-polished}
\|\widehat{\mb T}(\mb I)\|_\infty
=
\mathrm{negl}(n).
\end{equation}

\medskip
\noindent
\textbf{Case 3:
\(m_3>\min\{\beta,\beta'\}n\) and \(m_1<\alpha'n\).}

Only the orientation \(m_3=n_{\mathfrak b}\) can occur in
Case~3. Indeed, in the \(\mathfrak c\)-orientation we have
\(m_3\le m_1\). Since \(\beta=\beta'=\alpha'\), the assumptions of
Case~3 would require simultaneously
\[
m_3>\alpha'n
\qquad\text{and}\qquad
m_1<\alpha'n,
\]
which is impossible when \(m_3\le m_1\). Therefore
\(m_3=n_{\mathfrak b}\) throughout this case. The factor associated
with \(n_{\mathfrak b}=m_3\) may now be exponentially large, so we
apply the fixed-cardinality block-combination argument to the
\(\mathfrak b\)-block and the residual
\(\widetilde\Delta_3\)-positions. We use the same construction as in
Case~2, with the roles of \(n_{\mathfrak a}\) and
\(n_{\mathfrak b}\) interchanged.

For fixed choices of the pairwise disjoint position sets $\{I^{(\mathfrak a)},I^{(\mathfrak c)},I^{(\mathfrak d)},
I^{(\mathfrak e)},I^{(\mathfrak f)},I^{(\mathfrak g)},
I^{(\mathfrak h)},I^{(\mathfrak i)}\}$,
we set $\Lambda
:=
[n]\setminus
\bigsqcup_{\mathfrak x\in
\{\mathfrak a,\mathfrak c,\mathfrak d,\mathfrak e,
\mathfrak f,\mathfrak g,\mathfrak h,\mathfrak i\}}
I^{(\mathfrak x)}$ in this case.
Then \(\Lambda=I^{(0)}\sqcup I^{(\mathfrak b)}\). Applying
\cref{lem:fixed-cardinality-block-sum} to a new $\mb S_{n_{\mathfrak b}}(\Lambda)
:=
\sum_{\substack{
I^{(\mathfrak b)}\subseteq\Lambda\\
|I^{(\mathfrak b)}|=n_{\mathfrak b}
}}
\widetilde B(\mathfrak b)^{\otimes I^{(\mathfrak b)}}
\otimes
\widetilde\Delta_3^{\otimes
(\Lambda\setminus I^{(\mathfrak b)})}$
gives $\left\|
\mb S_{n_{\mathfrak b}}(\Lambda)
\right\|_\infty
\le
\left(\frac23\right)^{|\Lambda|}.$
Proceeding exactly as in~\cref{eq:case2-T-decomposition,eq:case2-master-bound},
we obtain the analogue of the Case~2 master bound, with the binomial factor
associated with \(n_{\mathfrak b}\) rather than \(n_{\mathfrak a}\) removed.

We again distinguish whether \(m_1m_2\) is larger or smaller than
\(n^2/\log n\). If $m_1m_2\ge\frac{n^2}{\log n}$,
following the argument of Subcase~2.1 gives an exponentially small bound for \(\gamma>1\).

Suppose instead that $m_1m_2<\frac{n^2}{\log n}.$ Since \(m_3>\min\{\beta,\beta'\} n\),
\cref{eq:board-bound-1,eq:board-bound-3,eq:board-bound-4} imply
\begin{equation}\label{eq:case3-count-bounds}
\begin{gathered}
n_{\mathfrak a},
n_{\mathfrak c},
n_{\mathfrak e},
n_{\mathfrak h}
\le
m_1
<
\alpha'n,
\\
n_{\mathfrak d},
n_{\mathfrak f},
n_{\mathfrak g},
n_{\mathfrak i}
\le
\frac{m_1m_2}{m_3}
<
\frac{n}{\min\{\beta,\beta'\}\log n}.
\end{gathered}
\end{equation}
Assuming \(0<\alpha'<1/2\), the first four counts contribute at most
\(2^{4nH_2(\alpha')}\), while the remaining four contribute \(2^{o(n)}\).
Consequently,
\[
\left\|
\widehat{\mb T}(\mb I)
\right\|_\infty
\le
2^{4nH_2(\alpha')+o(n)}
\left(\frac23\right)^n.
\]
Choose \(\alpha'>0\) sufficiently small so that
\begin{equation}\label{eq:anc_const3}
4H_2(\alpha')
<
\log_2\!\left(\frac32\right).
\end{equation}
Then the last expression is exponentially small for sufficiently large \(n\).

Combining the two subcases, we conclude that
\begin{equation}\label{eq:case3-final-polished}
\left\|
\widehat{\mb T}(\mb I)
\right\|_\infty
=
\mathrm{negl}(n).
\end{equation}

\textbf{Case 4: \(m_3>\min\{\beta,\beta'\}n\) and
\(m_1\ge \alpha'n\).}

This is the most challenging regime, since both \(m_1\) and \(m_3\) are
macroscopically large. The single-block combination arguments used in
Cases~2 and~3 are no longer sufficient. We instead combine the blocks
associated with \(\mathfrak a,\mathfrak b,\mathfrak c\), together with
\(\widetilde\Delta_3\).

Our goal remains to bound
\(\|\widehat{\mb T}(\mb I)\|_\infty\). To do so, we return temporarily to
the untwirled blocks, whose Boolean-support structure is explicit in the
computational basis. Define
\begin{equation}\label{eq:k3-XI}
\mb X(\mb I)
:=
\left(1-\frac{2\gamma\log n}{n}\right)^{N_{\mb I}}
\sum_{B_{\mathfrak I}\in\mc C_{\mb I}}
B_{\mathfrak I},
\end{equation}
and its Hermitian pairing
\begin{equation}\label{eq:k3-XI-hat}
\widehat{\mb X}(\mb I)
:=
\frac12\left[
\mb X(\mb I)+\mb X(\mb I^\dagger)
\right].
\end{equation}
By the same argument as in the proof of~\cref{eq:TIdagger},
\[
\mb X(\mb I^\dagger)=\mb X(\mb I)^\dagger,
\]
and hence \(\widehat{\mb X}(\mb I)\) is Hermitian.

Let \(\mc T_3\) denote the single-qubit three-copy Clifford twirling channel.
By the definitions of the twirled blocks and the linearity of \(\mc T_3\),
\[
\mb T(\mb I)
=
\mc T_3^{\otimes n}\bigl(\mb X(\mb I)\bigr),
\qquad
\widehat{\mb T}(\mb I)
=
\mc T_3^{\otimes n}\bigl(\widehat{\mb X}(\mb I)\bigr).
\]
Since \(\mc T_3^{\otimes n}\) is a convex combination of unitary
conjugations, it is contractive in the operator norm. Therefore,
\begin{equation}\label{eq:case4-contractive}
\|\widehat{\mb T}(\mb I)\|_\infty
\le
\|\widehat{\mb X}(\mb I)\|_\infty.
\end{equation}
Thus, it suffices to bound the untwirled Hermitian operator
\(\widehat{\mb X}(\mb I)\), for which the computational-basis
Boolean-support structure can be exploited directly.


By~\cref{Ap:moment3}, every matrix entry of
\(\mb X(\mb I)\) and
\(\mb X(\mb I^\dagger)\) in the computational basis is real and
nonnegative. Indeed, each block \(B_{\mathfrak I}\) is a
Boolean-support operator in the computational basis, and the prefactor $\left(1-\frac{2\gamma\log n}{n}\right)^{N_{\mb I}}$ in~\cref{eq:k3-XI}
is nonnegative by the standing parameter regime. Therefore,
\(\widehat{\mb X}(\mb I)\) is a Hermitian matrix with entrywise
nonnegative matrix elements in the computational basis.

By the Perron--Frobenius theorem, there exists a unit vector
\(\ket x\) with nonnegative computational-basis entries such that
\begin{equation}\label{eq:case4-pf-rayleigh}
\|\widehat{\mb X}(\mb I)\|_\infty
=
\rho(\widehat{\mb X}(\mb I))
=
\bra{x}\widehat{\mb X}(\mb I)\ket{x},
\end{equation}
where \(\rho(A)\) denotes the spectral radius of \(A\).
Since \(\widehat{\mb X}(\mb I)\) is Hermitian and entrywise
nonnegative, its Perron eigenvalue equals its spectral radius, and its
spectral radius equals its operator norm. No irreducibility assumption is
needed for this conclusion.

Furthermore, since both \(\mb X(\mb I)\) and \(\ket{x}\) are real,
\begin{equation}\label{eq:case4-reduce-to-X}
\begin{split}
\|\widehat{\mb X}(\mb I)\|_\infty
&=
\frac12
\left(
\bra{x}\mb X(\mb I)\ket{x}
+
\bra{x}\mb X(\mb I)^\dagger\ket{x}
\right)=
\bra{x}\mb X(\mb I)\ket{x}.
\end{split}
\end{equation}

Combining
\cref{eq:case4-contractive,eq:case4-pf-rayleigh,eq:case4-reduce-to-X}, we obtain
\begin{equation}\label{eq:tIleqxI}
    \|\widehat{\mb T}(\mb I)\|_\infty
\le
\|\widehat{\mb X}(\mb I)\|_\infty = \bra{x}\mb X(\mb I)\ket{x}.
\end{equation}
Thus, it remains to bound the quadratic form of the single contribution
\(\mb X(\mb I)\) on a nonnegative unit vector \(\ket{x}\). For Boolean-support operators \(A\) and \(B\), we write
\(A\subseteq B\) when their computational-basis matrix elements satisfy  $0\le A_{u,v}\le B_{u,v}, 
\text{ for all computational-basis indices \(u,v\)}.$
In particular, this relation implies inclusion of their
computational-basis supports. Since \(\ket{x}\) has nonnegative
computational-basis entries, \(A\subseteq B\) implies
\begin{equation}\label{eq:subseteqUnitx}
\bra{x}A\ket{x}
\le
\bra{x}B\ket{x}.
\end{equation}

For each fixed choice of the six pairwise disjoint position sets $I^{(\mathfrak d)},I^{(\mathfrak e)},I^{(\mathfrak f)},
I^{(\mathfrak g)},I^{(\mathfrak h)},I^{(\mathfrak i)},$
define $\Lambda
:=
[n]\setminus
\bigsqcup_{\mathfrak x\in
\{\mathfrak d,\mathfrak e,\mathfrak f,
\mathfrak g,\mathfrak h,\mathfrak i\}}
I^{(\mathfrak x)}.$
The remaining qubit sites are then partitioned as $\Lambda
=
I^{(0)}
\sqcup I^{(\mathfrak a)}
\sqcup I^{(\mathfrak b)}
\sqcup I^{(\mathfrak c)}.$
For this fixed \(\Lambda\), define
\begin{equation}\label{eq:def-R-Lambda}
\begin{split}
\mb R_{\Lambda}
:=
\sum_{\substack{
I^{(\mathfrak a)},I^{(\mathfrak b)},I^{(\mathfrak c)}
\subseteq\Lambda\\
|I^{(\mathfrak a)}|=n_{\mathfrak a},\
|I^{(\mathfrak b)}|=n_{\mathfrak b},\
|I^{(\mathfrak c)}|=n_{\mathfrak c}\\
\text{pairwise disjoint}
}}
&
B(\mathfrak a)^{\otimes I^{(\mathfrak a)}}
\otimes
B(\mathfrak b)^{\otimes I^{(\mathfrak b)}}
\otimes
B(\mathfrak c)^{\otimes I^{(\mathfrak c)}}
\otimes
\Delta_3^{\otimes
\left(
\Lambda\setminus
\bigl(
I^{(\mathfrak a)}
\sqcup I^{(\mathfrak b)}
\sqcup I^{(\mathfrak c)}
\bigr)
\right)}.
\end{split}
\end{equation}
Accordingly, the definition of \(\mb X(\mb I)\) can be reorganized as
\begin{equation}\label{eq:case4-X-decomposition}
\begin{split}
\mb X(\mb I)
&=
\left(1-\frac{2\gamma\log n}{n}\right)^{N_{\mb I}}
\sum_{\substack{
I^{(\mathfrak d)},\ldots,I^{(\mathfrak i)}\subseteq[n]\\
|I^{(\mathfrak x)}|=n_{\mathfrak x},
\ \mathfrak x\in\{\mathfrak d,\ldots,\mathfrak i\}\\
\text{pairwise disjoint}
}}
\mb R_{\Lambda}
\otimes
\bigotimes_{\mathfrak x\in
\{\mathfrak d,\ldots,\mathfrak i\}}
B(\mathfrak x)^{\otimes I^{(\mathfrak x)}}.
\end{split}
\end{equation}

Recall from~\cref{obs:B-partition-m3} that
\begin{equation}\label{eq:row-sum-untwirled-polished}
\Delta_3
+
B(\mathfrak a)
+
B(\mathfrak b)
+
B(\mathfrak c)
=
V_1(\pi_1),
\end{equation}
where \(\pi_1\in S_{3,\mathrm{odd}}\) indexes the first row of the
chessboard in~\cref{eq:chessboard-new}.

Expanding $\left(
\Delta_3+
B(\mathfrak a)+
B(\mathfrak b)+
B(\mathfrak c)
\right)^{\otimes\Lambda}$ amounts to summing over all partitions $\Lambda
=
I^{(0)}
\sqcup I^{(\mathfrak a)}
\sqcup I^{(\mathfrak b)}
\sqcup I^{(\mathfrak c)},$ with the corresponding local block assigned to each position set. The operator
\(\mb R_\Lambda\) is obtained by restricting this sum to those partitions
satisfying $|I^{(\mathfrak a)}|=n_{\mathfrak a},
|I^{(\mathfrak b)}|=n_{\mathfrak b},
|I^{(\mathfrak c)}|=n_{\mathfrak c}.$
Consequently, \(\mb R_\Lambda\) is entrywise dominated by the full
expansion:
\begin{equation}\label{eq:first-row-support-inclusion}
\mb R_\Lambda
\subseteq
\left(
\Delta_3+
B(\mathfrak a)+
B(\mathfrak b)+
B(\mathfrak c)
\right)^{\otimes\Lambda}
=
V_1(\pi_1)^{\otimes\Lambda}.
\end{equation}
Since the remaining factors
\(B(\mathfrak x)\),
\(\mathfrak x\in
\{\mathfrak d,\mathfrak e,\mathfrak f,
\mathfrak g,\mathfrak h,\mathfrak i\}\),
are also Boolean-support operators, tensoring preserves the support
inclusion. Hence
\begin{equation}\label{eq:subseteqTensor}
\mb R_{\Lambda}
\otimes
\bigotimes_{\mathfrak x\in
\{\mathfrak d,\mathfrak e,\mathfrak f,\mathfrak g,\mathfrak h,\mathfrak i\}}
B(\mathfrak x)^{\otimes I^{(\mathfrak x)}}
\subseteq
V_1(\pi_1)^{\otimes\Lambda}
\otimes
\bigotimes_{\mathfrak x\in
\{\mathfrak d,\mathfrak e,\mathfrak f,\mathfrak g,\mathfrak h,\mathfrak i\}}
B(\mathfrak x)^{\otimes I^{(\mathfrak x)}} .
\end{equation}
Taking the quadratic form of $\mb{X}(\mb{I})$ with respect to
\(\ket{x}\), and then using
\cref{eq:subseteqUnitx,eq:subseteqTensor}, we obtain
\begin{equation}\label{eq:case4-X-first-bound}
\begin{split}
\bra{x}\mb X(\mb I)\ket{x}
&=\left(1-\frac{2\gamma\log n}{n}\right)^{N_{\mb I}}
\sum_{\substack{
I^{(\mathfrak d)},\dots,I^{(\mathfrak i)}\subseteq[n]\\
|I^{(\mathfrak d)}|=n_{\mathfrak d},\dots,
|I^{(\mathfrak i)}|=n_{\mathfrak i}\\
\text{pairwise disjoint}
}}
\bra{x}\mb R_{\Lambda}
\otimes
\bigotimes_{\mathfrak x\in
\{\mathfrak d,\mathfrak e,\mathfrak f,
\mathfrak g,\mathfrak h,\mathfrak i\}}
B(\mathfrak x)^{\otimes I^{(\mathfrak x)}}\ket{x}
\\
&\le
\left(1-\frac{2\gamma\log n}{n}\right)^{N_{\mb I}}\quad\times
\sum_{\substack{
I^{(\mathfrak d)},\dots,I^{(\mathfrak i)}\subseteq[n]\\
|I^{(\mathfrak d)}|=n_{\mathfrak d},\dots,
|I^{(\mathfrak i)}|=n_{\mathfrak i}\\
\text{pairwise disjoint}
}}
\bra{x}
\left[
V_1(\pi_1)^{\otimes\Lambda}
\otimes
\bigotimes_{\mathfrak x\in
\{\mathfrak d,\mathfrak e,\mathfrak f,\mathfrak g,\mathfrak h,\mathfrak i\}}
B(\mathfrak x)^{\otimes I^{(\mathfrak x)}}
\right]
\ket{x}
\\
&\le
\left(1-\frac{2\gamma\log n}{n}\right)^{N_{\mb I}}\quad\times
\sum_{\substack{
I^{(\mathfrak d)},\dots,I^{(\mathfrak i)}\subseteq[n]\\
|I^{(\mathfrak d)}|=n_{\mathfrak d},\dots,
|I^{(\mathfrak i)}|=n_{\mathfrak i}\\
\text{pairwise disjoint}
}}
\left\|
V_1(\pi_1)^{\otimes\Lambda}
\otimes
\bigotimes_{\mathfrak x\in
\{\mathfrak d,\mathfrak e,\mathfrak f,\mathfrak g,\mathfrak h,\mathfrak i\}}
B(\mathfrak x)^{\otimes I^{(\mathfrak x)}}
\right\|_\infty .
\end{split}
\end{equation}

It remains to bound the terms appearing in the last sum of
\cref{eq:case4-X-first-bound}. The operator
\(V_1(\pi_1)\) is a permutation operator on the single-qubit three-copy space,
and therefore  $\|V_1(\pi_1)\|_\infty=1.$

Moreover, property 5 of~\cref{obs:B-partition-m3} shows that 
$\|B(\mathfrak x)\|_\infty\le 1,$ for $
\mathfrak x\in
\{\mathfrak a,\mathfrak b,\mathfrak c,
\mathfrak d,\mathfrak e,\mathfrak f,
\mathfrak g,\mathfrak h,\mathfrak i\}.$
Therefore, by the multiplicativity of the operator norm under tensor products,
each term in the sum of~\cref{eq:case4-X-first-bound} is bounded by one.
Consequently, the remaining contribution is controlled only by the number of
admissible choices of the six remaining position sets:
\begin{equation}\label{eq:case4-rayleigh-bound-revised}
\begin{split}
\bra{x}\mb X(\mb I)\ket{x}
&\le
\left(1-\frac{2\gamma\log n}{n}\right)^{N_{\mb I}}\quad\times
\sum_{\substack{
I^{(\mathfrak d)},\dots,I^{(\mathfrak i)}\subseteq[n]\\
|I^{(\mathfrak d)}|=n_{\mathfrak d},\dots,
|I^{(\mathfrak i)}|=n_{\mathfrak i}\\
\text{pairwise disjoint}
}}
\left\|
V_1(\pi_1)^{\otimes\Lambda}
\otimes
\bigotimes_{\mathfrak x\in
\{\mathfrak d,\mathfrak e,\mathfrak f,\mathfrak g,\mathfrak h,\mathfrak i\}}
B(\mathfrak x)^{\otimes I^{(\mathfrak x)}}
\right\|_\infty .
\\
&\le
\left(1-\frac{2\gamma\log n}{n}\right)^{N_{\mb I}}\quad\times
\sum_{\substack{
I^{(\mathfrak d)},\dots,I^{(\mathfrak i)}\subseteq[n]\\
|I^{(\mathfrak d)}|=n_{\mathfrak d},\dots,
|I^{(\mathfrak i)}|=n_{\mathfrak i}\\
\text{pairwise disjoint}
}}{1}\\
&\leq \left(1-\frac{2\gamma\log n}{n}\right)^{N_{\mb I}}\quad\times
\binom{n}{n_{\mathfrak d}}
\binom{n-n_{\mathfrak d}}{n_{\mathfrak e}}
\binom{n-n_{\mathfrak d}-n_{\mathfrak e}}{n_{\mathfrak f}}
\cdots
\binom{
n-\sum_{\mathfrak x\in
\{\mathfrak d,\ldots,\mathfrak h\}}
n_{\mathfrak x}
}{
n_{\mathfrak i}
}.
\end{split}
\end{equation}

Combining~\cref{eq:case4-rayleigh-bound-revised} with~\cref{eq:tIleqxI}, we
obtain
\begin{equation}\label{eq:case4-after-row-reduction}
\begin{split}
\|\widehat{\mb T}(\mb I)\|_\infty
&\le
\left(1-\frac{2\gamma\log n}{n}\right)^{N_{\mb I}}\quad\times
\binom{n}{n_{\mathfrak d}}
\binom{n-n_{\mathfrak d}}{n_{\mathfrak e}}
\binom{n-n_{\mathfrak d}-n_{\mathfrak e}}{n_{\mathfrak f}}
\cdots
\binom{
n-\sum_{\mathfrak x\in
\{\mathfrak d,\ldots,\mathfrak h\}}
n_{\mathfrak x}
}{
n_{\mathfrak i}
}\\
&\le
\left(1-\frac{2\gamma\log n}{n}\right)^{N_{\mb I}}\quad\times
\binom{n}{n_{\mathfrak d}}
\binom{n}{n_{\mathfrak e}}
\binom{n}{n_{\mathfrak f}}
\cdots
\binom{
n
}{
n_{\mathfrak i}
}.
\end{split}
\end{equation}
The first-row block combination has removed the three potentially
large multiplicity factors associated with
\(n_{\mathfrak a},n_{\mathfrak b},n_{\mathfrak c}\). It remains to
control the six counts in the other two rows.

If \(m_3=n_{\mathfrak b}\), then
\cref{eq:board-bound-3,eq:board-bound-4} give
\[
n_{\mathfrak e},
n_{\mathfrak h}
\le
m_2,
\qquad
n_{\mathfrak d},
n_{\mathfrak f},
n_{\mathfrak g},
n_{\mathfrak i}
\le
\frac{m_1m_2}{m_3}.
\]
If \(m_3=n_{\mathfrak c}\), then
\cref{eq:board-bound-3,eq:board-bound-4c} instead give
\[
n_{\mathfrak f},
n_{\mathfrak i}
\le
m_2,
\qquad
n_{\mathfrak d},
n_{\mathfrak e},
n_{\mathfrak g},
n_{\mathfrak h}
\le
\frac{m_1m_2}{m_3}.
\]
Thus, in either orientation, two of the six remaining counts are
bounded by \(m_2\), while the other four are bounded by
\(m_1m_2/m_3\). Since Case~4 satisfies
\(m_3>\min\{\beta,\beta'\}n\) and \(m_1\le n\),
\[
\frac{m_1m_2}{m_3}
\le
\frac{m_2}{\min\{\beta,\beta'\}}.
\]
Consequently, in both orientations,
\begin{equation}\label{eq:case4-total-count}
n_{\mathfrak d}
+
n_{\mathfrak e}
+
n_{\mathfrak f}
+
n_{\mathfrak g}
+
n_{\mathfrak h}
+
n_{\mathfrak i}
\le
\left(
2+\frac{4}{\min\{\beta,\beta'\}}
\right)m_2
=:
c_{\beta,\beta'}m_2.
\end{equation}

Using the elementary estimate $\binom{n}{k}\le n^k,$
the remaining multiplicity in~\cref{eq:case4-after-row-reduction} satisfies
\begin{equation}\label{eq:case4-multiplicity-final}
\begin{split}
&
\binom{n}{n_{\mathfrak d}}
\binom{n}{n_{\mathfrak e}}
\cdots
\binom{
n
}{
n_{\mathfrak i}
}\le
n^{
n_{\mathfrak d}
+n_{\mathfrak e}
+n_{\mathfrak f}
+n_{\mathfrak g}
+n_{\mathfrak h}
+n_{\mathfrak i}
}
\le
n^{c_{\beta,\beta'}m_2}.
\end{split}
\end{equation}

On the other hand, since \(m_1\ge\alpha'n\) in Case~4, and by
\cref{eq:m1m2max,eq:def-N-I-m3}, $N_{\mb I}
\ge
m_1m_2
\ge
\alpha'nm_2 .$
Therefore, for sufficiently large \(n\),
\begin{equation}\label{eq:case4-sparse-decay}
\begin{split}
\left(
1-\frac{2\gamma\log n}{n}
\right)^{N_{\mb I}}
&\le
\exp\left(
-\frac{2\gamma\log n}{n}N_{\mb I}
\right)
\le
n^{-2\alpha'\gamma m_2}.
\end{split}
\end{equation}

Substituting
\cref{eq:case4-multiplicity-final,eq:case4-sparse-decay}
into~\cref{eq:case4-after-row-reduction}, we obtain
\begin{equation}\label{eq:case4-pre-final-polished}
\|\widehat{\mb T}(\mb I)\|_\infty
\le
n^{(c_{\beta,\beta'}-2\alpha'\gamma)m_2}.
\end{equation}

Choose \(\gamma\) sufficiently large such that $2\alpha'\gamma>c_{\beta,\beta'}$.
Then the exponent in~\cref{eq:case4-pre-final-polished} is negative. Since
\(N_{\mb I}>0\) implies \(m_2\ge1\), we finally obtain
\begin{equation}\label{eq:case4-final-polished}
\|\widehat{\mb T}(\mb I)\|_\infty
\le
n^{c_{\beta,\beta'}-2\alpha'\gamma}.
\end{equation}
\textit{Discussion the two orientations.}
We finally verify that the preceding case analysis covers both
representative orientations of \(m_3\). Cases~1 and~2 use only the
orientation-independent bounds
\cref{eq:board-bound-1,eq:board-bound-2,eq:board-bound-3} and therefore
apply whether \(m_3=n_{\mathfrak b}\) or
\(m_3=n_{\mathfrak c}\). As shown at the beginning of Case~3, the
\(\mathfrak c\)-orientation is incompatible with that regime, so only
\(m_3=n_{\mathfrak b}\) needs to be considered there. Case~4 treats
the two orientations separately, using \cref{eq:board-bound-4} in the
\(\mathfrak b\)-orientation and \cref{eq:board-bound-4c} in the
\(\mathfrak c\)-orientation, and yields the same final bound in both
cases. Hence no possible orientation of the maximizing entry has been
omitted.

\paragraph{\textbf{Step 3: Summing over all count data.}}

We now combine the four cases.
Cases~1--3 are exponentially small by
\cref{eq:case1-final-polished,eq:case2-final-polished,eq:case3-final-polished}, whereas Case~4 gives the polynomial bound
in~\cref{eq:case4-final-polished}. The estimates in Cases~1--3 depend only
on the fixed thresholds \(\alpha,\alpha',\beta,\beta'\) and the fixed lower
bound \(\gamma>\gamma_0\). Thus, they may be represented by a single
negligible function \(\eta(n)\) that is uniform over all count data
\(\mb I\). Hence, for every \(\mb I\) satisfying \(N_{\mb I}>0\),
\begin{equation}\label{eq:uniform-TI-bound}
\|\widehat{\mb T}(\mb I)\|_\infty
\le
n^{c_{\beta,\beta'}-2\alpha'\gamma}
+
\eta(n).
\end{equation}

The number of possible count data is bounded by $\#\{\mb I\}
\le
\binom{n+9}{9}
=
O(n^9).$
Therefore, summing~\cref{eq:uniform-TI-bound} over all count data
in~\cref{eq:triangle-hermitian} gives
\begin{equation}
\begin{split}
&
\left\|
\mb M_{\mc E_{\mathrm{shallow}}^{(\gamma)}}^{(3)}
-
\mb M_{\mc E_{\mathrm{phase}}}^{(3)}
\right\|_\infty\le
2^{-3n}
\binom{n+9}{9}
\left[
n^{c_{\beta,\beta'}-2\alpha'\gamma}
+
\eta(n)
\right].
\end{split}
\end{equation}
Since a polynomial factor times a negligible function remains negligible,
we may choose $B_0:=2\alpha',
A_0:=c_{\beta,\beta'}+10,$ and choose $\gamma_0>
\max\left\{
1,\frac{c_{\beta,\beta'}}{2\alpha'}
\right\}.$ Then, for every constant \(\gamma>\gamma_0\) and all sufficiently
large \(n\), 
\begin{equation}
\left\|
\mb M_{\mc E_{\mathrm{shallow}}^{(\gamma)}}^{(3)}
-
\mb M_{\mc E_{\mathrm{phase}}}^{(3)}
\right\|_\infty
\le
2^{-3n}
\left[
n^{A_0-B_0\gamma}
+
\mathrm{negl}(n)
\right].
\end{equation}
This proves the proposition.
\end{proof}

For completeness, we make all constants in the above argument explicit.
We choose $\alpha=\frac1{500},
\beta=\beta'=\alpha'=\frac1{50}.$
These constants satisfy $0<\alpha<\alpha'<\frac12,
0<\beta=\beta'<\frac12,$ and the conditions in
\cref{eq:anc_const1,eq:anc_const2,eq:anc_const3}
are all satisfied.

For these choices, we can take
\[
A_0:=212,
\qquad
B_0:=2\alpha'=\frac1{25},
\qquad
\gamma_0:=5300.
\]
Indeed, for every \(\gamma>\gamma_0\), the requirement in Case~4 is satisfied, while $A_0-B_0\gamma
=
212-\frac{\gamma}{25}
<0.$

We emphasize that these numerical choices are not optimized. They are chosen
only for convenience, but they satisfy all conditions required throughout
Cases~1--4 and make \(A_0\), \(B_0\), and \(\gamma_0\) explicit absolute
constants in~\cref{prop:twirled-third-moment-main}.

\section{Proof of~\cref{prop:log-sparsity-optimal-main}}
\label{ap:log-sparsity-necessary}

Our construction uses a sparse phase circuit architecture in which each $CZ$ interaction is included independently with probability $p=\gamma\log n/n$, followed by local Clifford twirling. A natural question is whether this logarithmic connectivity is intrinsic to the architecture or merely a feature of our analysis. We show that it is necessary: within the same phase circuit architecture, any ensemble with bounded second-order relative error must have $p=\Omega(\log n/n)$. 

\begin{theorem}[Necessity of logarithmic sparsity]
\label{thm:log-sparsity-necessary}
Let $0\leq p\leq 1/2$ for all sufficiently large $n$. If
$\mc E_{\mathrm{shallow}}^{(\gamma)}$ has uniformly bounded second-order relative error, then necessarily
\begin{equation}
p=\Omega\!\left(\frac{\log n}{n}\right).
\end{equation}
\end{theorem}

\begin{proof}
We prove the stronger statement that, for any constant $\eta>0$, if
\begin{equation}
p\leq
\left(\frac32-\eta\right)\frac{\log n}{n},
\label{eq:subcritical-edge-probability}
\end{equation} then
$\epsilon_2(\mc E_{\mathrm{shallow}}^{(\gamma)})\to\infty$. In the parametrization
$p=\gamma\log n/n$, condition
\cref{eq:subcritical-edge-probability} corresponds to
$\gamma\leq 3/2-\eta$. To do so, we choose a basis vector
$\ket{\mb z}\in\mc B_2^{\otimes n}$ in the fully symmetric sector,
$I_{\mathrm{sym}}^{(\mb z)}=[n]$. The eigenvalue formulas in
\cref{lem:k2-eigenvalues} give
\begin{equation}
\begin{split}
G_n(p):=
D^2\left[
\lambda_{\mc E_{\mathrm{shallow}}^{(p)}}(\ket{\mb z})
-
\lambda_{\mc E_{\mathrm{phase}}}(\ket{\mb z})
\right]
=
3^{-n}
\sum_{i=0}^{n-2}\binom{n}{i}2^i
\sum_{j=1}^{n-i-1}
\binom{n-i}{j}
(1-2p)^{j(n-i-j)} .
\end{split}
\label{eq:necessary-sparsity-eigenvalue}
\end{equation}
This eigenvalue gap directly lower bounds the relative error. Indeed, by the normalized operator-norm characterization and the reverse triangle inequality,
\begin{align}
\epsilon_2\!\left(\mc E_{\mathrm{shallow}}^{(p)}\right)
&\geq
r_{\mathrm{sym}}^{(2)}
\left|
\lambda_{\mc E_{\mathrm{shallow}}^{(\gamma)}}(\ket{\mb z})
-
\lambda_{\mc E_{\mathrm{Haar}}}(\ket{\mb z})
\right|
\nonumber\\
&\geq
r_{\mathrm{sym}}^{(2)}
\left[\left|
\lambda_{\mc E_{\mathrm{shallow}}^{(\gamma)}}(\ket{\mb z})
-
\lambda_{\mc E_{\mathrm{phase}}}(\ket{\mb z})\right|
-
\left|
\lambda_{\mc E_{\mathrm{phase}}}(\ket{\mb z})
-
\lambda_{\mc E_{\mathrm{Haar}}}(\ket{\mb z})
\right|
\right]
\nonumber\\
&\geq
\frac{D+1}{2D}
\left[
G_n(p)-\left(\frac23\right)^n
\right],
\label{eq:relative-error-from-eigenvalue-gap}
\end{align}
where we used
$r_{\mathrm{sym}}^{(2)}=D(D+1)/2$ and
\cref{prop:k2-phase-haar-norm}. For
$p=\gamma\log n/n$, the ensemble
$\mc E_{\mathrm{shallow}}^{(p)}$ is precisely
$\mc E_{\mathrm{shallow}}^{(\gamma)}$, so
\cref{eq:relative-error-from-eigenvalue-gap} directly bounds
$\epsilon_2(\mc E_{\mathrm{shallow}}^{(\gamma)})$.

All terms in \cref{eq:necessary-sparsity-eigenvalue} are nonnegative because $0\leq p\leq1/2$. Retaining only the terms with $j=1$ gives
\begin{align}
G_n(p)
&\geq
3^{-n}
\sum_{i=0}^{n-2}
\binom{n}{i}2^i(n-i)(1-2p)^{n-i-1}
\nonumber\\
&=
\sum_{k=2}^{n}
\binom{n}{k}
\left(\frac13\right)^k
\left(\frac23\right)^{n-k}
k(1-2p)^{k-1}
\nonumber\\
&=
\frac n3
\left(1-\frac{2p}{3}\right)^{n-1}
-
\frac n3\left(\frac23\right)^{n-1},
\label{eq:necessary-sparsity-j1}
\end{align}
where in the second line we set $k=n-i$, and the last equality follows by differentiating the corresponding binomial expansion.

Under \cref{eq:subcritical-edge-probability}, one has
$p=O(\log n/n)$ and hence $np^2=o(1)$. It follows that
\begin{equation}
\left(1-\frac{2p}{3}\right)^{n-1}
=
\exp\!\left[-\frac23np+o(1)\right]
\geq
n^{-1+\frac{2\eta}{3}+o(1)}.
\label{eq:subcritical-decay}
\end{equation}
Substituting this estimate into
\cref{eq:necessary-sparsity-j1} yields
\begin{equation}
G_n(p)
\geq
\frac13 n^{\frac{2\eta}{3}+o(1)}
-
\frac n3\left(\frac23\right)^{n-1}
\longrightarrow\infty .
\label{eq:eigenvalue-gap-divergence}
\end{equation}
Combining this with
\cref{eq:relative-error-from-eigenvalue-gap}, we obtain
\begin{equation}
\begin{split}
\epsilon_2\!\left(\mc E_{\mathrm{shallow}}^{(p)}\right)
&\geq
\frac{D+1}{2D}
\left[
G_n(p)-\left(\frac23\right)^n
\right]
\geq
\frac{D+1}{2D}
\left[
\frac13 n^{\frac{2\eta}{3}+o(1)}
-
\frac n3\left(\frac23\right)^{n-1}
-
\left(\frac23\right)^n
\right]
\longrightarrow\infty .
\end{split}
\label{eq:necessary-sparsity-divergence}
\end{equation}
Therefore, any edge probability satisfying
\cref{eq:subcritical-edge-probability} along an infinite subsequence leads to unbounded second-order relative error. Consequently, uniformly bounded relative error requires
$p=\Omega(\log n/n)$.

\end{proof}

Thus, within the phase circuit architecture, logarithmic expected degree is necessary for bounded second-order relative error. Together with \cref{thm:second-moment}, this establishes the optimal system-size scaling $p=\Theta\!\left(\frac{\log n}{n}\right),$
equivalently $\Theta(\log n)$ expected degree and
$\Theta(n\log n)$ expected number of $CZ$ gates, up to constant factors.

\end{appendix}
\end{document}